\documentclass[twocolumn]{svjour3}

\usepackage[noend]{algpseudocode}
\usepackage{amsmath,amssymb}
\usepackage{cite}
\usepackage{commands-tam}
\usepackage{float}
\usepackage{graphicx}
\usepackage{hyperref}
\usepackage{multirow}
\usepackage{wrapfig}
\usepackage{subcaption}
\usepackage{xcolor}
\usepackage[title]{appendix}

\newif\iftodo

\newif\ifhighlight
\newcommand{\update}[1]{\ifhighlight\color{blue} #1 \color{black} \else #1 \fi}

\newcommand{\STAMstar}{STAM^*}

\smartqed

\title{Self-Replication via Tile Self-Assembly} 

\titlerunning{Self-Replication}

\author{Andrew Alseth\thanks{A. Alseth, D. Hader, and M.\ J.\ Patitz are funded in part by the National Science Foundation under award CAREER-1553166.} \and
Daniel Hader \and
Matthew J. Patitz}
\authorrunning{A. Alseth, D. Hader, and M.\,J. Patitz}

\institute{Andrew Alseth \and Daniel Hader \and Matthew J. Patitz
\at Department of Computer Science and Computer Engineering, University of Arkansas \\ \email{awalseth@uark.edu,dhader@uark.edu, patitz@uark.edu}}

\journalname{Natural Computing}
\date{Received: date / Accepted: date}

\begin{document}

\maketitle

\begin{abstract}
In this paper we present a model containing modifications to the Signal-passing Tile Assembly Model (STAM), a tile-based self-assembly model whose tiles are capable of activating and deactivating glues based on the binding of other glues. These modifications consist of an extension to 3D, the ability of tiles to form ``flexible'' bonds that allow bound tiles to rotate relative to each other, and allowing tiles of multiple shapes within the same system. We call this new model the STAM*, and we present a series of constructions within it that are capable of self-replicating behavior. Namely, the input seed assemblies to our STAM* systems can encode either ``genomes'' specifying the instructions for building a target shape, or can be copies of the target shape with instructions built in. A universal tile set exists for any target shape (at scale factor 2), and from a genome assembly creates infinite copies of the genome as well as the target shape. An input target structure, on the other hand, can be ``deconstructed'' by the universal tile set to form a genome encoding it, which will then replicate and also initiate the growth of copies of assemblies of the target shape. Since the lengths of the genomes for these constructions are proportional to the number of points in the target shape, we also present a replicator which utilizes hierarchical self-assembly to greatly reduce the size of the genomes required. The main goals of this work are to examine minimal requirements of self-assembling systems capable of self-replicating behavior, with the aim of better understanding self-replication in nature as well as understanding the complexity of mimicking it. 

\keywords{self-assembly \and self-replication \and Tile Assembly Model \and signal-passing tiles}
\end{abstract}

\update{\paragraph{Notes.}
A conference version of this paper was presented at DNA27 in September 2021\cite{SelfReplicationDNA}.
This paper differs substantially from the conference version; in particular, this version includes significantly more details of the constructions and proofs of their correctness, as well as Theorem~\ref{thm:need-to-deconstruct} which demonstrates the necessity of deconstruction during the process of self-replication for a class of shapes.}

\section{Introduction}
\label{sec:intro}

\vspace{-5pt}
\subsection{Background and motivation}
\vspace{-5pt}

Research in tile based self-assembly is typically focused on modeling the computational and shape-building capabilities of biological nano-materials whose dynamics are rich enough to allow for interesting algorithmic behavior. Polymers such as DNA, RNA, and poly-peptide chains are of particular interest because of the complex ways in which they can fold and bind with both themselves and others. Even when only taking advantage of a small subset of the dynamics of these materials, with properties like binding and folding generally being restricted to very manageable cases, tile assembly models have been extremely successful in exhibiting vast arrays of interesting behavior \cite{RotWin00,SolWin05,IUSA,OneTile,2HAMIU,jCCSA,SummersTemp,DotKarMasNegativeJournal,BeckerRR06,AGKS05g,Dot09}. Among other things, a typical question in the realm of algorithmic tile assembly asks what the minimal set of requirements is to achieve some desired property. Such questions can range from very concrete, such as ``how many distinct tile types are necessary to construct specific shapes?'', to more abstract such as ``under what conditions is the construction of self-similar fractal-like structures possible?''. Since the molecules inspiring many tile assembly models are used in nature largely for the purpose of self-replication of living organisms, a natural tile assembly question is thus whether or not such behavior is possible to model algorithmically.

In this paper we show that we can define a model of tile assembly in which the complexities of self-replication type behavior can be captured, and provide constructions in which such behavior occurs. We define our model with the intention of it (1) being hopefully physically implementable in the (near) future, and (2) using as few assumptions and constraints as possible. Our constructions therefore provide insight into understanding the basic rules under which the complex dynamics of life, particularly self-replication, may occur.

We chose to use the Signal-passing Tile Assembly Model (STAM) as a basis for our model, which we call the STAM*, because (1) there has been success in physically realizing such systems \cite{SignalTilesExperimental} and potential exists for further, more complex, implementations using well-established technologies like DNA origami \cite{RotOrigami05,OrigamiTiles,wei2012complex,OrigamiBox,OrigamiSeed} and DNA strand displacement \cite{Qian1196,WangE12182,Simmel2019,ZhangDavidYu2011DDnu,StrandDispInTiles,bui2018localized}, and (2) the STAM allows for behavior such as cooperative tile attachment as well as detachment of subassemblies.
We modify the STAM by bringing it into 3 dimensions and making a few simplifying assumptions, such as allowing multiple tile shapes and tile rotation around flexible glues and removing the restriction that tiles have to remain on a fixed grid.
Allowing flexibility of structures and multiple tile shapes provides powerful new dynamics that can mimic several aspects of biological systems and suffice to allow our constructions to model self-replicating behavior.
Prior work, theoretical \cite{STAMPatternRep} and experimental \cite{SchulYurWinfEvolution}, has focused on the replication of patterns of bits/letters on 2D surfaces, as well as the replication of 2D shapes in a model using staged assembly \cite{RNaseSODA2010}, or in the STAM \cite{STAMshapes}. However, all of these are fundamentally 2D results and our 3D results, while strictly theoretical, are a superset with constructions capable of replicating all finite 2D and 3D patterns and shapes.

Biological self-replication requires three main categories of components: (1) instructions, (2) building blocks, and (3) molecular machinery to read the instructions and combine building blocks in the manner specified by the instructions. We can see the embodiment of these components as follows: (1) DNA/RNA sequences, (2) amino acids, and (3) RNA polymerase, transfer RNA, and ribosomes, among other things. With our intention to study the simplest systems capable of replication, we started by developing what we envisioned to be the simplest model that would provide the necessary dynamics, the STAM*, and then designed modular systems within the STAM* which each demonstrated one or more important behaviors related to replication. Quite interestingly, and unintentionally, our constructions resulted in components with strong similarities to biological counterparts. As our base encoding of the instructions for a target shape, we make use of a linear assembly which has some functional similarity to DNA. Similar to DNA, this structure also is capable of being replicated to form additional copies of the ``genome''. In our main construction, it is necessary for this linear sequence of instructions to be ``transcribed'' into a new assembly which also encodes the instructions but which is also functionally able to facilitate translation of those instructions into the target shape. Since this sequence is also degraded during the growth of the target structure, it shares some similarity with RNA and its role in replication. Our constructions don't have an analog to the molecular machinery of the ribosome, and can therefore ``bootstrap'' with only singleton copies of tiles from our universal set of tiles in solution. However, to balance the fact that we don't need preexisting machinery, our building blocks are more complicated than amino acids, instead being tiles capable of a constant number of signal operations each (turning glues on or off due to the binding of other glues).


\vspace{-10pt}
\subsection{Our results}
\vspace{-5pt}

Beyond the definition of the STAM* as a new model, we present a series of STAM* constructions. They are designed and presented in a modular fashion, and we discuss the ways in which they can be combined to create various (self-)replicating systems.

\vspace{-10pt}
\subsubsection{Genome-based replicator}
\vspace{-5pt}

We first develop an STAM* tileset which functions as a simple self-replicator (in Section \ref{sec:simple-replicator}) that begins from a seed assembly encoding information about a target structure, a.k.a. a \emph{genome}, and grows arbitrarily many copies of the genome and target structure, a.k.a. the \emph{phenotype}. This tileset is universal for all 3D shapes comprised of $1\times 1 \times 1$ cubes when they are inflated to scale factor 2 (i.e. each $1 \times 1 \times 1$ block in the shape is represented by a cube of $2 \times 2 \times 2$ tiles). This construction requires a genome whose length is proportional to the number of cube tiles in the phenotype; for non-trivial shapes the genome is a constant factor longer in order to follow a Hamiltonian path through an arbitrary 3D shape at scale factor 2. This is compared to the Soloveichik and Winfree universal (2D) constructor \cite{SolWin07} where a ``genome'' is optimally shortened, but the scale factor of blocks is much larger.

The process by which this occurs contains analogs to natural systems. We progress from a genome sequence (acting like DNA), which is translated into a messenger sequence (somewhat analogous to RNA), that is modified and consumed in the production of tertiary structures (analogous to proteins). We have a number of helper structures that fuel both the replication of the genome and the translation of the messenger sequence.

    
    
    
    
    
    
    

\vspace{-10pt}
\subsubsection{Deconstructive self-replicator}
\vspace{-5pt}

In Section \ref{sec:deconstruct}, we construct an STAM* tileset that can be used in systems in which an arbitrarily shaped seed structure, or phenotype, is disassembled while simultaneously forming a genome that describes its structure. This genome can then be converted into a linear genome (of the form used for the first construction) to be replicated arbitrarily and can be used to construct a copy of the phenotype. We show that this can be done for any 3D shape at scale factor 2 which is sufficient, and in some cases necessary, to allow for a Hamiltonian path to pass through each point in the shape. This Hamiltonian path, among other information necessary for the disassembly and, later, reassembly processes, is encoded in the glues and signals of the tiles making up the phenotype. We then show how, using simple signal tile dynamics, the phenotype can be disassembled tile by tile to create a genome encoding that same information. Additionally, a reverse process exists so that once the genome has been constructed from a phenotype, a very similar process can be used to reconstruct the phenotype while disassembling the genome.

In sticking with the DNA, RNA, protein analogy, this disassembly process doesn't have a particular biological analog; however, this result is important because it shows that we can make our system robust to starting conditions. That is, we can begin the self-replication process at any stage be it from the linear genome, ``kinky genome'' (the messenger sequence from the first construction), or phenotype. Finally, since this construction requires the phenotype to encode information in its glues and signals, we show that this can be computed efficiently using a polynomial time algorithm given the target shape. This not only shows that the STAM* systems can be described efficiently for any target shape via a single universal tile set, but that results from intractable computations aren't built into our phenotype (i.e. we're not ``cheating'' by doing complex pre-computations that couldn't be done efficiently by a typical computationally universal system). Due to space constraints we only include a result about the necessity for deconstruction in a universal replicator in Section \ref{sec:need-to-deconstruct}.

    
    
    
    
    

\vspace{-12pt}
\subsubsection{Hierarchical assembly-based replicator}
\vspace{-5pt}

For our final construction, in Section \ref{sec:2ham}, our aims were twofold. First, we wanted to compress the genome so that its total length is much shorter than the number of tiles in the target shape. Second, we wanted to more closely mimic the biological process in which individual proteins are constructed via the molecular machinery, and then they are released to engage in a hierarchical self-assembly process in which proteins combine to form larger structures. 

Biological genomes are many orders of magnitude smaller than the organisms which they encode, but for our previous constructions the genomes are essentially equivalent in size to the target structures. Our final construction is presented in a ``simple'' form in which the general scaling approximately results in a genome which is length $n^{\frac{1}{3}}$ for a target structure of size $n$. However, we discuss relatively simple modifications which could, for some target shapes, result in genome sizes of approximately $\log{n}$, and finally we discuss a more complicated extension (which also consumes a large amount of ``fuel'', as opposed to the base constructions which consume almost no fuel) that can achieve asymptotically optimal encoding.


\vspace{-10pt}
\subsubsection{Combinations and permutations of constructions}

Due to length restrictions for this version of the paper, and our desire to present what we found to be the ``simplest'' systems capable of combining to perform self-replication, there are several additions to our results which we only briefly mention. For instance, to make our first construction (in Section \ref{sec:simple-replicator}) into a standalone self-replicator, and one which functions slightly more like biological systems, the input to the system, i.e. the seed assembly, could instead be a copy of the target structure with a genome ``tail'' attached to it. The system could function very similarly to the construction of Section \ref{sec:simple-replicator} but instead of genome replication and structure building being separated, the genome could be replicated and then initiate the growth of a connected messenger structure so that once the target structure is completed, the genome is attached. Thus, the input assembly would be completed replicated, and be a self-replicator more closely mirroring biology where the DNA along with the structure cause the DNA to replicate itself and the structure. Attaching the genome to the structure is a technicality that could satisfy the need to have a single seed assembly type, but clearly it doesn't meaningfully change the behavior.
At the end of Section \ref{sec:2ham} we discuss how that construction could be combined with those from Sections \ref{sec:simple-replicator} and \ref{sec:deconstruct}, as well as further optimized. The main body of the paper contains high-level overviews of the definition of the STAM* as well as of the results. Full technical details for each section can be found in the corresponding section of the technical appendix.


\vspace{-5pt}
\section{Preliminaries}
\label{sec:prelims}
\vspace{-5pt}

In this section we define the notation and models used throughout the paper.

We define a \emph{3D shape} $S \subset \mathcal{Z}^3$ as a connected set of $1 \times 1 \times 1$ cubes (a.k.a. unit cubes) which define an arbitrary polycube, i.e. a shape composed of unit cubes connected face to face where each cube represents a voxel (3-D pixel) of $S$.
For each shape $S$, we assume a canonical translation and rotation of $S$ so that, without loss of generality, we can reference the coordinates of each of its voxels and directions of its surfaces, or faces.
We say a unit cube is \emph{scaled by factor $c$} if it is replaced by a $c \times c \times c$ cube composed of $c^3$ unit cubes. Given an arbitrary 3D shape $S$, we say $S$ is \emph{scaled by factor $c$} if every unit cube of $S$ is scaled by factor $c$ and those scaled cubes are arranged in the shape of $S$. We denote a shape $S$ scaled by factor $c$ as $S^c$.

\vspace{-10pt}
\subsection{Definition of the STAM*}

The 3D Signal-passing Tile Assembly Model* \\(3D-STAM*, or simply STAM*) is a generalization of the STAM \cite{jSignals,jSignals3D,jSTAM-fractals,SignalsReplication} (that is similar to the model in \cite{JonoskaSignals1,JonoskaSignals2}) in which (1) the natural extension from 2D to 3D is made (i.e. tiles become 3-dimensional shapes rather than 2-dimensional squares), (2) multiple tile shapes are allowed, (3) tiles are allowed to flip and rotate \cite{OneTile,jRTAM}, and (4) glues are allowed to be rigid (as in the aTAM, 2HAM, STAM, etc., meaning that when two adjacent tiles bind to each other via a rigid glue, their relative orientations are fixed by that glue) or \emph{flexible} (as in \cite{FTAM}) so that even after being bound tiles and subassemblies are free rotate with respect to tiles and subassemblies to which they are bound by bending or twisting around a ``joint'' in the glue. (This would be analogous to rigid glues forming as DNA strands combine to form helices with no single-stranded gaps, while flexible glues would have one or more unpaired nucleotides leaving a portion of single-stranded DNA joining the two tiles, which would be flexible and rotatable.) See Figure~\ref{fig:glue-example} for a simple example. These extensions make the STAM* a hybrid model of those in previous studies of hierarchical assembly \cite{AGKS05g,DDFIRSS07,j2HAMIU,2HAM-temp1,j2HAMSim}, 3D tile-based self-assembly \cite{CookFuSch11,OptimalShapes3D,BeckerTimeOpt,DDDIU}, systems allowing various non-square/non-cubic tile types \cite{Polyominoes,Polygons,OneTile,GeoTiles,GeoTilesUCNC2019,KariTriHex}, and systems in which tiles can fold and rearrange \cite{FTAM,FlexibleVsRigid,FlexibleCompModel,JonoskaFlexible}.

Due to space constraints, we now provide a high-level overview of several aspects of the STAM* model, and full definitions can be found in Section \ref{sec:prelims-append} of the Technical Appendix.

The basic components of the model are \emph{tiles}. Tiles bind to each other via \emph{glues}. Each glue has a \emph{glue type} that specifies its domain (which is the string label of the glue), integer strength, \emph{flexibility} (a boolean value with {\tt true} meaning \emph{flexible} and {\tt false} meaning \emph{rigid}), and length (representing the length of the physical glue component). A glue is an instance of a glue type and may be in one of three states at any given time, {\tt latent,} {\tt on,} {\tt off}. A pair of adjacent glues are able to bind to each other if they have complementary domains and are both in the {\tt on} state, and do so with strength equal to their shared strength values (which must be the same for all glues with the same label $l$ or the complementary label $l^*$).

A \emph{tile type} is defined by its 3D shape (and although arbitrary rotation and translation in $\mathbb{R}^3$ are allowed, each is assigned a canonical orientation for reference), its set of glues, and its set of \emph{signals}. Its set of glues specify the types. locations, and initial states of its glues. Each signal in its set of signals is a triple $(g_1,g_2,\delta)$ where $g_1$ and $g_2$ specify the \emph{source} and \emph{target} glues (from the set of the tile type's glues) and \\$\delta \in \{\texttt{activate,deactivate}\}$. Such a signal denotes that when glue $g_1$ forms a bond, an action is initiated to turn glue $g_2$ either {\tt on} (if $\delta == $  $\texttt{activate}$) or {\tt off} (otherwise). A \emph{tile} is an instance of a tile type represented by its type, location, rotation, set of glue states (i.e. $\texttt{latent,on}$ or $\texttt{off}$ for each), and set of \emph{signal states}. Each signal can be in one of the  signal states $\{\texttt{pre,firing,post}\}$. A signal which has never been activated (by its source glue forming a bond) is in the {\tt pre} state. A signal which has activated but whose action has not yet completed is in the {\tt firing} state, and if that action has completed it is in the {\tt post} state. Each signal can ``fire'' only one time, and each glue which is the target of one or more signals is only allowed to make the following state transitions:  (1) $\texttt{latent} \rightarrow \texttt{on}$, (2) $\texttt{on} \rightarrow \texttt{off}$, and (3) $\texttt{latent} \rightarrow \texttt{off}$.

We use the terms \emph{assembly} and \emph{supertile}, interchangeably, to refer to the full set of rotations and translations of either a single tile (the base case) or a collection of tiles which are bound together by glues. A supertile is defined by the tiles it contains (which includes their glue and signal states) and the glue bonds between them. A supertile may be flexible (due to the existence of a cut consisting entirely of flexible glues that are co-linear and there being an unobstructed path for one subassembly to rotate relative to the other), and we call each valid positioning of it sets of subassemblies a \emph{configuration} of the supertile. A supertile may also be translated and rotated while in any valid configuration. We call a supertile in a particular configuration, rotation, and translation a \emph{positioned supertile}.

Each supertile induces a \emph{binding graph}, a multigraph whose vertices are tiles, with an edge between two tiles for each glue which is bound between them. The supertile is \emph{$\tau$-stable} if every cut of its binding graph has strength at least $\tau$, where the weight of an edge is the strength of the glue it represents. That is, the supertile is $\tau$-stable if cutting bonds of at least summed strength of $\tau$ is required to separate the supertile into two parts.

For a supertile $\alpha$, we use the notation $|\alpha|$ to represent the number of tiles contained in $\alpha$. The \emph{domain} of a positioned supertile $\alpha$, written $\dom \alpha$, is the union of the points in $\mathbb{R}^3$ contained within the tiles composing $\alpha$. Let $\alpha$ be a positioned supertile. Then, for $\vec{v} \in \mathbb{R}^3$, we define the partial function $\alpha(\vec{v}) = t$ where $t$ is the tile containing $\vec{v}$ if $\vec{v} \in \dom \alpha$, otherwise it is undefined. Given two positioned supertiles, $\alpha$ and $\beta$, we say that they are \emph{equivalent}, and we write $\alpha \approx \beta$, if for all $\vec{v} \in \mathbb{R}^3$ $\alpha(\vec{v})$ and $\beta(\vec{v})$ both either return tiles of the same type, or are undefined. We say they're \emph{equal}, and write $\alpha \equiv \beta$, if for all $\vec{v} \in \mathbb{R}^3$ $\alpha(\vec{v})$ and $\beta(\vec{v})$ either both return tiles of the same type having the same glue and signal states, or are undefined.

\begin{figure}
    \centering
    \includegraphics[width=0.3\textwidth]{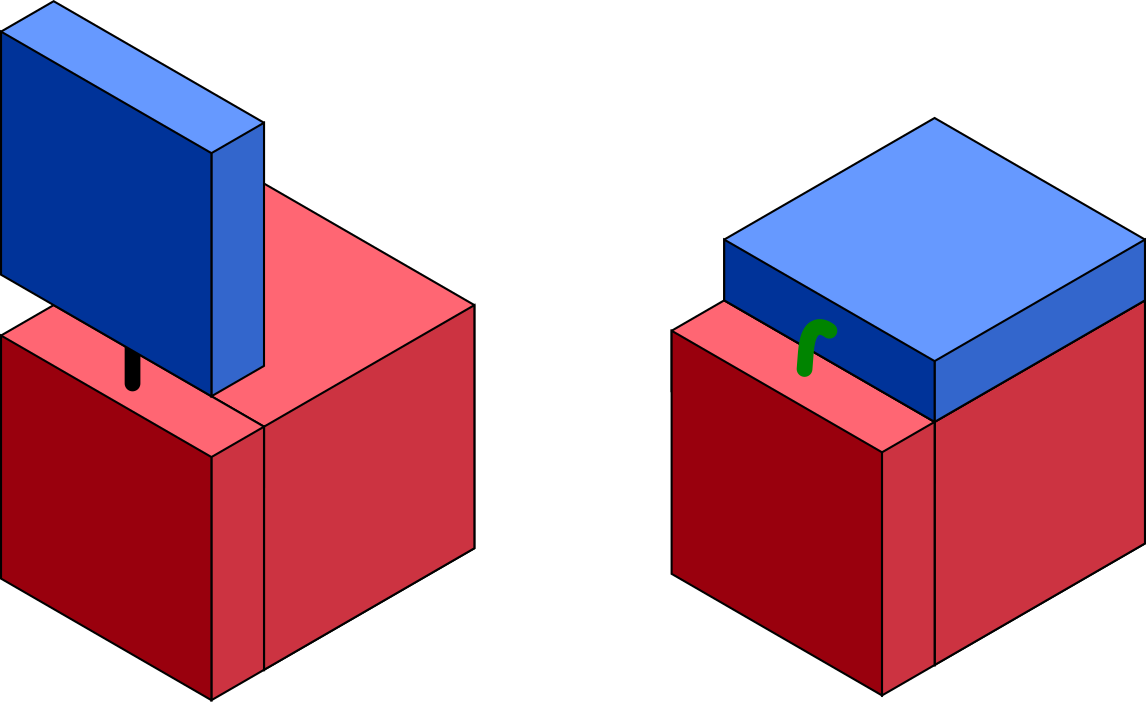}
    \caption{Example showing flat and cubic tiles, and possible behavior of a flexible glue allowing the blue tile to fold upward, away from the red cubic tile, or down against it. In all constructions, we assume lengths for all flexible glues which make the folding and alignment in this figure possible, and length 0 for rigid glues between cubic and flat tiles (as though one tile's glue strand binds into a cavity).\label{fig:glue-example}}
\end{figure}

An STAM* tile assembly system, or TAS, is defined as $\mathcal{T} = (T,C,\tau)$ where $T$ is a finite set of tile types, $C$ is an initial configuration, and $\tau \in \N$ is the minimum binding threshold (a.k.a. temperature) specifying the minimum binding strength that must exist over the sum of binding glues between two supertiles in order for them to attach to each other. The initial configuration $C = \{(S,n) \mid S$ is a supertile over the tiles in $T$ and $n \in \N \cup \infty$ is the number of copies of $S\}$. Note that for each $s \in S$, each tile $\alpha = (t,\vec{l},S,\gamma) \in s$ has a set of glue states $S$ and signal states $\gamma$. By default, it is assumed that every tile in every supertile of an initial configuration begins with all glues in the initial states for its tile type, and with all signal states as $\texttt{pre}$, unless otherwise specified. The initial configuration $C$ of a system $\mathcal{T}$ is often simply given as a set of supertiles, which are also called \emph{seed} supertiles, and it is assumed that there are infinite counts of each seed supertile as well as of all singleton tile types in $T$. If there is only one seed supertile $\sigma$, we will we often just use $\sigma$ rather than $C$.

\vspace{-10pt}
\subsubsection{Overview of STAM* dynamics}

\vspace{-5pt}

An STAM* system $\mathcal{T} = (T,C,\tau)$ evolves nondeterministically in a series of (a possibly infinite number of) steps. Each step consists of randomly executing one of the following actions: (1) selecting two existing supertiles which have configurations allowing them to combine via a set of neighboring glues in the {\tt on} state whose strengths sum to strength $\ge \tau$ and combining them via a random subset of those glues whose strengths sum to $\ge \tau$ (and changing any signals with those glues as sources to the state {\tt firing} if they are in state {\tt pre}), or (2) randomly select two adjacent unbound glues of a supertile which are able to bind, bind them and change attached signals in state {\tt pre} to {\tt firing}, or (3) randomly select a supertile which has a cut $< \tau$ (due to glue deactivations) and cause it to \emph{break} into 2 supertiles along that cut, or (4) randomly select a signal on some tile of some supertile where that signal is in the {\tt firing} state and change that signal's state to {\tt post}, and as long as its action ({\tt activate} or {\tt deactivate}) is currently valid for the signal's target glue, change the target glue's state appropriately.\footnote{The asynchronous nature of signal firing and execution is intended to model a signalling process which can be arbitrarily slow or fast. Please see Section \ref{sec:prelims-append} for more details.} Although at each step the next choice is random, it must be the case that no possible selection is ever ignored infinitely often. (See Section \ref{sec:prelims-append} for more details.)

Given an STAM* TAS $\calT=(T,C,\tau)$, a supertile is \emph{producible}, written as $\alpha \in \prodasm{T}$, if either it is a single tile from $T$, or it is the result of a (possibly infinite) series of combinations of pairs of finite producible assemblies (which have each been positioned so that they do not overlap and can be $\tau$-stably bonded), and/or breaks of producible assemblies. 
A supertile $\alpha$ is \emph{terminal}, written as $\alpha \in \termasm{T}$, if (1) for every $\beta \in \prodasm{T}$, $\alpha$ and $\beta$ cannot be $\tau$-stably attached, (2) there is no configuration of $\alpha$ in which a pair of unbound complementary glues in the {\tt on} state are able to bind, and (3) no signals of any tile in $\alpha$ are in the {\tt firing} state. 

In this paper, we define a shape as a connected subset of $\mathbb{Z}^3$ to both simplify the definition of a shape and to capture the notion that to build an arbitrary shape out of a set of tiles we will actually approximate it by ``pixelating'' it. Therefore, given a shape $S$, we say that assembly $\alpha$ has shape $S$ if $\alpha$ has only one valid configuration (i.e. it is \emph{rigid}) and there exist (1) a rotation of $\alpha$ and (2) a scaling of $S$, $S'$, such that the rotated $\alpha$ and $S'$ can be translated to overlap where there is a one-to-one and onto correspondence between the tiles of $\alpha$ and cubes of $S'$ (i.e. there is exactly $1$ tile of $\alpha$ in each cube of $S'$, and none outside of $S'$).\footnote{In this paper we only consider completely rigid assemblies for target shapes, since the target shapes are static. We could also target ``reconfigurable'' shapes, i.e. sets of shapes, but don't do so in this paper. Also, it could be reasonable to allow multiple tiles in each pixel location as long as the correct overall shape is maintained, but we don't require that.}

\begin{definition}
We say a shape $X$ \emph{self-assembles in $\calT$ with waste size $c$}, for $c \in \mathbb{N}$, if there exists terminal assembly $\alpha\in\termasm{T}$ such that $\alpha$ has shape $X$, and for every $\alpha\in\termasm{T}$, either $\alpha$ has shape $X$, or $|\alpha| \le c$. If $c = 1$, we simply say $X$ \emph{self-assembles in} $\calT$.
\end{definition}

\begin{definition}
We call an STAM* system $\mathcal{R} = (T,C,\tau)$ a \emph{shape self-replicator for shape $S$} if $C$ consists exactly of infinite copies of each tile from $T$ as well as of a single supertile $\sigma$ of shape $S$, there exists $c \in \mathbb{N}$ such that $S$ self-assembles in $\mathcal{R}$ with waste size $c$, and the count of assemblies of shape $S$ increases infinitely.
\end{definition}

\begin{definition}
We call an STAM* system $\mathcal{R} = (T,C,\tau)$ a \emph{self-replicator for $\sigma$ with waste size $c$} if $C$ consists exactly of infinite copies of each tile from $T$ as well as of a single supertile $\sigma$, there exists $c \in \mathbb{N}$ such that for every terminal assembly $\alpha \in \termasm{T}$ either (1) $\alpha \approx \sigma$, or (2) $|\alpha| \le c$, and the count of assemblies $\approx \sigma$ increases infinitely.\footnote{We use $\approx$ rather than $\equiv$ since otherwise either both the seed assemblies and produced assemblies are terminal, meaning nothing can attach to a seed assembly and the system can't evolve, or neither are terminal and it becomes difficult to define the product of a system. However, our construction in Section \ref{sec:deconstruct} can be modified to produce assemblies satisfying either the $\approx$ or $\equiv$ relation with the seed assemblies.} If $c=1$, we simply say $\mathcal{R}$ is a self-replicator for $\sigma$.
\end{definition}



The multiple aspects of STAM* tiles and systems give rise to a variety of metrics with which to characterize and measure the complexity of STAM* systems, beyond metrics seen for models such as the aTAM or even STAM. For a brief discussion, please see the end of Section~\ref{sec:prelims-append}.

\subsubsection{STAM* conventions used in this paper}

\begin{figure}
    \centering
    \includegraphics[width=0.15\textwidth]{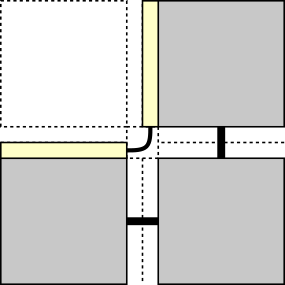}
    \caption{The glue lengths used in our constructions: (1) length $2\epsilon$ rigid bonds between cubic tiles, (2) length 0 rigid bonds between flat and cubic tiles, and (3) length $3\sqrt{2}\;\epsilon/2$ flexible glues between flat tiles.}
    \label{fig:glue-lengths}
\end{figure}

Although the STAM* is a highly generalized model allowing for variety in tile shapes, glue lengths, etc., throughout this paper all constructions are restricted to the following conventions.

\begin{enumerate}

    \item All tile types have one of two shapes (shown in Figure~\ref{fig:glue-example}):
    \begin{enumerate}
        \item A \emph{cubic} tile is a tile whose shape is a $1 \times 1 \times 1$ cube.
    
        \item A \emph{flat} tile is a tile whose shape is a  $1 \times 1 \times \epsilon$ rectangular prism, where $\epsilon < 1$ is a small constant.
        
        \item We call a $1 \times 1$ face of a tile a \emph{full} face, and a $1 \times \epsilon$  face is called a \emph{thin} face.
    \end{enumerate}

    \item Glue lengths are the following (and are shown in Figure \ref{fig:glue-lengths}):
    \begin{enumerate}
        \item All rigid glues between cubic tiles, as well as between thin faces of flat tiles, are length $2\epsilon$.
        
        \item All rigid glues between cubic and flat tiles are length $0$. (Note that this could be implemented via the glue strand of one tile extending into the tile body of the other tile in order to bind, thus allowing the tile surfaces to be adjacent without spacing between the faces.)
       
        \item All flexible glues are length $\frac{3}{2}\sqrt{2}\epsilon$. \footnote{These glue lengths were chosen so that (1) rigidly bound cubic tiles could each have a flat tile bound to each of their sides if needed and (2) so that two flat tiles attached to diagonally adjacent rigid tiles could be attached via flexible glue.}
        
    \end{enumerate}
\end{enumerate}

Given that rigidly bound cubic tiles cannot rotate relative to each other, for convenience we often refer to rigidly bound tiles as though they were on a fixed lattice. This is easily done by first choosing a rigidly bound cubic tile as our origin, then using the location $\vec{l}$, orientation matrix $R$, and rigid glue length $g$, put in one-to-one correspondence with each vector $\vec{v}$ in $\mathbb{Z}^3$, the vector $\vec{l} + g R \vec{v}$. Once we define an absolute coordinate system in this way, we refer to the directions in 3-dimensional space as North ($+y$), East ($+x$), South ($-y$), West ($-x$), Up ($+z$), and Down ($-z$), abbreviating them as $N,E,S,W,U,$ and $D$, respectively.

\update{
    \vspace{-10pt}
    \subsection{Detailed STAM* dynamics}
    \label{sec:prelims-append}
    \begin{enumerate}
    
        \item The binding of a glue causes any signals associated with that glue to change states, i.e. fire (if they haven't already fired due to a prior binding event).
        
        \item A glue and its complementary pair which are bound overlap, causing the distance between their tiles to be the length of the glue (not two times the length).
        
        \item The binding of a single rigid glue or two flexible glues on different surfaces lock a tile in place. Two flexible glues on the same surface prevent ``flipping'' (or ``twisting'') but allow ``hinge-like'' rotation.
        
        \item The assembly process proceeds step by step by nondeterministically selecting one of the following types of moves to execute unless and until none is available. While the following set of choices for a next step are made randomly, no action which is valid can be postponed infinitely long.
        \begin{enumerate}
            \item Randomly select any pair of supertiles, $\alpha$ and $\beta$, which can bind via a sum of $\ge \tau$ strength bonds if appropriately positioned (and binding only via glues in the $\texttt{on}$ state). Position $\alpha$ and $\beta$ to combine them to form a new supertile by binding a random subset of the glues which can bind between them whose strengths sum to $\ge \tau$. For each bound glue which has a signal associated with it, but that signal is still in the \texttt{pre} state, change the signal's state to \texttt{firing}. Note that rigid glues must form bonds which extend perpendicularly from their surfaces, but flexible glues are free to bend to form bonds.
            
            \item Randomly select any supertile which has a cut in its binding graph $< \tau$ (due to one or more glue deactivations), and split that supertile into two supertiles along that cut. We call this operation a \emph{break}.
            
            \item Randomly select any pair of subassemblies (each of one or more tiles) in the same supertile but bound only by flexible glues so that the subassemblies are free to rotate relative to each other, and perform a valid rotation of one of those subassemblies.
            
            \item Randomly select a supertile and pair of unbound glues within it such that the supertile has a valid configuration in which those glues are able to bind (i.e. they are complementary, both in the $\texttt{on}$ state, and the glues can reach each other), and bind them. For each which has a signal associated with it, but that signal is still in the \texttt{pre} state, change the signal's state to \texttt{firing}.
            
            \item Randomly select a signal whose state is \texttt{firing} from any tile and execute it. This entails, based on the signal's definition, that its target glue is either \texttt{activate}d or \texttt{deactivate}d if that is still a valid transition for that glue, and for the signal's state to change to \texttt{post}, marking it as completed and unable to fire again. The STAM* is based on the STAM and it preserves the design goal of modeling physical mechanisms that implement the signals on tiles but which are arbitrarily slower or faster than the average rates of (super)tile attachments and detachments.  Therefore, rather than immediately enacting the actions of signals, each signal is put into a state of \texttt{firing} along with all signals initiated by the glue (since it is technically possible for more than one signal to have been initiated, but not yet enacted, for a particular glue). Any \texttt{firing} signal can be randomly selected from the set, regardless of the order of arrival in the set, and the ordering of either selecting some signal from the set or the combination of two supertiles is also completely arbitrary.  This provides fully asynchronous timing between the initiation, or firing, of signals and their execution (i.e. the changing of the state of the target glue), as an arbitrary number of supertile binding (or breaking) events may occur before any signal is executed from the \texttt{firing} set, and vice versa.
        \end{enumerate}
    \end{enumerate}
    
    The multiple aspects of STAM* tiles and systems give rise to a variety of metrics with which to characterize and measure the complexity of STAM* systems. Following is a list of some such metrics.
    \begin{enumerate}
        \item Tile complexity: the number of unique tile types
        
        \item Tile shape complexity: the number of unique tile shapes, or the maximum number of surfaces on a tile shape, or the maximum difference in sizes between tile shapes
        
        \item Tile glue complexity: the maximum number of glues on any tile type
        
        \item Seed complexity: the size of the seed assembly (and/or the number of unique seed assemblies.
        
        \item Signal complexity: the maximum number of signals on any tile type
        
        \item Junk complexity: the size of the largest terminal assembly which is not considered the ``target assembly'' (a.k.a. \emph{junk assembly}), or the number of unique types of junk assemblies
    \end{enumerate}
}

\section{A Genome Based Replicator}
\label{sec:simple-replicator}

We now present our first construction in the STAM*, in which a ``universal'' set of tiles will cause a pre-formed seed assembly encoding a Hamiltonian path through a target structure, which we call the \emph{genome}, to replicate infinitely many copies of itself as well as build infinitely many copies of the target structure at temperature 2. We consider 4 unique structures which are generated/utilized as part of the self-replication process: $\sigma,\mu,\mu^\prime$, and $\pi$. The seed assembly, $\sigma$, is composed of a connected set of flat tiles considered to be the \emph{genome}. Let $\pi$ represent an assembly of the target shape encoded by $\sigma$. $\mu$ is an intermediate ``messenger'' structure directly copied from $\sigma$, which is modified into $\mu^\prime$ to assemble $\pi$. We split $T$ into subsets of tiles, $T = \{ T_{\sigma} \cup T_{\mu} \cup T_{\varphi} \cup T_{\pi}\}$. $T_\sigma$ are the tiles used to replicate the genome, $T_\mu$ are the tiles used to create the messenger structure, $T_\pi$ are the cubic tiles which comprise the phenotype $\pi$, and $T_\phi$ are the set of tiles which combine to make fuel structures used in both the genome replication process and conversion of $\mu$ to $\mu^\prime$. 

The tile types which make up this replicator are carefully designed to prevent spurious structures and enforce two key properties for the self-replication process. First, a genome is never consumed during replication, allowing for exponential growth in the number of completed genome copies. Second, the replication process from messenger to phenotype strictly follows $\mu \rightarrow \mu^\prime \rightarrow \pi$; each step in the assembly process occurs only after the prior structure is in its completed form. This prevents unexpected geometric hindrances which could block progression of any further step. Complete details of $T$ are located in Section \ref{sec:tiles}. 

\vspace{-10pt}
\subsection{Replication of the genome}
\label{sec:detail-genome}

The minimal requirements to generate copies of $\sigma$ in $\mathcal{R}$ are the following: (1) for all individual tile types $s\in \sigma, s \in T_\sigma$, (2) the last tile is the end tile $E$, and (3) the first tile in $\sigma$ is a start tile in the set $(S^+,S^-)$.  However, for the shape-self replication of $S$ one additional property must hold: (4) $\sigma$ encodes a Hamiltonian path which ends on an exterior cubic tile. We define the genome to be `read' from left to right; given requirements (2) and (3), the leftmost tile in a genome is a start tile and the rightmost is an end tile. (4) can be guaranteed by scaling $S$ up to $S^2$ and utilizing the algorithm in Section \ref{sec:hampath}, selecting a cubic tile on the exterior as a start for the Hamiltonian path and then reversing the result. This requirement ensures the possibility of cubic tile diffusion into necessary locations at all stages of assembly.

\begin{figure}[ht]
    \centering
    \begin{subfigure}{0.46\textwidth}
        \centering
        \includegraphics[height=4cm]{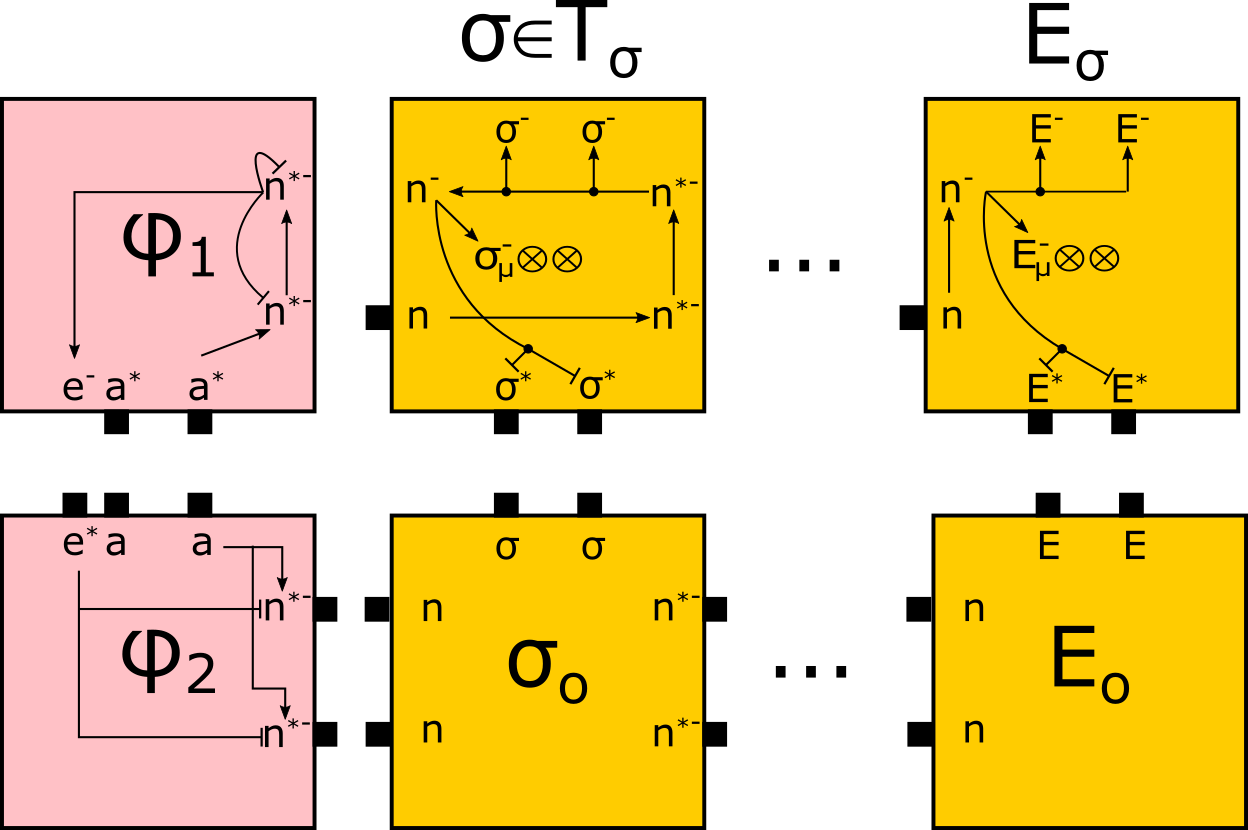}
    \caption{\label{fig:genome_tiles}}
    \end{subfigure}
    \hspace{0.04\textwidth}
    \begin{subfigure}{0.46\textwidth}
        \centering
        \includegraphics[width=0.4\linewidth]{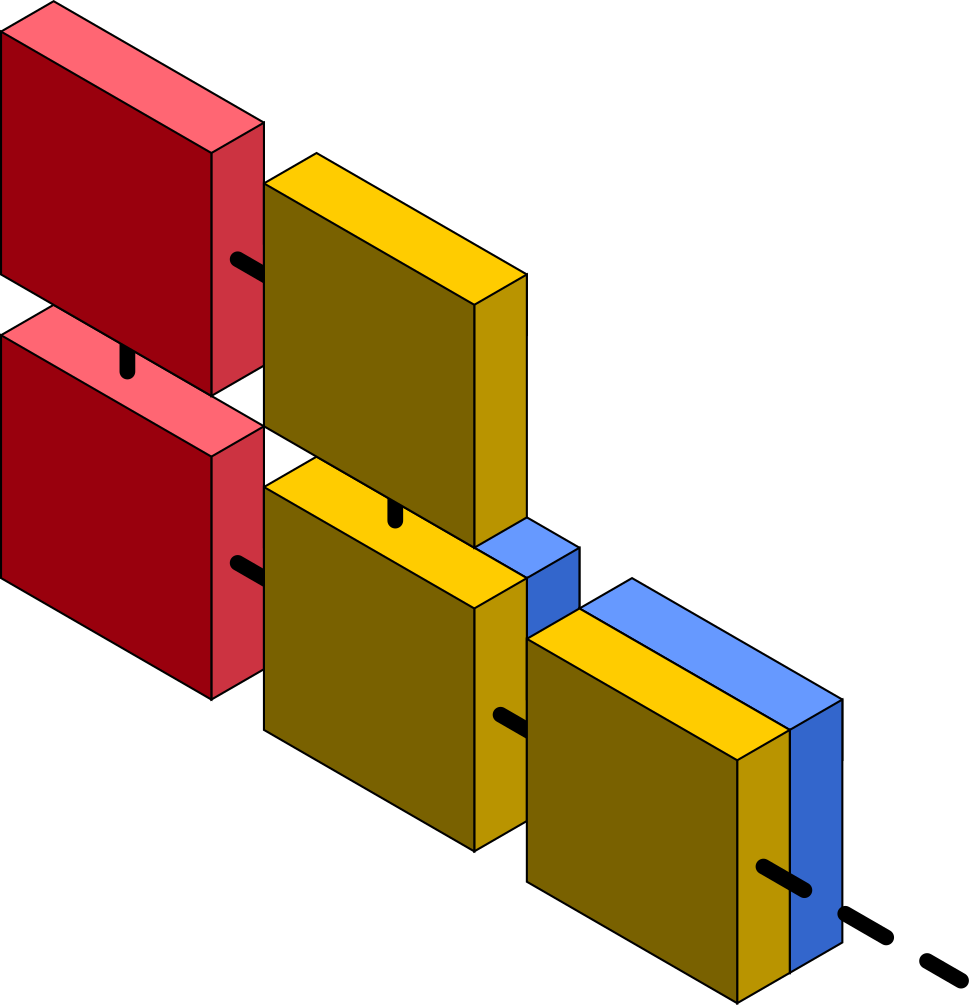}
    \caption{\label{fig:SR-copy-while-translating}}
    
    \end{subfigure}
    \caption{
    (a) Initial genome replicator tiles. Note that $\otimes \otimes$ represents a two strength 1 glues which are on the full face of the seed tiles opposite from the reader (b) 
    Illustration of an arbitrary translation process occurring at the same time as genome replication. Red tiles are representative of $\varphi$, gold tiles are representative of $\sigma$, and blue tiles are representative of $\mu$.}
    \label{fig:SR-replication}
    
    \vspace{-10px}
\end{figure}


\update{
    The replication process of $\sigma$ begins with the attachment of tiles from the set $T_{\sigma}$ to $\sigma$ due to the two strength-1 glues on the north face of individual tiles comprising $\sigma$.
    We denote the incomplete copy of $\sigma$ as $\sigma^\prime$.
    Asynchronously, a fuel tile assembly $\varphi$ comprised of two subtiles $\varphi_1, \varphi_2 \in T_\phi$ binds to the leftmost tile of $\sigma$.
    Upon the binding of a start tile to the north thin face of the start tile of $\sigma^\prime$, the signal provided by $\varphi$ begins a chain reaction binding to the the active `n' glue on the west thin face of the newly attached tile and the signal propagates through the chain of connected $\sigma^\prime$ tiles. Once the end tile $E_\sigma$ is bound to the remainder of $\sigma^\prime$ by the active `n' glue, it returns a signal through its newly activated west glue to fully connect it to the prior tile and then detach from the genome to the south. This signal cascades back through the remaining tiles of $\sigma^\prime$ until reaching $\varphi$, at which point $\varphi$ deactivates its glues. allowing the newly replicated copy of $\sigma$ to separate and begin the process of replicating itself and translating copies of $\mu$.
}
\begin{figure}[ht]
    \centering
    \begin{subfigure}[b]{0.45\textwidth}
    \includegraphics[height=3.0cm]{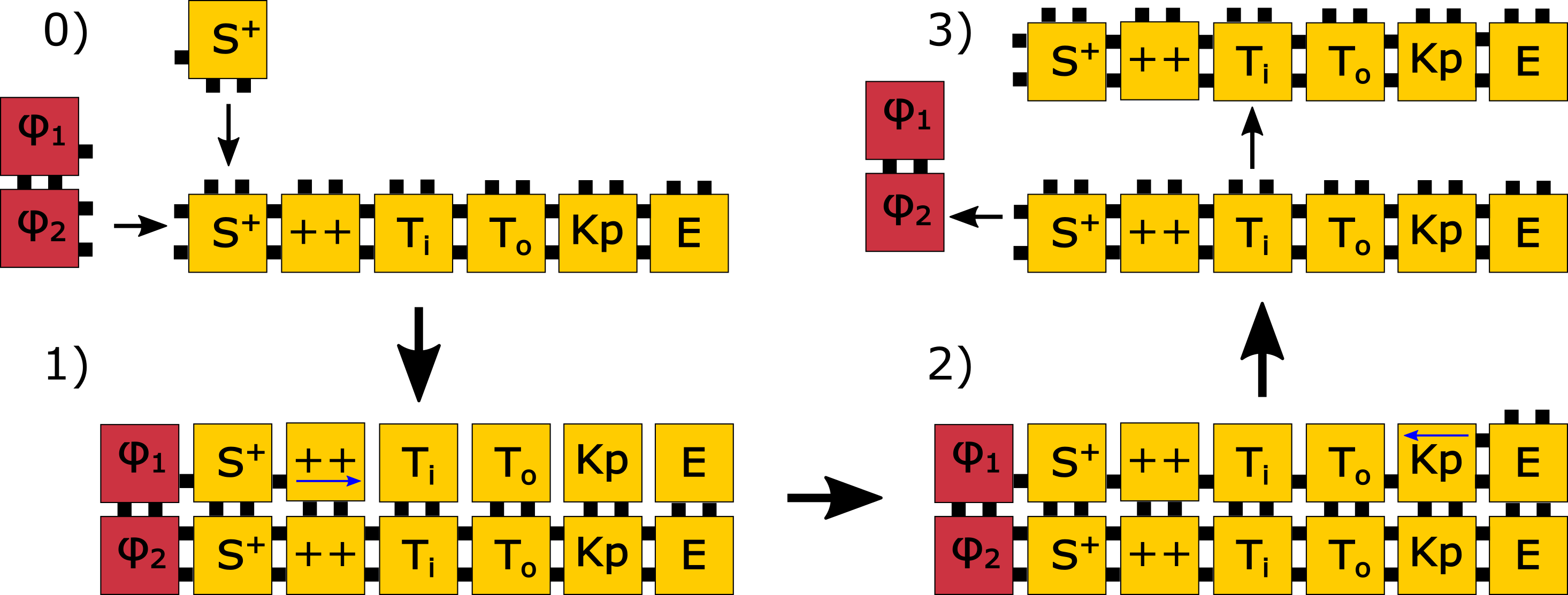}
    \caption{\label{fig:seed-copy-steps}}
    \end{subfigure}
    \begin{subfigure}[b]{0.32\textwidth}
        \centering
        \includegraphics[height=4.2cm]{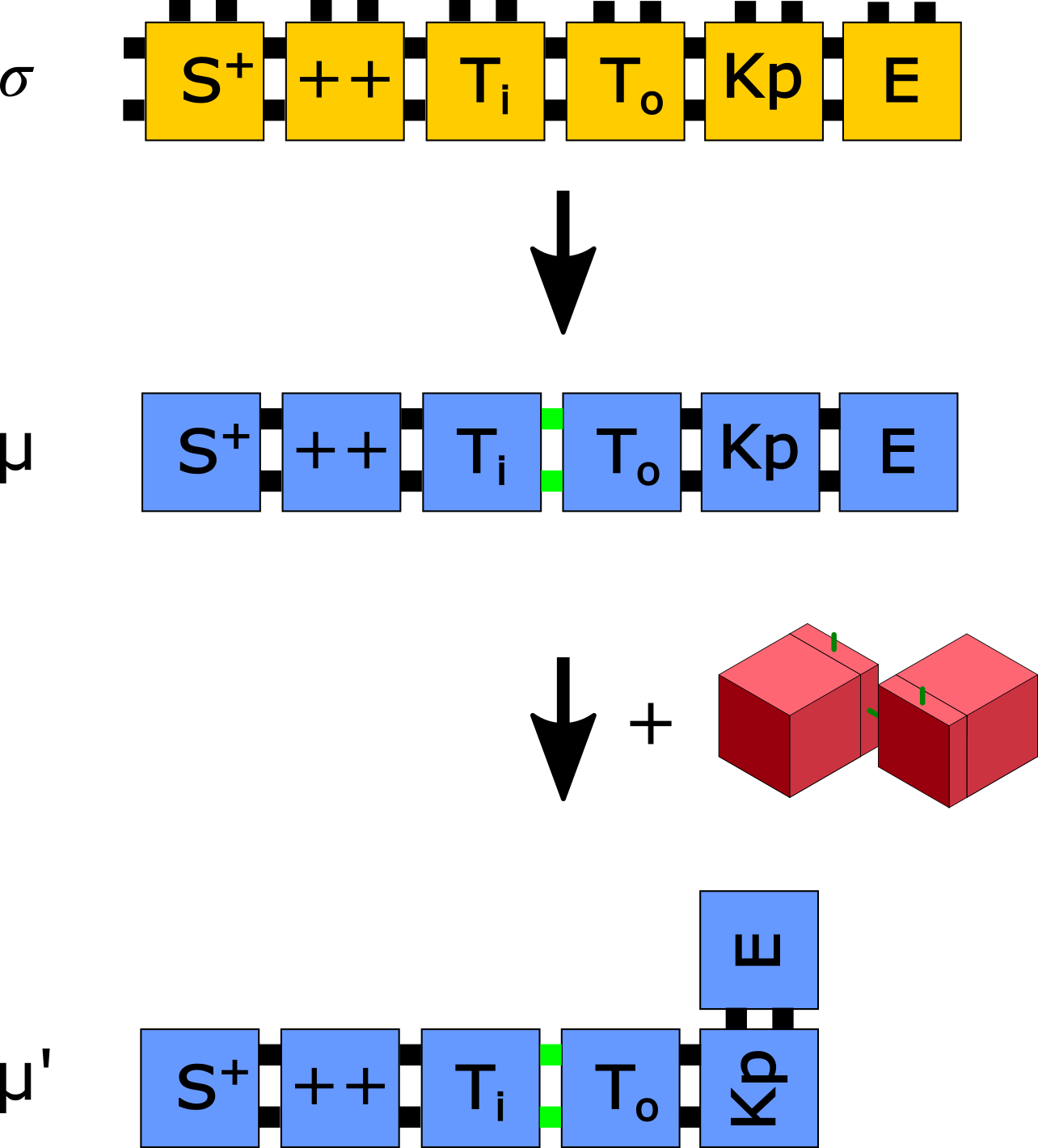}
    \caption{\label{fig:SR_translation}}

    \end{subfigure}
    \caption{(a) In step 0 (before replication begins) both fuel and tiles from $T_\sigma$ bind to $\sigma$. Step 1 indicates the fuel tile binding with the leftmost $S^+$ tile in $\sigma^\prime$, propagating the binding of tiles from west to east indicated by blue arrow on the ++ tile. Step 2 begins after all $\sigma^\prime$ glues are bound by strength-1, leading to the propagation of a second glue binding $\sigma^\prime$ from east to west. Additionally, glues on the north face of $\sigma^\prime$ tiles are activated and glues on the south face binding to $\sigma$ are deactivated once they have a strength-2 connection to. Step 3 demonstrates the detachment - once the second glue binds to the fuel duple ($\varphi_1, \varphi_2$) signals propagate to detach from $\sigma$ and $\sigma^\prime$. 
    (b) Process of translation: the information encoded in $\sigma$ is copied to $\mu$ by a mapping of tiles via glue domains. Green glues on $\mu$ and $\mu^\prime$ are flexible. One kink-ase (red) is used to convert $\mu$ to $\mu^\prime$}
    \vspace{-10px}
\end{figure}

\vspace{-10pt}
\subsection{Translation of $\sigma$ to $\mu$}

\emph{Translation} is defined as the process by which the Hamiltonian path encoded in $\sigma$ is built into a new messenger assembly $\mu$. Since the signals to attach and detach $\mu$ from $\sigma$ are fully contained in the tiles of $T_{\mu}$, translation continues as long as $T_{\mu}$ tiles remain in the system. We note that the translation process can occur at the same time as $\sigma$ is replicating. This causes no unwanted geometric hindrances as demonstrated in Figure \ref{fig:SR-copy-while-translating}.

\vspace{-10pt}
\subsubsection{Placement of $\mu$ tiles}

Messenger tiles from the set $T_\mu$ attach to $\sigma$ as soon as complementary glues on the back flat face of $\sigma$ are activated after the binding of $\varphi$ to $\sigma^\prime$. The process of building $\mu$ does not require a fuel structure to continue, as the messenger tiles have built-in signals to deactivate the glues on $\mu$ which attach $\mu$ to $\sigma$. This allows for a genome to replicate the messenger structure without itself being consumed in any manner.

Each genome tile contains two active strength-1 glues on its full face which are mapped to a single messenger tile type. Messenger tiles from the set $T_\mu$ attach to $\sigma$ as soon as complementary glues on the back flat face of $\sigma$ are activated after the binding of the fuel duple $\varphi$ to $\sigma^\prime$. The process of building $\mu$ does not require a fuel structure to continue, as the messenger tiles have built-in signals to deactivate the glues on $\mu$ which attach $\mu$ to $\sigma$. This allows for a genome to replicate the messenger structure without itself being consumed in any manner. Once a flat tile in $\mu$ is bound to its eastern neighbor, signals are fired from the eastern glues to deactivate the glue connecting $\mu$ to $\sigma$. This leaves $\mu$ as its own separate assembly when every tile has attached to its neighbor(s). The example of translation shown in Figure \ref{fig:SR_translation} illustrates that the same information (i.e., sequence of tiles representing a Hamiltonian path) remains encoded in $\mu$, but allows for new structural functionality that would otherwise not be possible by $\sigma$.






\vspace{-10pt}
\subsubsection{Modification of $\mu$ to $\mu^\prime$}\label{sec:details-kinkase}

The current shape of $\mu$ is such that it could only replicate a trivial 2D structure; $\mu$ must be modified to follow a Hamiltonian path in 3 dimensions as made possible by a set of \emph{turning} tiles. Additionally, in the current state of $\mu$ no cubic tiles can be placed as all the glues which are complementary to cubic tiles are currently in the \texttt{latent} state. Once a glue of type `p' is bound on the start tile, we then consider $\mu$ to have completed its modification into $\mu^\prime$. The `p' glue on turning tiles can only be bound once they have been turned, and as such the turning tiles present in $\mu^\prime$ must be turned before assembly of $\pi$ begins.

Turning tiles modify the shape of $\mu$ by adding `kinks' into the otherwise linear structure by the use of a fuel-like structure called a \emph{kink-ase}. The kink-ase structure is generated from a set of 2 flat tiles and 2 cube tiles. These tiles must first fully bind to each other before connections can be made to a turning tile. The unique form of kink-ase allows for the orientation of two adjacent tiles to be modified without separating $\mu$, shown in Figure \ref{fig:linear-to-kin}. The turning tiles are physically rotated such that the connection between a turning tile and its predecessor along the west thin edge of the turning tile is broken, and then reattached along either the up or down thin edge of the turning tile. Each turning tile requires the use of a single kink-ase, which turns into a junk assembly. 

\begin{figure}[ht]
    \centering
    \includegraphics[width=\linewidth]{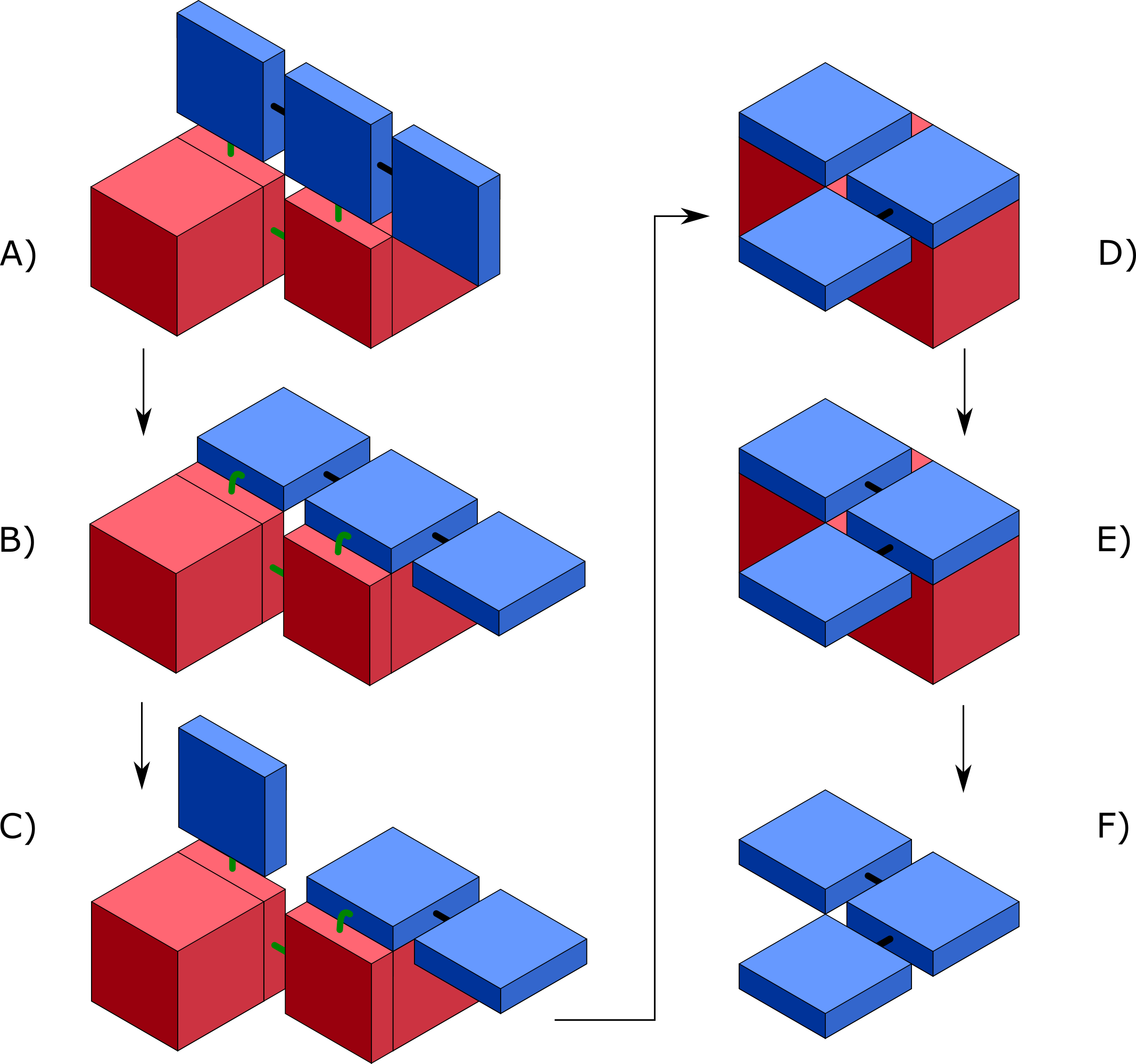}
 
    \caption{Conversion of one turning tile. Blue tiles indicate $\mu$, whereas the red indicate the kink-ase.
    }
    \label{fig:linear-to-kin}
    \vspace{-10px}
\end{figure}

\update{
    We now describe in detail how $\mu$ is converted to $\mu^\prime$ utilizing the kink-ase structure, with the steps in this section matching up with the intermediate structures shown in Figure~\ref{fig:linear-to-kin}.
    \begin{enumerate}
        \item[A)] kink-ase attaches to a turning tile and the predecessor which will be re-oriented in $\mu$. Simultaneously, glues are activated on the kink-ase cube structure attached to the turning tile to bind the turning tile face and to the kink-ase cube structure attached to the predecessor tile to enable the folding of the cube structure in step D). Note - glues connecting tiles in $\mu$ may be either rigid or flexible depending upon the Hamiltonian path generated for $\pi$. This does not effect any intermediate steps presented.  
        \item[B)] The turning tile's rear face binds to the kink-ase due to random movement allowed by the flexible glues which attach the kink-ase to the turning and predecessor tiles, i.e. the flexible bond allows the tile to rotate and randomly assume various relative positions. When it enters the correct configuration, the glues bind to ``lock it in''.
        \item[C)] Upon connection of turning tile face to kink-ase cube, a signal deactivates the rigid glue attaching the predecessor tile to the turning tile. A signal activates glues on the exposed face of the kink-ase tile attached to cube and turning tile structure. The flexible connection between the predecessor tile and kink-ase ensures $\mu$ does not split into two pieces.
        \item[D)] Kink-ase cube and kink-ase tile with activated glue bind on faces when they rotate into the correct configuration, bringing the turning tile into correct geometry with the predecessor tile. The kink-ase cube face adjacent to the predecessor tile activates its glue, allowing for binding with the face of the two. The flexible glue allows for random movement for the complementary glues to attach and bind. Concurrently, the flexible glue on the turning tile is deactivated and a rigid glue of similar type to the turning tile glue deactivated in step C) is activated. 
        \item[E)] A rigid glue between the turning tile and predecessor tile binds, leading to re-connection between both prior detached portions of $\mu$. Activation of the final glue leads to the turning tile signaling to kink-ase to detatch from $\mu$. 
        \item[F)] This structure represents $\mu$ after one turning tile has been resolved. A completion signal is passed through glues attaching the turning tile and predecessor tile. This process continues for all turning tiles serially, working backwards from the termination tile. This is to prevent any interference between structures incurred by multiple adjacent turning tiles.
    \end{enumerate}
}
\vspace{-15pt}
\begin{figure*}
    \centering
    \begin{subfigure}{0.6\textwidth}
        \includegraphics[width=1.0\textwidth]{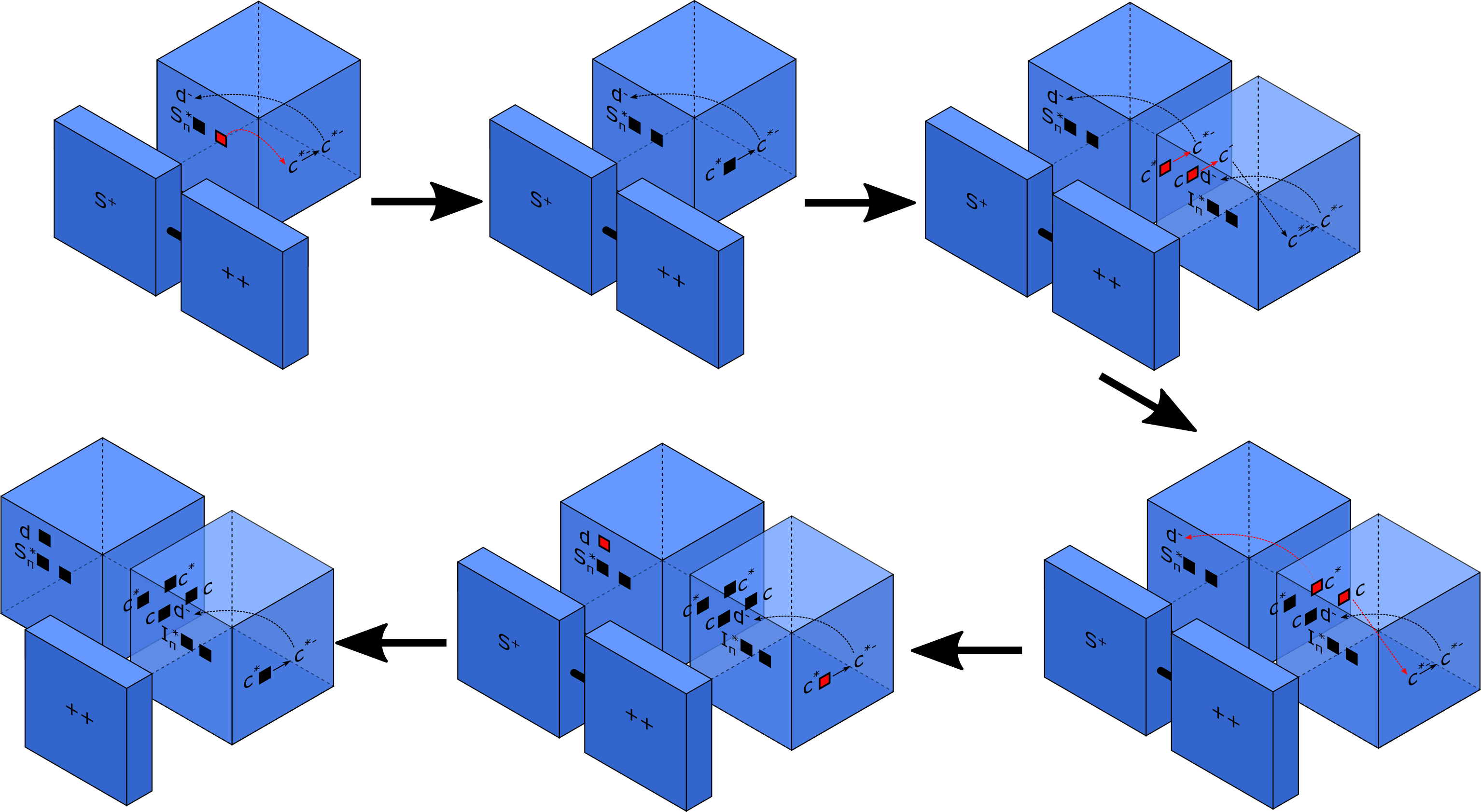}
        \caption{ \label{fig:assembling-pi}}
    \end{subfigure}
    \hspace{0.04\textwidth}
    \begin{subfigure}{0.2\textwidth}
        \includegraphics[width=1.0\textwidth]{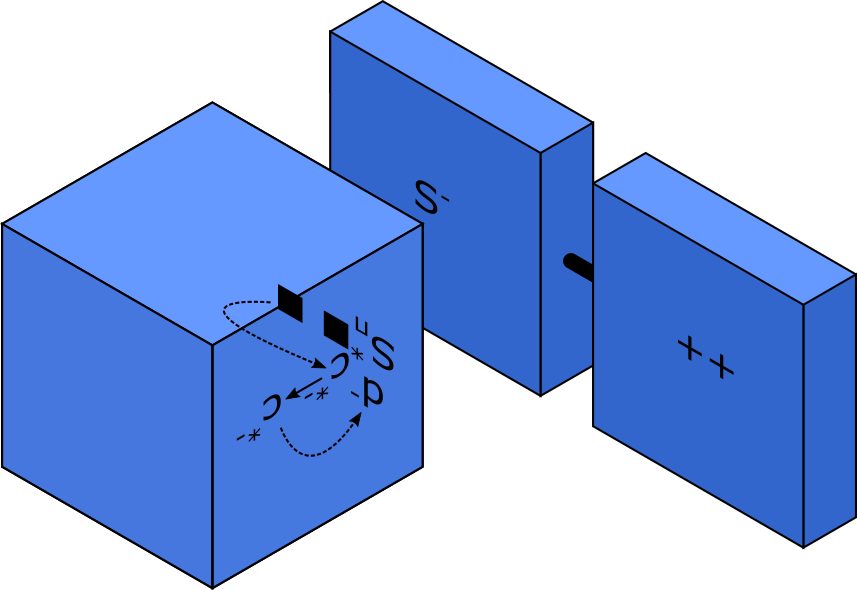}
        \caption{\label{fig:minustile}}
    \end{subfigure}
    \caption{(a) The process of assembling $\pi$ from $\mu^\prime$. Spaces between flat tiles and cubic tiles are exaggerated for illustrative purposes. The red arrows and squares demonstrate the signals propagating through adjacent tiles to solidify connections between two successive cubic tiles in the Hamiltonian path of a phenotype. At the final step of this figure, The cubic tile associated with the ``++'' tile of $\mu^\prime$ has its $c^*$ glue in the \texttt{on} state (b) Demonstration of cube tile placement utilizing $S^-$ genome tile in place of $S^+$ - this causes the glues on the flat tile to place the cube tile on the opposite full face of $\mu^\prime$}
    \label{fig:piconstruction}
\end{figure*}

 
  

\vspace{-10pt}
\subsection{Assembly of $\pi$}

At the end of translation, two strength-1 glues complementary to tiles in $T_\pi$ are active on all tiles of $\mu^\prime$. The only cubic tile which starts with two complementary glues \texttt{on} is the start cubic tile. Once this cubic tile is bound to the start tile, a strength-1 glue of type `c' is activated on the cube. This glue allows for the cooperative binding of the next cubic tile in the Hamiltonian path to the superstructure of both $\mu^\prime$ and the first tile of $\pi$. 

After this process continues and a cubic tile is bound to both its neighbors (or just one neighbor in the case of the start and end tiles) with strength 2, a `d' glue is activated on the face of the cubic tile bound to $\mu^\prime$. This indicates to the flat tile of $\mu^\prime$ that the cube tile is fully connected to its neighbors with strength 2. To prevent any hindrances to the placement of any cubic tiles in $\pi$, the flat tile jettisons itself from the remaining tiles of $\mu^\prime$ by deactivating all active glues and becoming a junk tile\footnote{Due to the asynchronous nature of signals, there may be instances which the addition of cubic tiles of $\pi$ are temporarily blocked. These will be eventually resolved, allowing assembly to continue.}. This process is repeated, adding cube by cube until the end tile in $\mu^\prime$ is reached. Once the end cube has been added to $\pi$, it has shape $S^2$ and $\mu^\prime$ has been disassembled into junk tiles. An example process is shown in Figure \ref{fig:SR_translation}, with a detailed step-by-step visualization of glue activation shown in Figure~\ref{fig:piconstruction}. Additionally, Figure~\ref{fig:minustile} demonstrates how tiles are placed on the opposite side of the genome.

\begin{figure}[ht]
    \centering
    \includegraphics[width=\linewidth]{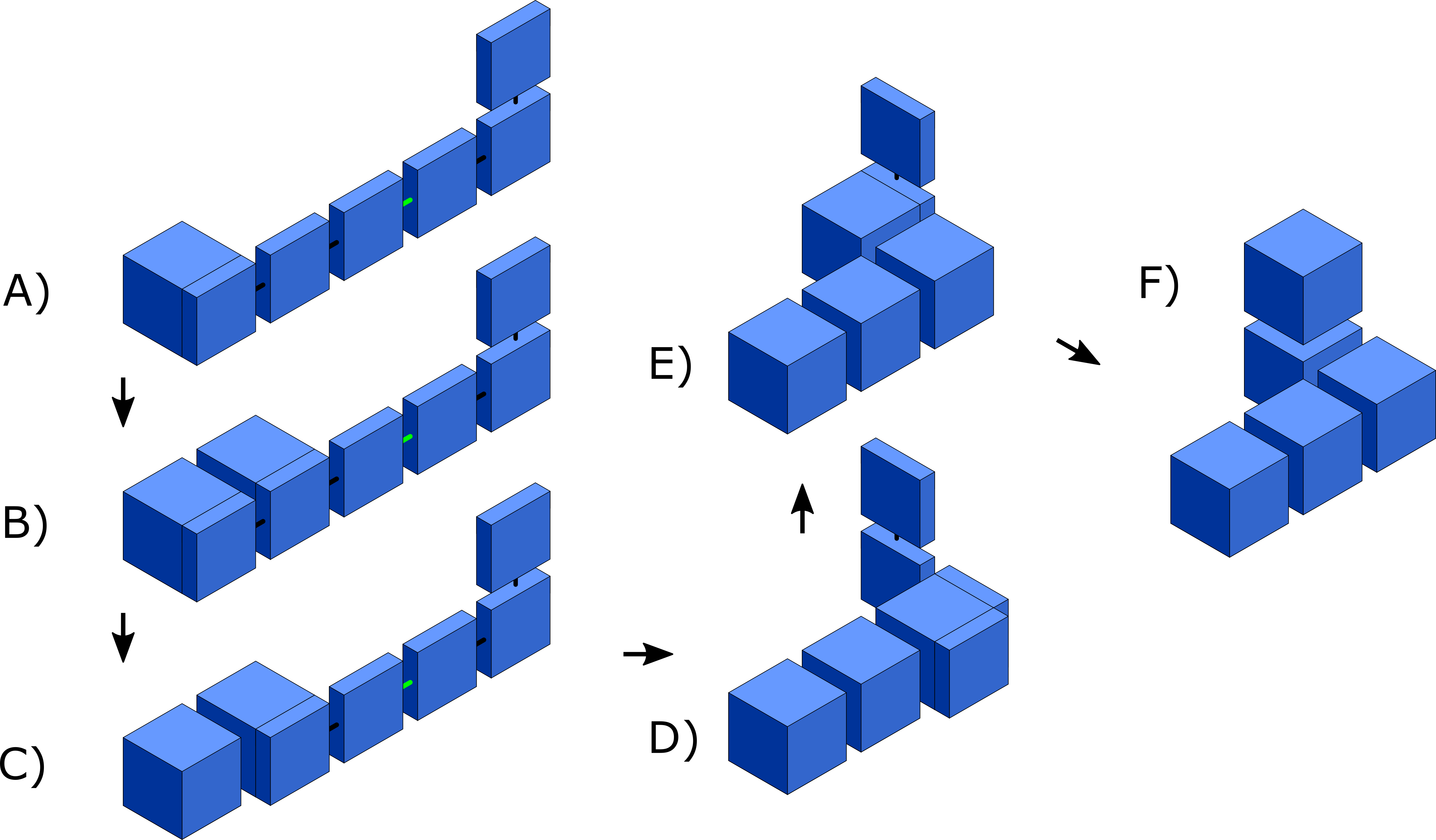}
    \caption{Building $\pi$ from $\mu^\prime$ (same as in Figure \ref{fig:SR_translation}). After the start cube binds to $\mu^\prime$ in step A, the process of assembling $\pi$ successively adds cubic tiles then detaches flat tiles from $\mu^\prime$. Step F is phenotype $\pi$ originally encoded by $\sigma$.}
    \vspace{-10px}
\end{figure}

\update{
    \subsection{Tiles of \texorpdfstring{$T$}{T}}
    \label{sec:tiles}
    We provide the enumerated sets of tiles in this section which provide for the dynamics as described in the prior sections.
    
    \vspace{-10pt}
    \subsubsection{\texorpdfstring{$T_\sigma$}{T sigma}}
    As shown in Figure \ref{fig:genome_tiles}, all tiles except for the end tile have the same structure of signals and glues, where the glues are a specific mapping to tiles in $T_\mu$. Glues which bind between $T_\sigma$ and $T_\mu$ have the $\mu$ subscript in the glue description. Glues without the $\mu$ subscript bind between the north and south glues of tiles in $T_\sigma$.
    
    \vspace{-10pt}
    \subsubsection{$T_\mu$}
    \begin{figure*}
        \centering
        \includegraphics[width=0.95\textwidth]{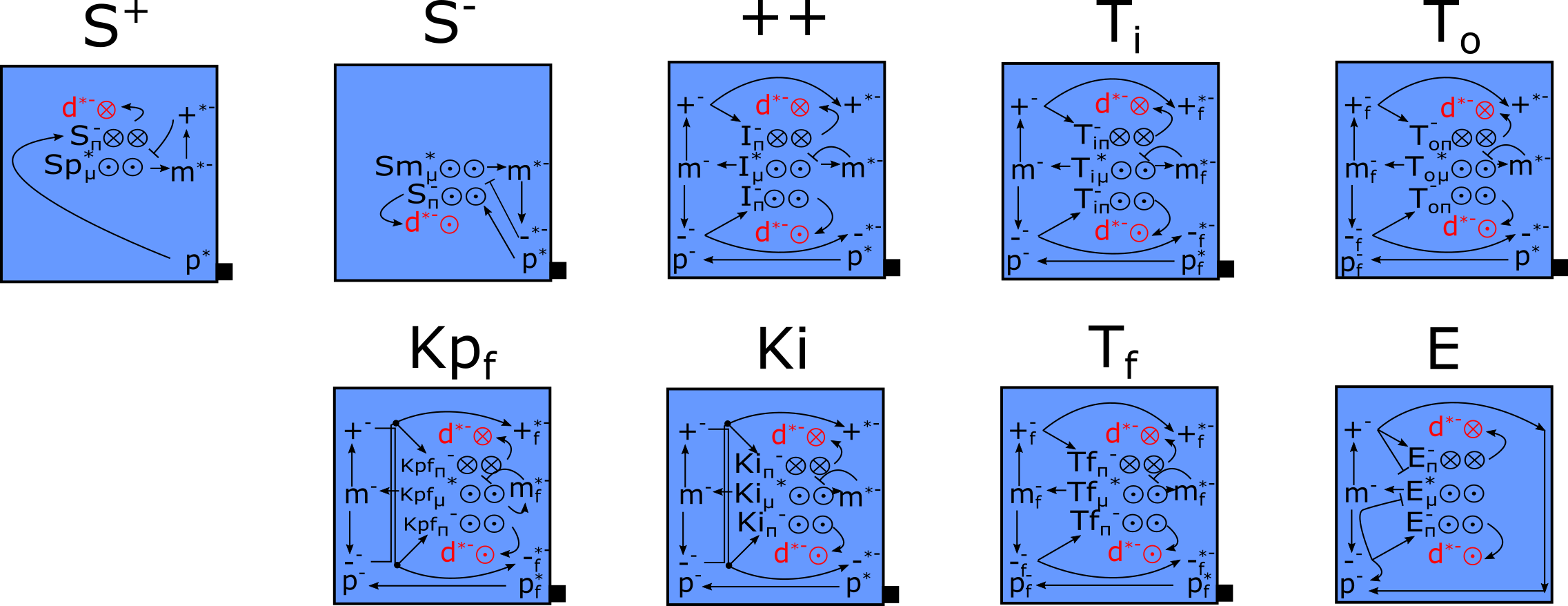}
        
        \caption{Messenger tile types (non-kink). Note that the red `d' glues have deactivation signals to all glues on the tile, but are omitted for visual clarity. This turns the messenger tile into a `junk' product.}
        \label{fig:start-tile-vox}
        \vspace{-10pt}
    \end{figure*}

    The tiles presented in Figure \ref{fig:start-tile-vox} represent the base tiles which make up a messenger sequence. Any glue which contains an `f' subscript is a flexible glue. The tile denoted $Ki$ is a placeholder for both $Kp$ and $Km$ tiles, where all glues which contain an `$i$' can be replaced with $p$ or $m$, respectively. All of the tiles aside from $T_i, T_f, Kp_f \text{ or } E$ can be a predecessor to a turning tile. This requires additional glues and signals in order to attach to a kink-ase structure. These modifications are shown in Figure \ref{fig:kink-mods}, and we note that these glues and signals overlay on top of the tiles in Figure \ref{fig:start-tile-vox}; glues not used in the turning process are omitted. The tiles to the right indicate the specific glues and signals for the $Kp,\:Km$ tiles. The tiles to the left indicate the specific glues and signals which must be present on the predecessor tiles to $Kp$ or $Km$. We note that $Kp$ and $Km$ can also be modified with the tiles on the left hand side. In the case of either two $Kp$ or $Km$ tiles in a row, it is required to leave the flexible glues $f_f,g_f$ \texttt{on} instead of \texttt{off} when  the `p' glue on the east side of a tile is bound.
    
    We note that the modifications require a mapping of a specific glue from $T_\sigma$ to $T_\mu$. This is accomplished by adding an additional `m' or `p' to the glue based upon the modification made. Glue which connect $T_\mu$ and $T_\pi$ have the subscript $\pi$.
    
    \begin{figure}[ht]
        \centering
        \includegraphics[width=0.67\linewidth]{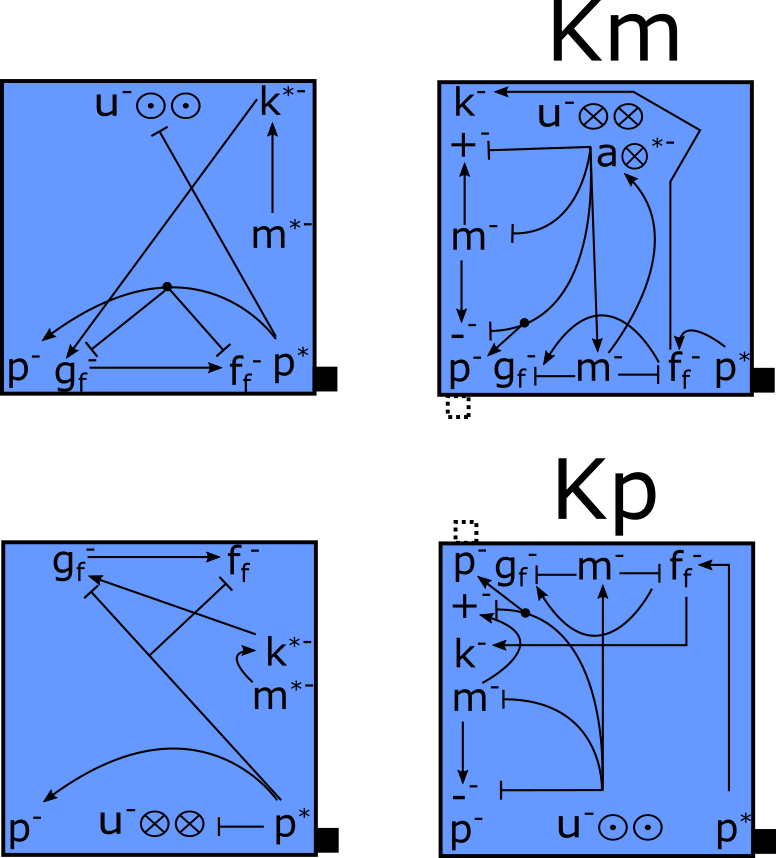}
        
        \caption{Tile modifications for use with kink-ase. Note that the dashed square indicates the face that the `p' glue is attached}
        \label{fig:kink-mods}
        \vspace{-5pt}
    \end{figure}
    
    \subsubsection{$T_\phi$}
    \begin{figure}[ht]
        \centering
        \includegraphics[width=0.99\linewidth]{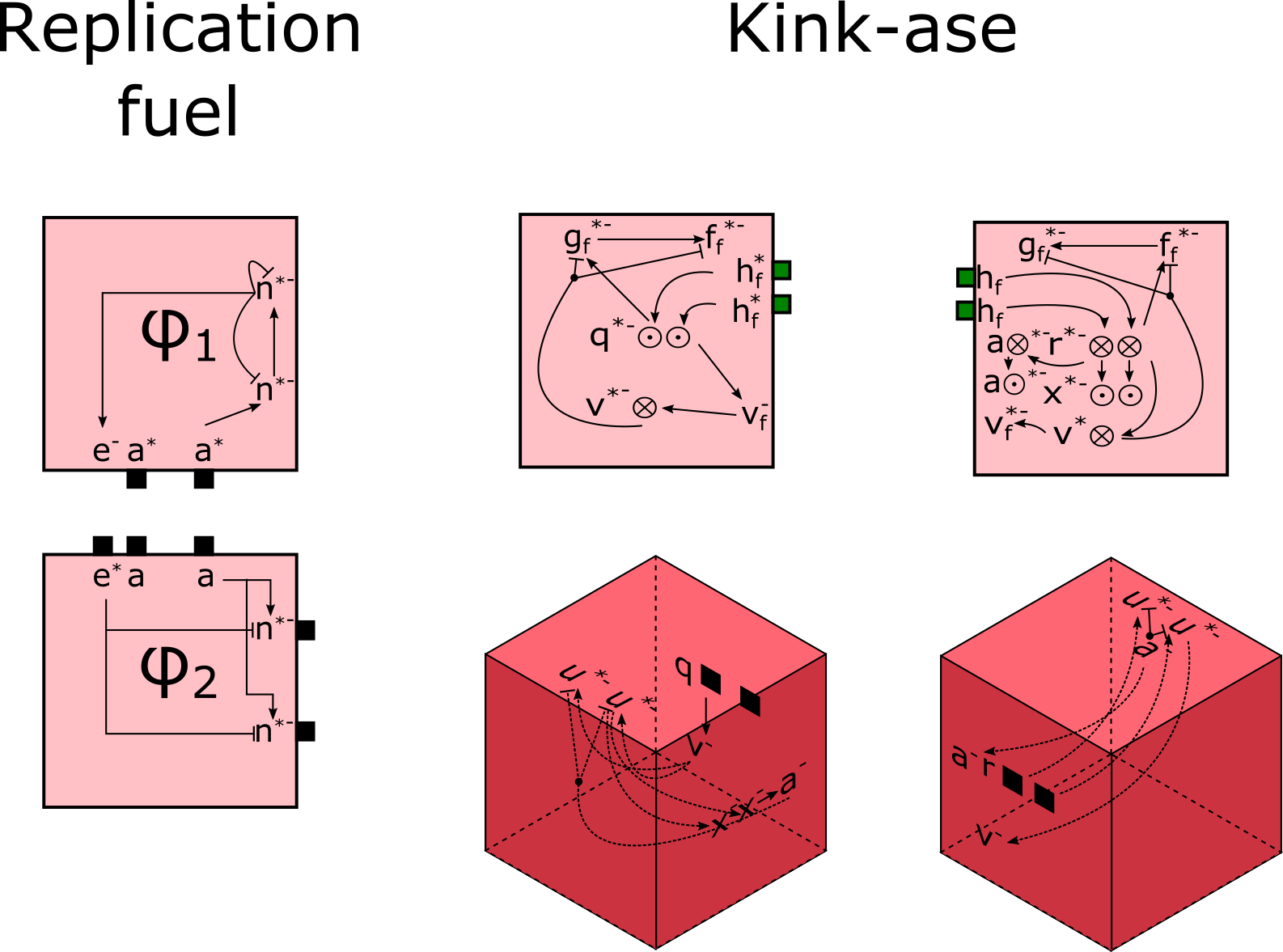}
        
        \caption{Fuel tiles. Tiles on left utilized during replication of $\sigma$, tiles on right combine to form kink-ase structure}
        \label{fig:fuel-tiles}
        \vspace{-10pt}
    \end{figure}
    The tiles presented in Figure \ref{fig:fuel-tiles} are those that cause the replication of $\sigma$ and form kink-ase. The kink-ase tiles first combine to form supertiles of size 4 as shown in Figure \ref{fig:linear-to-kin}. These supertiles are then able to perform the designated functions of the kink-ase. Similarly, the tiles $\varphi_1$ and $\varphi_2$ combine to a supertile in before replication of $\sigma$ can begin.
    
    \vspace{-10pt}
    \subsubsection{$T_\pi$}
    The tiles $T_\pi$ are the structural blocks which recreate a desired shape given an input genome. Two strength 1 glues of the type `c' bind the final structure between cubic tiles in the Hamiltonian path dictated by $\sigma$.
    \begin{figure}[ht]
        \centering
        \includegraphics[width=\linewidth]{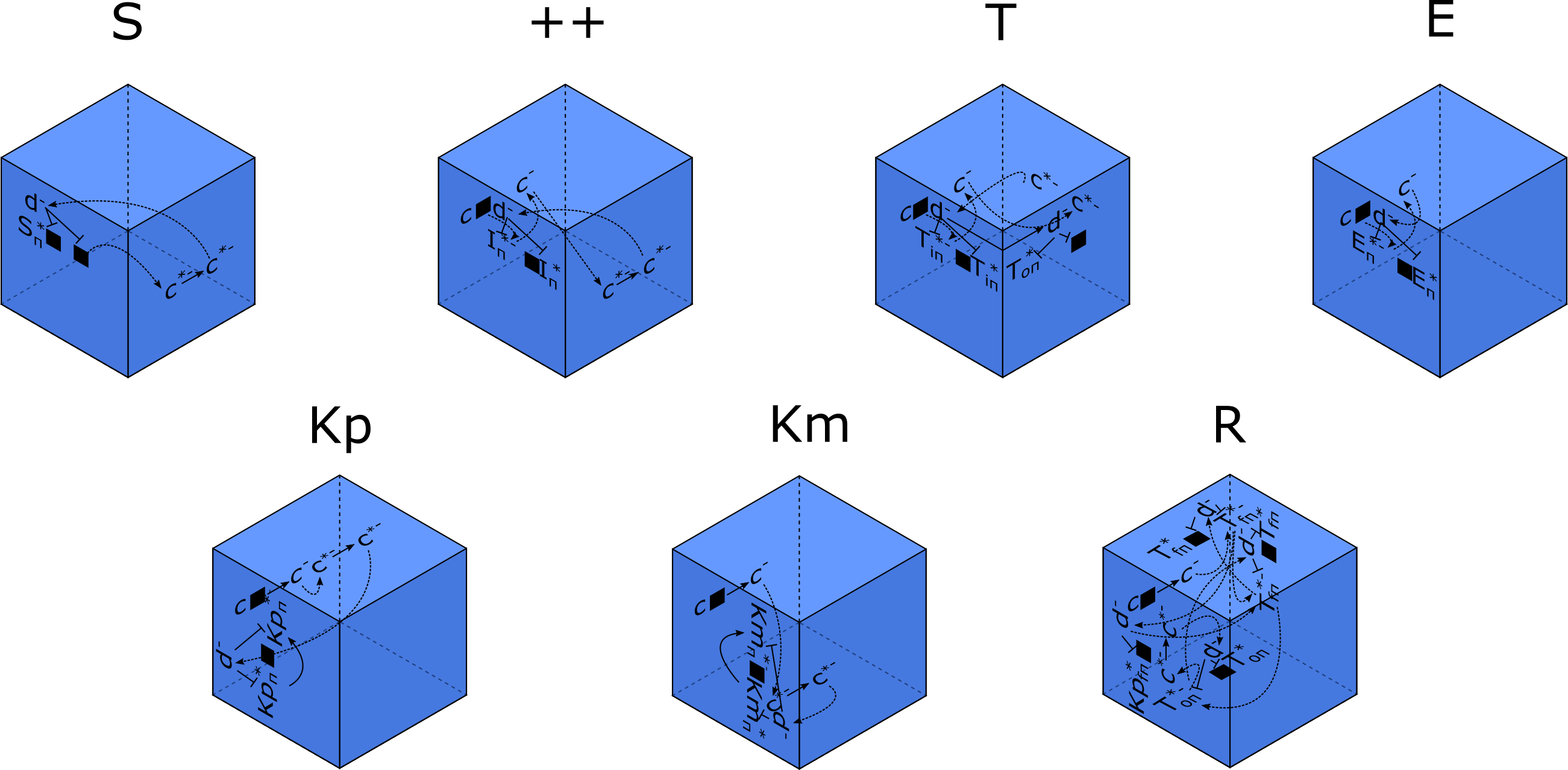}
        
        \caption{Structural tiles which create the assembly $\pi$. Note that the `R' tile has a second $Kp^*_{f\pi}$ glue activated, however is omitted for visual purposes.}
        \label{fig:struct-tiles}
        \vspace{-10pt}
    \end{figure}
}

\vspace{-10pt}
\subsection{Analysis of \texorpdfstring{$\mathcal{R}$}{R} and its correctness}

\begin{figure}[ht]
    \centering
        \centering
        \includegraphics[width=\linewidth]{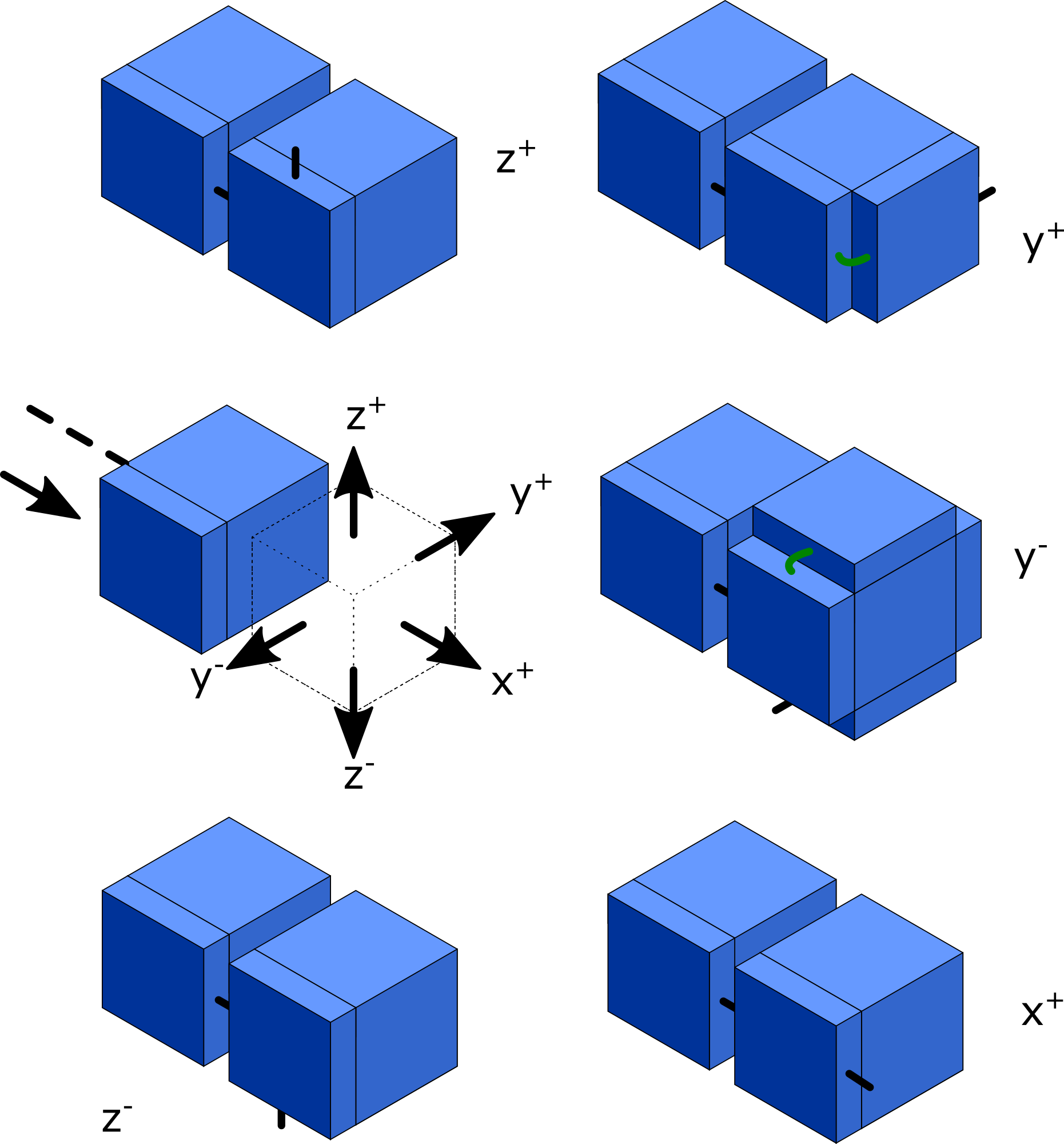}
    \caption{  \label{fig:SR-correctness}}

    \caption{
    The inductive steps required in the creation of $\pi$ which follows a Hamiltonian path given by a $\sigma$. The arrow going into the flat tile is the direction taken by the Hamiltonian path in the prior tile addition step. The five arrows indicate possible directions for the direction of the Hamiltonian path after the placement of the transparent cubic tile.}
    \vspace{-10pt}
\end{figure}



\begin{theorem}\label{thm:universal-constructor}
There exists an STAM* tile set $T$ such that, given an arbitrary shape $S$, there exists STAM* system $\mathcal{R} = (T,\sigma,2)$ and $S^2$ self-assembles in $\mathcal{R}$ with waste size $4$.
\end{theorem}



\update{
    We prove Theorem \ref{thm:universal-constructor} via induction. 
    Our base case is the start flat tile and its associated cube.
    Our inductive step is the addition of a cube and a direction associated with the next step of the Hamiltonian path within $S^2$.
    This direction is provided by the successor tile in $\mu^\prime$, and all possible directions are enumerated in Figure \ref{fig:SR-correctness}.
    At each step, we place a cubic tile in its associated direction based upon the flat tile in $\mu^\prime$.
    We analyze the possible direction of placement.
    Since $\mu$ is a translation of $\sigma$, $x^-$ is not included as it is the location of the prior cubic tile.
    As a note, the directions provided in the proof reflect those indicated in Figure \ref{fig:SR-correctness}, not necessarily the absolute reference of the entire system.
    Additionally, as our genome $\sigma$ has a Hamiltonian path ending on an exterior face of $S$, we can guarantee that diffusion is possible for a tile at any stage of construction
    
    
    \begin{itemize}
        \item[$x^+$:] This placement and output direction is carried out by the ++ tile type - the cubic tile is placed in the existing direction of travel
        \item[$y^+$:] This correlates to the $T_i$ and $T_o$ tile type. 
        \item[$y^-$:] This case is the most complex; we are changing the direction of travel in a direction which takes us through the tile of $\mu^\prime$. This requires the use of the following 4 tiles: $Kpf,T_f,T_f,T_o$. This could also be completed with a set of 3 tiles $Kp, Km, Km$, however this increases fuel usage per $y^-$ from 1 to 3, and overall tile usage from 8 to 19 when including all the singleton tiles utilized to create the kink-ase structures consumed by the 3 turning tiles.
        \item[$z^-$:] A single $Km$ tile carries out this tile placement and path change. Note, the prior flat tile must additionally be modified to carry out the turning action by the kink-ase.
        \item[$z^+$:] A single $Kp$ tile carries out this tile placement and path change. Note, the prior tile must additionally be modified to carry out the turning action by the kink-ase.
    \end{itemize}
    
    After the addition of a tile, we re-orient the frame of reference to align with that shown in Figure \ref{fig:SR-correctness}.
    The last tile in the Hamiltonian path will not have a new direction - this is indicated by the end tile.
    We have then generated the structure $S^2$ utilizing $R$.
    
    \subsubsection{STAM* Metrics of \texorpdfstring{$R$}{R}}
    The STAM* metrics of $R$ follow from the tileset found in Section~\ref{sec:tiles}: 
    \begin{itemize}
        \item Tile complexity $= 57$
            \begin{itemize}
                \item $|T_\sigma|=22$  
                \item $|T_\mu|=22$
                \item $|T_\pi|=7$
                \item $|T_\phi|=6$
            \end{itemize}
        \item Tile shape complexity $= 2$
        \item Signal complexity $= 7$
        \item Seed complexity $= O(n)$; each cube in the phenotype must be placed by a tile, with some requiring multiple (e.g. turns). As described above, for any structure with greater than 2 tiles we end up with the following number of tiles in $\sigma$ based upon the changes in directions which must occur: ``start tile'' $+$ ``end tile'' $+ |z^+| + |z^-|+ 2|y^+|+4|y^-|+|x^+|$.
    \end{itemize}

}


\section{A Self-Replicator that Generates its own Genome}
\label{sec:deconstruct}

In this section we outline our main result: a system which, given an arbitrary input shape, is capable of disassembling an assembly of that shape block-by-block to build a genome which encodes it. We describe the process by which this disassembly occurs and then show how, from our genome, we can reconstruct the original assembly. Here we describe the construction at a high level.
We prove the following theorem by implicitly defining the system $\mathcal{R}$, describing the process by which an input assembly is disassembled to form a ``kinky'' genome which is then used to make a copy of a linear genome (which replicates itself) and of the original input assembly.

\newcounter{tempcounter}
\setcounter{tempcounter}{\value{theorem}}

\begin{theorem}\label{thm:deconstructor}
There exists a universal tile set $T$ such that for every shape $S$, there exists an STAM* system $\mathcal{R} = (T,\sigma_{S^2},2)$ where $\sigma_{S^2}$ has shape $S^2$ and $\mathcal{R}$ is a self-replicator for $\sigma_{S^2}$ with waste size 2.
\end{theorem}


In this construction, there are two main components which here we call the \emph{phenotype} and the \emph{kinky genome}.

\update{
    Given a shape $S$, the phenotype $P$ will be a 2-scaled copy of the shape, so that each cube in $S$ corresponds to a $2\times 2\times 2$ block of tiles in $P$. The shape of the phenotype will therefore be identical to $S$ modulo our small, constant scale-factor. $P$ will be made up of tiles from some fixed $\STAMstar$ tile system $\mathcal{T}$ which we will define in more detail later.
    
    Let $H$ be a Hamiltonian path that goes through each tile in $P$ exactly once. We will construct $H$ later, but for now assume that it exists. Each tile in $P$ will contain the following information encoded in its glues and signals.
    \vspace{-3pt}
    \begin{itemize}
        \item Which immediately adjacent tile locations belong to the phenotype
        
        \item Which immediately adjacent tile locations correspond to the next and previous points in the Hamiltonian path
        
        \item Any glues and signals necessary for allowing the deconstruction and reconstruction process to occur as described in Sections~\ref{sec:dissasembly} and~\ref{sec:reassembly}
    \end{itemize}
}

In our system, the genome will be constructed as the phenotype is deconstructed and then will be duplicated or used to make copies of the original phenotype. 
Throughout this section, we refer to the cubic tiles that make up the phenotype as structural tiles and the flat tiles that make up the genome as genome tiles. Additionally, the tiles used in this construction are part of a finite tile set $T$, making $T$ a universal tile set. The genome is referred to as ``kinky'' due to the fact it must contain flexible glues, in contrast to the linear genome utilized in Section~\ref{sec:simple-replicator}.

\begin{figure}
    \centering
    \includegraphics[width=\linewidth]{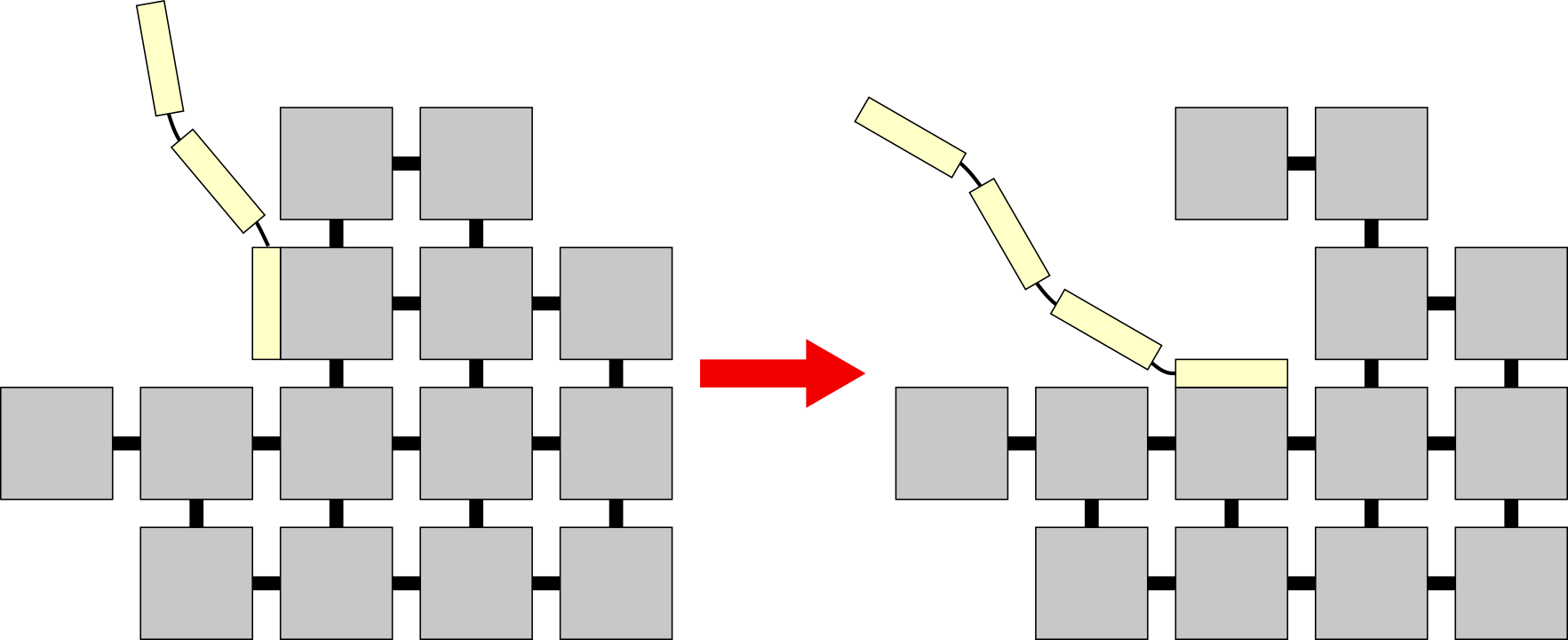}
    \caption{During disassembly, the genome will be dangling off of a single structural tile in the phenotype. In each iteration, a new genome tile will attach and the old structural tile will detach along the Hamiltonian path embedded in the phenotype.}
    \label{fig:DR_overview}
    \vspace{-10pt}
\end{figure}

\vspace{-10pt}
\subsection{Disassembly}
\label{sec:dissasembly}

\update{
    Given a phenotype $P$ with embedded Hamiltonian path $H$, the disassembly process occurs iteratively by the detachment of at most 2 of tiles at at time. The process begins by the attachment of a special genome tile to the start of the Hamiltonian path. In each iteration, depending on the relative structure of the upcoming tiles in the Hamiltonian path, new genome tiles will attach to the existing genome encoding the local structure of $H$ (to be used during the reassembly process) and, using signals from these incoming genome tiles, a fixed number of structural tiles belonging to nearby points in the Hamiltonian path will detach from $P$. A property called the \emph{safe disassembly criterion} will be preserved after each iteration assuring that disassembly can continue as described. This process will continue until we reach the last tile in the Hamiltonian path. Once the final genome tile binds to the existing genome and this final tile, signals will cause these final structural tiles to detach and leave the genome in its final state where it can be used to make linear DNA as described above or replicate that phenotype as described below.
    
    \vspace{-10pt}
    \subsubsection{Relevant Tiles and Directions}
    
    In each iteration of our disassembly procedure, indexed by $i$, we will label a few important directions and tiles which will be useful. Since our tiles in this model are not required to reside in a fixed lattice, we define our cardinal directions $\{N, E, S, W, U, D\}$ arbitrarily so that they are aligned with the faces of some arbitrarily chosen tile in our phenotype. These directions will only be used when referring to tiles bound rigidly to the phenotype so there will be no ambiguity in their use.
    
    The first tile, which we will call the \emph{previous structural tile} and write as $S^\text{prev}_i$, is the structural tile to which the genome is attached at the beginning of iteration $i$. This tile will detach from the rest of the phenotype by the end of iteration $i$. The \emph{next structural tile}, written $S^\text{next}_i$, is the structural tile to which the genome will be attached at the end of iteration $i$. Note that in some cases, this may not be the tile corresponding to the next tile in the Hamiltonian path, since we may detach more than one tile in an iteration.
    
    We will refer to the corresponding attached genome tiles accordingly and write $G^\text{prev}_i$ and $G^\text{next}_i$ respectively.
    
    The first direction, which we will call the \emph{next path direction} and write $D^p_i$, represents the direction from the previous structural tile to the next tile in the Hamiltonian path. Next, we will refer to the direction corresponding to the face of the previous structural tile upon which the previous genome tile is attached as the \emph{genome direction} and write $D^g_i$.
    
    We also define a direction called the \emph{dangling genome direction}, written $D^d_i$, relative to the previous genome tile attached to the previous structural tile. At each iteration of the disassembly process new genome tiles will attach to the existing genome and the phenotype. By the end of in iteration, the previous genome tile will have detached from the structure and the next genome tile will be attached to the next structural tile. The dangling genome direction is defined to be the direction relative to the previous genome tile in which the rest of the genome is attached.
    
    Figure \ref{fig:SDC_directions} illustrates what these directions look like in a particularly simple case.
    
    \begin{figure}
        \centering
        \includegraphics[width=0.45\linewidth]{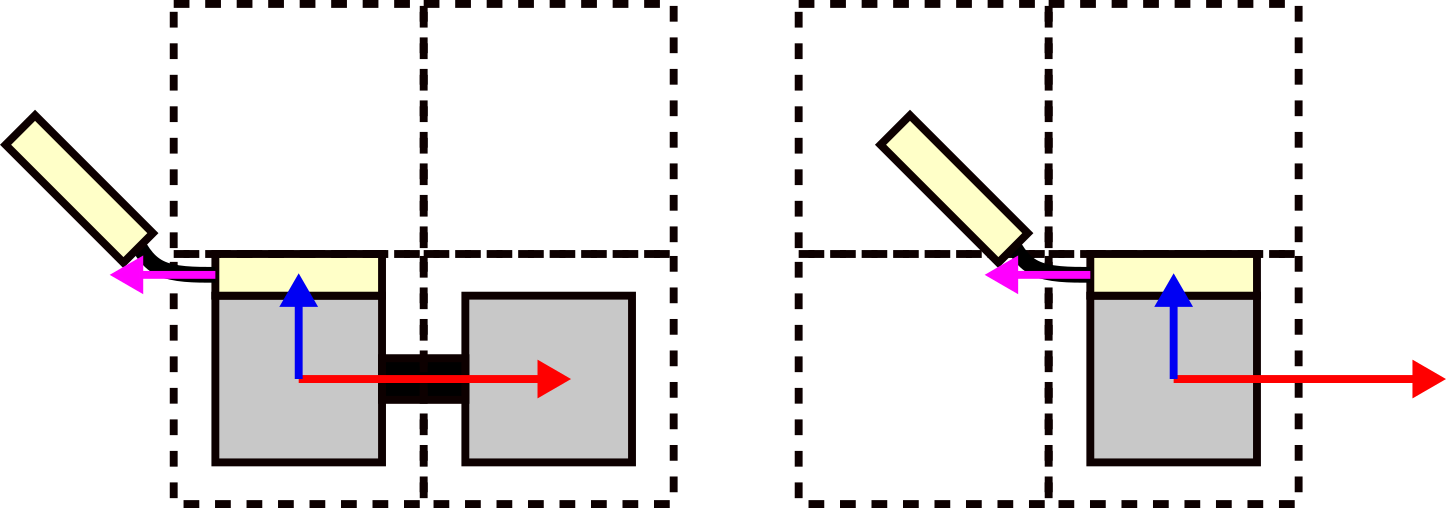}
        \caption{The relevant directions before and after an iteration of the disassembly process. The red arrow represents the next path direction, the blue arrow represents the genome direction, and the magenta arrow indicates the dangling genome direction. In this simple case the directions do not change after an iteration, but this is not always the case.}
        \label{fig:SDC_directions}
        \vspace{-10pt}
    \end{figure}
    
    \vspace{-10pt}
    \subsubsection{The Safe Disassembly Criterion}
    
    To facilitate in showing that the disassembly process works without error, we define a criterion which is preserved through each iteration of the disassembly process effectively acting as an induction hypothesis. We call this criterion, the \emph{safe disassembly criterion} or \emph{SDC}. The SDC is met exactly when all of the following are met:
    
    \begin{enumerate}
        \item There is no phenotype tile in the location location in the direction $D^g_i$ relative to the previous structural tile. This essentially means that there was room for the previous genome tile to attach to the previous structural tile.
        
        \item At the current stage of disassembly, there is a path of empty tile locations that connects the previous tile location to a location outside the bounding box of the phenotype. This condition ensures that if our path digs into the phenotype during disassembly, there is a path by which detached tiles can escape and new genome tiles can enter to attach.
        
        \item The dangling genome direction is not the same as the next path direction. This ensures that the existing genome is not dangling off of the previous genome tile in such a way that it would block the attachment of the next genome tile. This also ensures that our genome will never have to branch, though it may take turns.
        
        \item Both the previous genome tile and some adjacent structural tile are presenting glues which allow for the attachment of another genome tile.
    \end{enumerate}
    
    \subsubsection{Disassembly Cases}
    
    In each iteration of disassembly, there will be 6 effective possibilities regarding the local structure of the Hamiltonian path. Each of these possibilities will necessitate a different sequence of tile attachments and detachments for disassembly to occur. These cases are illustrated in figure \ref{fig:DR_case_enum} and described as follows.
    
    
    \begin{figure}
        \centering
        \includegraphics[width=\linewidth]{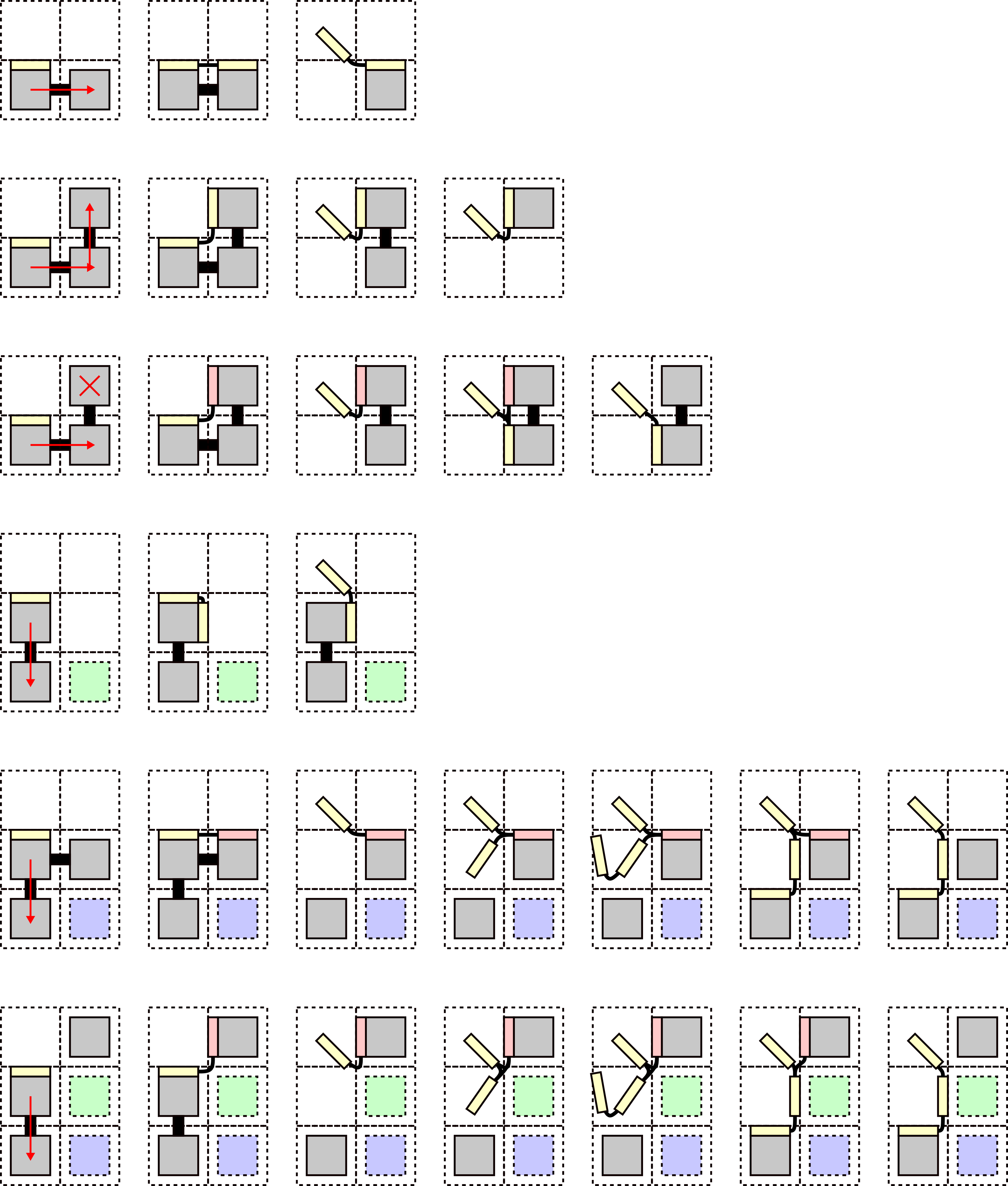}
        \caption{A side view of the disassembly process for all 6 cases. Each row is a unique case, where the leftmost image is the starting condition. We orient these illustrations so that the previous genome direction is always up for convenience. Also note that we always illustrate the dangling genome direction to the left, but this need not be the case, this is just for making visualization easier. In reality, the dangling genome direction could be in any direction relative to the previous genome, so long as it satisfied the SDC condition that it is not the same as the next path direction. Gray squares represent attached structural tiles, green squares represent a location in which it does not matter if an attached structural tile exists, and empty squares represent locations in which no attached structural tile exists.}
        \label{fig:DR_case_enum}
        \vspace{-10pt}
    \end{figure}
    
    \begin{lemma}\label{lem:DR_all_cases}
    The 6 cases illustrated in Figure~\ref{fig:DR_case_enum} are all of the possible cases for a disassembly iteration.
    \end{lemma}
    
    First note that the next path direction can either be perpendicular to the previous genome direction or not. If it is, we consider two cases. Either the tile location in the next genome direction relative to the next structural tile in the Hamiltonian path contains an attached structural tile or it doesn't. Case 1 is where it doesn't. If on the other hand it does, call the tile in that location the blocking tile; case 2 occurs when the blocking tile follows the next structural tile in the Hamiltonian path and case 3 occurs when it doesn't.
    
    Supposing that the next path direction is not perpendicular to the previous genome direction, either it's the same direction or the opposite direction. By condition 1 of the SDC, it cannot be the same direction since there can be no structural tile attached in that location so all other cases must have the next path direction opposite the previous genome direction.
    
    Now we define the working direction to be the direction opposite the dangling genome direction. This direction will be the direction in which genome tile attachments will occur during the remaining cases. Ultimately this choice is arbitrary, except that the working direction cannot be the dangling genome direction. Let location $a$ be the tile location in the working direction of the previous structural tile and location $b$ be the tile location in the opposite direction of the next path direction of location $a$. Case 4 is when neither location $a$ nor $b$ contains an attached structural tile, case 5 occurs when only location $a$ has an attached tile, and case 6 occurs otherwise.
    
    Notice that since we defined these cases by dividing the possibility space into pieces where either some condition is or isn't met, this enumeration of cases represents all possibilities, thus proving Lemma \ref{lem:DR_all_cases}.
    
    \subsubsection{The Disassembly Process}
    
    Here we describe the disassembly process in enough detail that anyone familiar with basic tile assembly constructions should be able to derive the full details of the process without much difficulty.
    
    Before any of the iterative disassembly cases can occur, the disassembly process begins with the attachment of the initial genome tile. The structural tile corresponding to the first point in the Hamiltonian path will be presenting a strength 2 glue to which this initial genome tile can attach. At this point in the process, this will be the only tile to which anything can attach with sufficient strength. This attachment activates a signal which turns off all glues in this initial structural tile except those holding it to the initial genome tile and the next structural tile in the Hamiltonian path. Also, now that this first genome tile has attached, the next genome tile can cooperatively attach initiating the disassembly process so that in the first iteration, the initial genome tile acts as the previous genome tile and the structural tile to which it's attached acts as the previous structural tile.
    
    In each following iteration, once complete, what used to be called the next structural tile and next genome tile become the previous structural tile and previous genome tile for the next iteration and any relevant directions in the next iteration are specified relative to these new previous tiles.
    
    Each of the cases as described above makes use of a unique sequence of tile genome attachments and signals; however, much of the logic in each of the cases is the same. We will describe two of the cases in greater detail than the rest, specifically cases 1 and 3, since understanding the details of those cases will make understanding the others much easier. Figure \ref{fig:DR_case_enum} illustrates the high level process of each case. It's important to keep in mind that the entire structure of the Hamiltonian path is encoded in the glues and signals of the phenotype tiles. This means that these cases can occur without issue since, for example, in an iteration where case 3 needs to occur, there will only be the glues and signals for case 3 present on the relevant tiles and none that would allow tiles for say case 5 to attach.
    
    \begin{figure}
        \centering
        \includegraphics[width=0.7\linewidth]{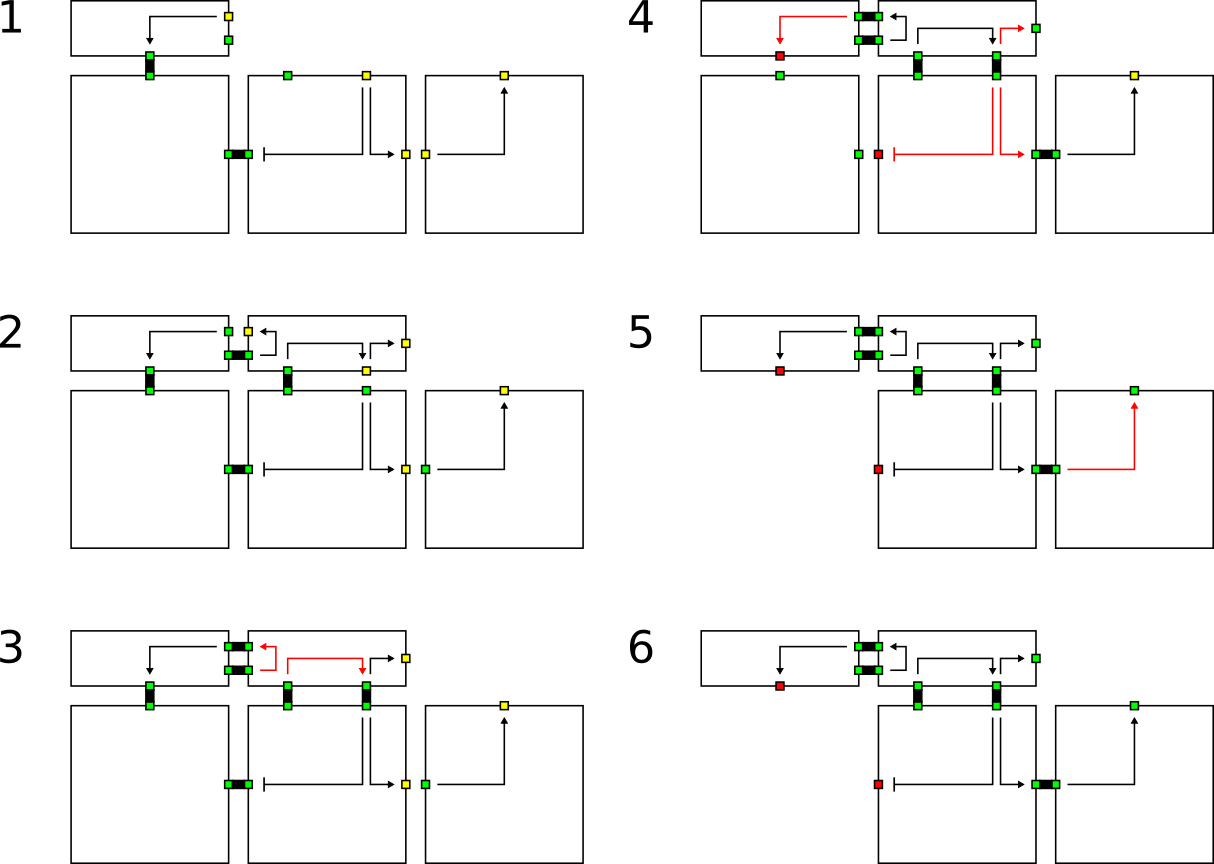}
        \caption{A side view of some of the relevant glues and signals firing during the simplest disassembly case.}
        \label{fig:DR_case_details-append}
        \vspace{-15pt}
    \end{figure}
    
    \begin{enumerate}
        \item This case is the simplest case and is illustrated in Figure \ref{fig:DR_case_details-append}. First, a genome tile $G$ attaches cooperatively to the previous genome tile and the next structural tile. This attachment causes signals to fire in $G$ that activate 2 glues from the latent state to the on state. The first of these glues is a rigid, strength 2 glue that allows $G$ to bind rigidly and with more strength to the next structural tile. The other glue is a flexible, strength 2 glue that allows the genome to more strongly attach to the previous genome tile. The attachment of these glues activate signals which turn the old glues serving the same purpose into the off state. Additionally, signals are activated in the previous genome tile and the next structural tile disabling the glues in both that held onto the previous structural tile. Signals also deactivate any glues in the next structural tile that are attached to all other structural tiles except for the one following it in the Hamiltonian path.
        
        At this point, there are no glues holding the previous structural tile to the genome nor the phenotype. This structural tile is now free to float away from what's left of the phenotype which is possible since the genome to which it was attached is now only bound with a flexible glue to the next genome tile and, by SDC condition 2, there is a path of empty tile locations along which it can escape.
        
        In addition to all of the signals described previously, signals also activate a glue on the next genome tile which enables the attachment of the genome tile that will initiate the next iteration of the disassembly process.
        
        By definition of case 1, SDC conditions 1 and 2 will be met after this process is done. Additionally, since the dangling genome direction now corresponds to the direction of the detached structural tile, condition 3 must also be satisfied. Condition 4 is also satisfied since glues were activated on the upcoming tile in the path to allow for cooperative binding of a new genome tile.
        
        \item This case is largely similar to case 1 except that the next genome tile attaches to the structural tile following the next structural tile in the Hamiltonian path since the next is being blocked. In this case, it will be necessary for this tile to ``know'' that the next genome tile will attach to it. To accomplish this, all of the necessary glues that allowed the disassembly process to occur in the first case exist on this tile instead of the one immediately following the previous structural tile in the Hamiltonian path.
        
        \item In this case, we have to remove the previous structural tile before we can attach the genome to the next structural tile since it is being blocked. We do this by utilizing what we call \emph{utility genome tiles}. These utility tiles are flat tiles that temporarily affix the genome to another part of the phenotype so that the previous tile can safely detach without the genome also detaching.
        
        At first, this case proceeds similar to case 2 (and is illustrated in Figure \ref{fig:DR_case_enum}), but with a utility tile attaching to the blocking structural tile instead of the next genome tile. This attachment activates signals which cause the previous structural tile to detach. Since the tile to which the utility tile attached is not immediately adjacent to the previous structural tile, this is done using a chain of signals (which is a common gadget in STAM systems). The detachment of the previous structural tile allows the next genome tile to cooperatively bind to the previous one and to the next structural tile. This attachment causes signals to deactivate glues holding the utility tile in place allowing it to detach.
        
        \item This case is largely degenerate and doesn't involve detachment of any tiles. Instead, utilizing cooperation, the next genome tile attaches to another face of the previous structural tile which also plays the part of the next structural tile. Depending on the tile or lack thereof in the green tile location from Figure \ref{fig:DR_case_enum}, the next iteration will either be case 1, 2, or 3.
        
        \item This case is largely similar to case 3 except that the utility tile attaches in a different location. Once this occurs, instead of a new tile attaching cooperatively to the next tile, which is impossible since the next tile is not adjacent to the previous genome tile, a filler genome tile attaches to glues that are now present after the attachment of the utility genome tile. This filler genome tile acts as a spacer and after signals activate its glues, the next genome tile can attach to it and the next genome tile.
        
        There is one consideration that needs to be made in this case. If the tile location illustrated in blue in case 5 of Figure \ref{fig:DR_case_enum} is the tile in the Hamiltonian path immediately following the next structural tile, then condition 3 of the SDC will not be met. This is because the dangling genome direction at the start of the next step will be in the same direction as the next path direction. To handle this, we simply require that two filler genome tiles attach between the utility tile and the next genome tile in this case. Since the structure of the Hamiltonian path is known in advance, this is possible, by requiring a different utility tile attach in the case where two filler tiles would be necessary than if only one was. Now, similar to case 3, the utility tile is free to detach following signals from the attachment of the next genome tile.
        
        \item This case is identical to case 5 except that the utility tile attaches in a different location.
    \end{enumerate}
}

\vspace{-20pt}
\subsection{Reassembly}
\label{sec:reassembly}


\update{
    At each iteration of the disassembly process, tiles attached to the genome encoding which tiles were detached.
    In some stages multiple tiles were detached, but it shouldn't be hard to see how that could be encoded in a single genome tile.
    Recall that this genome is a ``kinky'' genome.
    At this point, we could have defined the disassembly process above so that this genome immediately reconstructs the phenotype, the process for which is defined below; however, the definition of self-replicator requires that we construct arbitrarily many copies of the phenotype.
    Because of this, we can instead define the genome here so that it has the glues and signals necessary to convert into a linear genome as described in Section \ref{sec:simple-replicator}.
    
    We refer to the processes described in Section \ref{sec:details-kinkase}.
    There we use a gadget called kink-ase to convert a linear sequence of genome tiles into a ``kinky'' one which is capable of constructing a shape.
    This process is easily reversible using a similar gadget which follows the steps in Figure \ref{fig:linear-to-kin} in reverse.
    This process converts the kinky genome made during the disassembly of our phenotype into a linear genome which can be replicated arbitrarily using the process described in Section \ref{sec:detail-genome}.
    For our purposes, it's useful to modify this linear genome duplication process so that our linear genome is duplicated into two copies: one that can be further used for genome duplication and one that can be converted back to kinky form and used to reassemble the phenotype.
    This simply requires that we specify a second set of the corresponding glues and signals on the genome constructed from the disassembly process.
    This guarantees that we are generating arbitrarily many copies of the phenotype.
    
    
    Once we have kinky genomes ready to reconstruct the phenotype, we can begin the reassembly process. This process behaves much like the disassembly process, but with the genome being disassembled and the structure being reassembled. Once a reassembly fuel tile attaches to the special tile at the end of the genome, signals will activate glues allowing a structural tile, identical to the last tile in the Hamiltonian path of the original phenotype, to attach. This initiates the reassembly process and each of the tiles in the Hamiltonian path will attach in reverse order as the genome disassembles from the back. This process is in some ways more straightforward than disassembly because the only tiles that detach are genome tiles and they detach completely. In the assembly process, both structural tiles and genome tiles had to detach and the detachment of genome tiles had to happen in such a way that they were still attached by flexible glues to the rest of the genome.
    
    The following is an outline of the reassembly processes for each of the cases. Figure \ref{fig:DR_case_enum} can still be used as a reference but be careful to keep in mind that the process is happening in the opposite direction, initiated by the attachment of what was called the \emph{next} structural tile in the disassembly process. In this section we reverse the terminology so that in each iteration, what were the previous structural and genome tiles are now the next structural and genome tiles and vice-versa. In each iteration of this process, the attachment of the previous structural tile to our genome initiates the sequence of attachments, detachments, and signals that allow the next structural tile to attach and the previous genome tile to detach.
    
    \begin{enumerate}
        \item This is the most basic case, the attachment of the previous structural tile to the genome activates glues on the next genome tile. This enables the next structural tile to attach cooperatively which causes signals to deactivate glues so that the previous genome tile detaches.
        
        \item The attachment of the previous structural tile in this iteration activates glues on it which immediately allows the next structural tile to attach. Again this attachment activates signals which turn on glues to allow another tile to attach forming the corner. Finally, the next genome tile can bind to this last structural tile which causes glues to deactivate so that the previous genome tile detaches.
        
        \item The attachment of the structural tile to the genome in the previous iteration activates a glue on the genome tile and adjacent structural tile allowing a utility tile to attach. This causes signals to deactivate glues holding the previous genome tile and activating glues on the structural tile to which it was bound. This allows a new structural tile to attach and then the corresponding genome tile. These attachments create signal paths that deactivate glues on the utility tile and the structural tile to which it was attached, allowing it to fall off.
        
        \item This stage just represents the genome tile turning a corner which causes the old genome tile to detach after signals deactivate its glues. This can only happen after case 1, 2, or 3 similar to the analogous case during disassembly.
        
        \item The attachment of the structural tile activates glues which allow the utility tile to attach. This attachment initiates signals which do 3 things. the signals deactivate glues holding the previous genome to the structural tile, the signals deactivate glues holding the utility tile to the old genome tiles, and the signals activate glues on the next genome tile. The next genome tile can then cooperate with the old structural tile to attach a new structural tile. Note that in this case the filler genome tiles from the disassembly will remain attached to the previous genome tile and they will detach as a short chain.
        
        \item This case is almost identical to the previous case with a slightly different binding location for the utility tile.
    \end{enumerate}
    
    Note that in each of the cases described above it's possible to reassemble the phenotype structure using the same tiles that were originally in the seed phenotype. As described here, we require that some of the signals in these reassembled phenotype tiles will be fired to facilitate in the reassembly process; however, with a more careful design it wouldn't be difficult to describe a process which reassembles the phenotype without using any signals on the structural tiles if this was a desired property. Additionally, during cases 5 and 6, pairs of filler tiles will detach depending on the next direction of the path in that iteration. This results in our waste size being 2, but again with a more careful design it would be easy to specify tiles which, say, bind to these waste pairs and break them down into single tiles if having waste size 1 was a desired property.
}

\vspace{-10pt}

\update{
    \subsection{Phenotype Generation Algorithm}
    
    In this section, we describe an efficient algorithm for describing the $\STAMstar$ system in which this process runs. Given that we require complex information to be encoded in the glues and signals of our components, particularly in the phenotype since it requires an encoded Hamiltonian path, it might seem like we are ``cheating'' by baking potentially intractable computations in these glues and signals. This however is not the case in the sense that, as we will show, all of the required tiles, glues, signals, paths, etc. (all from a fixed, finite set of types) can be described by a polynomial time algorithm given an arbitrary shape to self-replicate.
    
    The algorithm described consists largely of two parts. First, we will determine a Hamiltonian path through our shape, and second we will use this path to determine which glues need to be placed where on our tiles.
    
    \subsubsection{Generating A Hamiltonian Path} \label{sec:hampath}
    
    \begin{lemma}
    Any scale factor 2 shape $S^2$ admits a Hamiltonian path and generating this path given a graph representing $S^2$ can be done in polynomial time.
    \end{lemma}
        
    In general, the problem of finding a Hamiltonian path through a graph is \textbf{NP}-complete and may be impossible for many shapes we may wish to use; however, if we scale our shape by a constant factor of 2, that is replace every voxel location with a $2\times 2\times 2$ block of tiles, then not only is there always a Hamiltonian path, but it can be computed efficiently. The algorithm for generating this Hamiltonian path is described in further detail in \cite{Moteins} and was inspired by \cite{SummersTemp}, but we will describe the procedure at a high level here using terminology that is convenient for our purposes.
    
    \begin{enumerate}
        \item Given a shape $S$, we first find a spanning tree $T$ through the graph whose vertices correspond to locations in $S$.
        \item We embed this spanning tree in a space scaled by a factor of 2 so that each vertex corresponds to a $2\times 2\times 2$ block of locations.
        \item To each $2\times 2\times 2$ block in this space, we assign one of two orientation graphs $G_o^1$ or $G_o^2$. These graphs each form a simple oriented cycle through all points. These graphs are assigned so that they form a checkerboard pattern such that no blocks assigned $G_o^1$ are adjacent to any blocks assigned $G_o^2$ and vice versa. Figure \ref{fig:orient_graph} illustrates what the orientation graphs look like for adjacent blocks.
        \item For each edge in the spanning tree $T$, we join the orientation graphs corresponding to the vertices of the edge so that they form a single continuous cycle as illustrated in Figure \ref{fig:orient_graph}. This process is described in more detail in \cite{Moteins}.
        \item Once we do this for all edges in our spanning tree, the connected orientation graphs will form a Hamiltonian circuit through the $2\times 2\times 2$ blocks corresponding to the tiles in our shape. This is easy to see by analyzing a few cases corresponding to all possible vertex types in the spanning tree and noting that in none of them does the path ever become disconnected. This is done in \cite{Moteins}.
    \end{enumerate}
    
    \begin{figure}
    \centering
    \subfloat{%
        \includegraphics[width=0.2\textwidth]{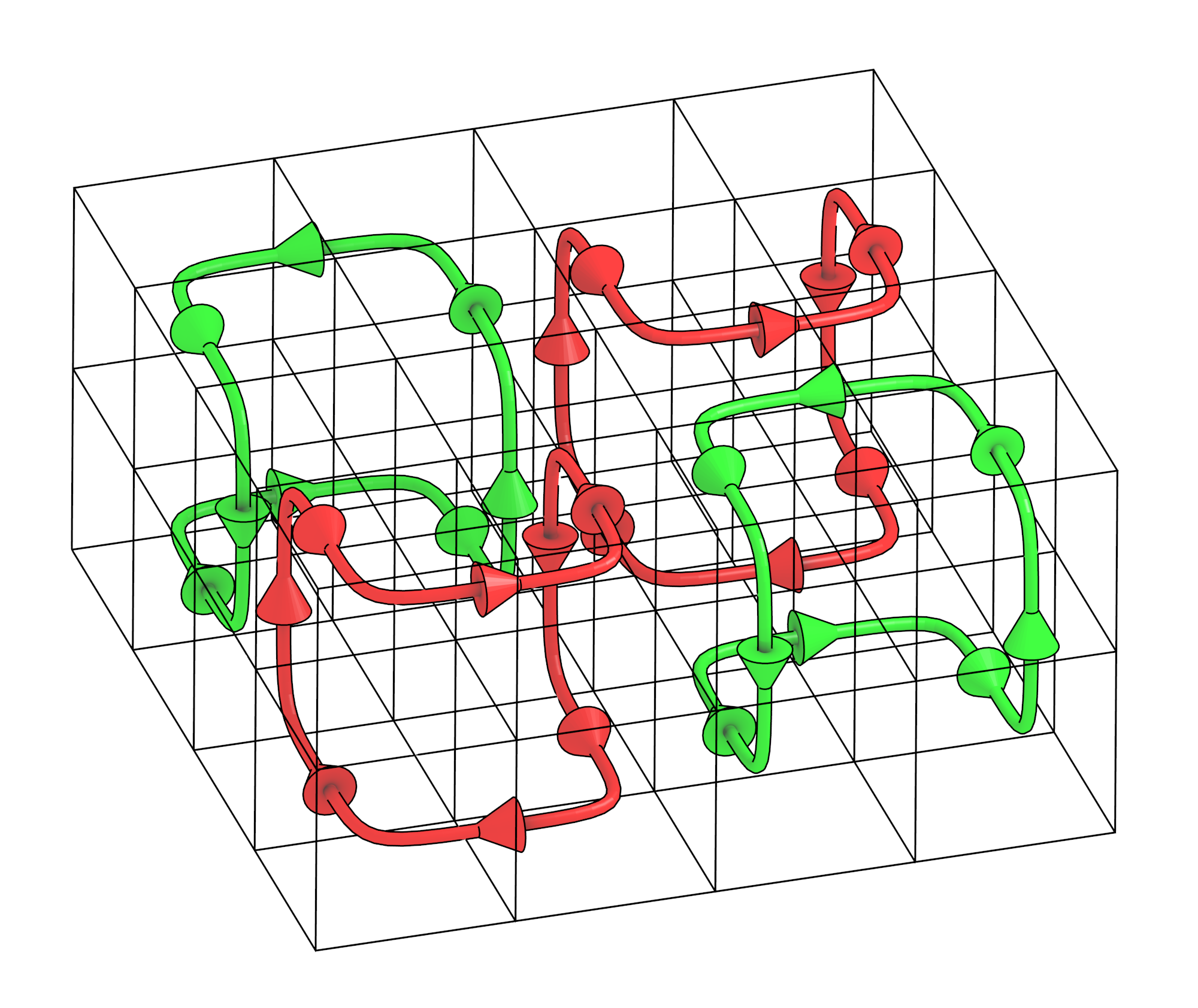}%
        \label{fig:DR_orient}%
    }%
    \hspace{20pt}
    \subfloat{%
        \includegraphics[width=0.2\textwidth]{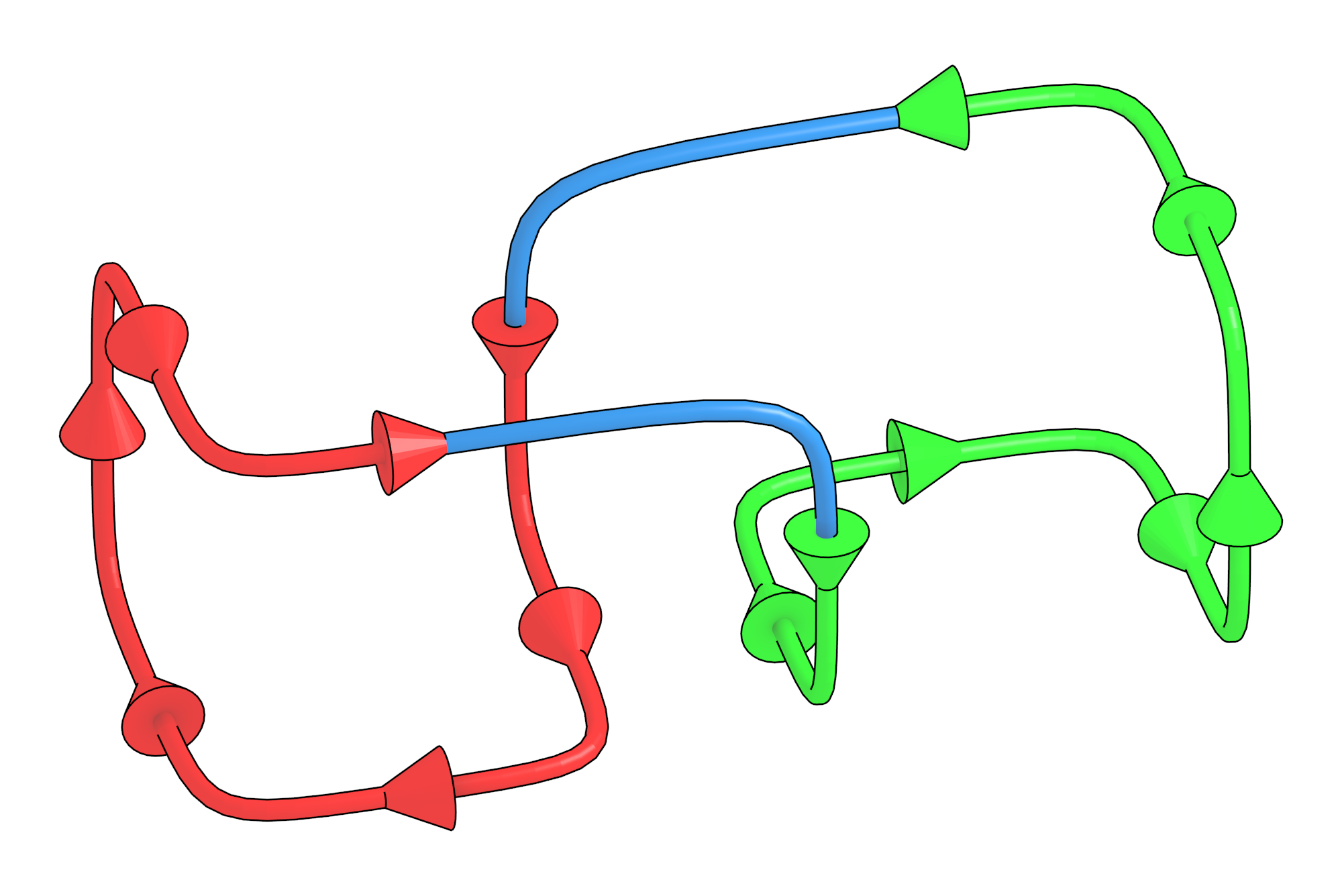}%
        \label{fig:DR_orient_join}%
    }%
    \caption{(Left) Each $2\times 2\times 2$ block of space is assigned an orientation graph which will be used to help generate the Hamiltonian path through our shape. Adjacent blocks are assigned opposite orientation graphs, the edges of which will help guide the Hamiltonian path around the shape. (Right) Orientation graphs of adjacent blocks are joined to form a continuous path}
    \label{fig:orient_graph}
    \end{figure}
    
    The resulting Hamiltonian path, which we will call $H$, passes through each tile in the 2-scaled version of our shape and only took a polynomial amount of time to compute since spanning trees can be found efficiently and only contain a polynomial number of edges. Given $H$, we can arbitrarily choose some vertex on the surface of our shape to represent the starting point of our path $H_1$ and label the rest of the path in order with respect to this one so that the next point is labeled $H_2$, then $H_3$, and so on. Additionally, we can also keep track of the location in space relative to some fixed origin to which each point in our path belongs and note that, using common data structures and basic arithmetic, determining the index of points in $H$ given a location can be done efficiently. 
    
    \subsubsection{Determining Necessary Information to encode in Glues and Signals}
    
    Recall that each case of the disassembly and reassembly processes sometimes required tiles nearby in space to have glues and signals to facilitate each step of the process. We define the following algorithm which is able to describe these glues and signals, showing that we can efficiently describe the tiles necessary for our construction.
    
    Begin with tile $H_1$ and iterate over the entire Hamiltonian path performing the following operations with the current tile labelled $T_i$ and keeping track of a counter $t$ which starts at 0.
    
    \begin{enumerate}
        \item Determine which of the 6 disassembly cases would apply to this particular tile by looking at adjacent tile locations and considering only those tiles not yet flagged with a detachment time.
        \item At this point, we know exactly which case $T_i$ will use during the detachment process. Assign any glues and signals necessary to this tile and adjacent tiles.
        \item Flag $T_i$ as being detached at time $t$.
        \item If $T_i$ used case 2, also mark the tile following $T_i$ as being detached at time $t$ and skip the next tile in the path for the next iteration.
        \item increment $t$ and $i$.
    \end{enumerate}
    
    Our algorithm now knows which glues and signals are necessary for each tile that will make up the phenotype. We can now iterate over all tiles in the construction and make a set consisting of each unique tile in the phenotype. Additionally, the genome tiles necessary for the process are even simpler to define since there is only a small fixed number needed for each case. This shows that the system in which this process occurs can be described efficiently by an algorithm and that we are not doing an unreasonable amount of pre-computation by including the necessary information in our glues and signals.
    
    \subsubsection{Glues for Converting to Linear DNA}\label{sec:kinky-to-linear-append}
    
    The disassembly process above results in arbitrarily many ``kinky'' genomes which are capable of being used to produce a replica of the original phenotype. In order for this process to be possible however, the kinky genome produced by the disassembly process needs glues and signals to indicate locations that should be ``un-kinked'' and replicated. This is no problem however since the only cases in the disassembly process that could induce a kink in our constructed genome are 1, 2, and 3. The kink induced in the genome in any of these cases solely depends on the dangling genome direction and next path direction. Since there are only a finite number of such cases and since our tileset will have a unique set of genome tiles that attach in each such case, we can easily specify the necessary glues and signals to the corresponding genome tiles. This guarantees that the conversion to linear DNA is possible for any genome constructed by the disassembly process.
    

    \vspace{-10pt}
    \subsection{Correctness of Theorem \ref{thm:deconstructor}}\label{sec:decontructor-correctness}
    
    First, we restate Theorem \ref{thm:deconstructor} for convenience:
    
    \newcounter{tempcounter2}
    \setcounter{tempcounter2}{\value{theorem}}
    
    \setcounter{theorem}{\value{tempcounter}}
    
    \begin{theorem}\label{thm:deconstructor-append}
    There exists a universal tile set $T$ such that for every shape $S$, there exists an STAM* system $\mathcal{R} = (T,\sigma_{S^2},2)$ where $\sigma_{S^2}$ has shape $S^2$ and $\mathcal{R}$ is a self-replicator for $\sigma_{S^2}$ with waste size 2.
    \end{theorem}
    
    \setcounter{theorem}{\value{tempcounter2}}
    
    We have shown how, given any shape $S$ as input, we can scale it by factor $2$ to $S^2$ and efficiently find a Hamiltonian path through $S^2$. We can then compute the tile types and signals needed at each location to build a phenotype which can serve as a seed supertile for an STAM* system $\mathcal{R}$ using a universal tile set $T$. At temperature $2$, $\mathcal{R}$ will deconstruct the input supertiles to create kinky genome assemblies. Each kinky genome assembly will then first create a copy of the linear genome, and then either continue to create copies of the linear genome, or initiate the growth of a new copy of the phenotype (which consumes the copy of the kinky genome). The new copies of the phenotype will become terminal assemblies, in the shape of $S^2$. The other terminal assemblies are junk assemblies of size $\le 2$ (during the reassembly process for cases 5 and 6, for certain next path directions, pairs of filler tiles will detach), and the linear genome assemblies are never terminal as each facilitates the growth of infinite new copies. Thus, $\mathcal{R}$ is a self-replicator for $S^2$ and since this works for arbitrary shapes at scale factor $2$, $T$ is a universal tile set for shape self-replication for the class of scale factor 2 shapes.

}

\section{Shape Building via Hierarchical Assembly}\label{sec:2ham}

In this section we present details of a shape building construction which makes use of hierarchical self-assembly. The main goals of this construction are to (1) provide more compact genomes than the previous constructions, and (2) to attempt to more closely mimic the hierarchical assembly that occurs in the replication of biological systems, e.g. individual proteins are independently constructed and then they combine with other proteins to form cellular structures. First, we define a class of shapes for which our base construction works, then we formally state our result.

Let a \emph{block-diffusable} shape be a shape $S$ which can be divided into a set of rectangular prism shaped blocks\footnote{A rectangular prism is simply a 3D shape that has 6 faces, all of which are rectangles.} whose union is $S$ (following the algorithm of Section~\ref{sec:2ham-decomp}) such that a connectivity tree $T$ can be constructed through those blocks and if any prism is removed but $T$ remains connected, that prism can be placed arbitrarily far away and move in an obstacle-free path back into its location in $S$.

\begin{theorem}\label{thm:2HAM}
There exists a tile set $U$ such that, for any block-diffusable shape $S$, there exists a scale factor $c \ge 1$ and STAM* system $\mathcal{T_S} = (U,\sigma_{S^c},2)$ such that $S^c$ self-assembles in $\calT_S$ with waste size 1. Furthermore, $|\sigma_S|=O(|S|^{1/3})$.
\end{theorem}

To prove Theorem~\ref{thm:2HAM}, we present the algorithm which computes the encoding of $S$ into seed assembly $\sigma_S$ as well as the value of the scale factor $c$ (which may simply be $1$), and then explain the tiles that make up $U$ so that $\mathcal{T_S}$ will produce components that hierarchically self-assemble to form a terminal assembly of shape $S^c$.
At a high level, in this construction the seed assembly is the \texttt{genome}, which is a compressed linear encoding of the target shape that is logically divided into separate regions (called \texttt{genes}), and each {\tt gene} independently initiates the growth a (potentially large) portion of the target shape called a \texttt{block}. Once sufficiently grown, each {\tt block} detaches from the \texttt{genome}, completes its growth, and freely diffuses until binding with the other {\tt blocks}, along carefully defined binding surfaces called {\tt interfaces}, to form the target shape.

It is important to note that there are many potential refinements to the construction we present which could serve to further optimize various aspects such as \texttt{genome} length, scale factor, tile complexity, etc., especially for specific categories of target shapes. For ease of understanding, we present a relatively simple version of the construction, and in several places we point out where such optimizations and/or tradeoffs could be made.
Throughout this section, we will refer to $S$ as the target shape of our system. Note that for some shapes, it may be the case that a scale factor $c>1$ is required for the input shape $S$ (and the details of how that is computed are provided in Section~\ref{sec:scale-and-interface}) but for simplicity we'll refer to the target shape as $S$ whether or not it is a scaled version. We will first describe how the shape $S$ can be broken into a set of constituent {\tt block}s, then how the {\tt interface}s between {\tt block}s are designed, then how individual {\tt block}s self-assemble before being freed to hierarchically combine into an assembly of shape $S$.

\begin{figure}
    \centering
    \begin{subfigure}{.45\linewidth}
        \includegraphics[width=0.95\textwidth]{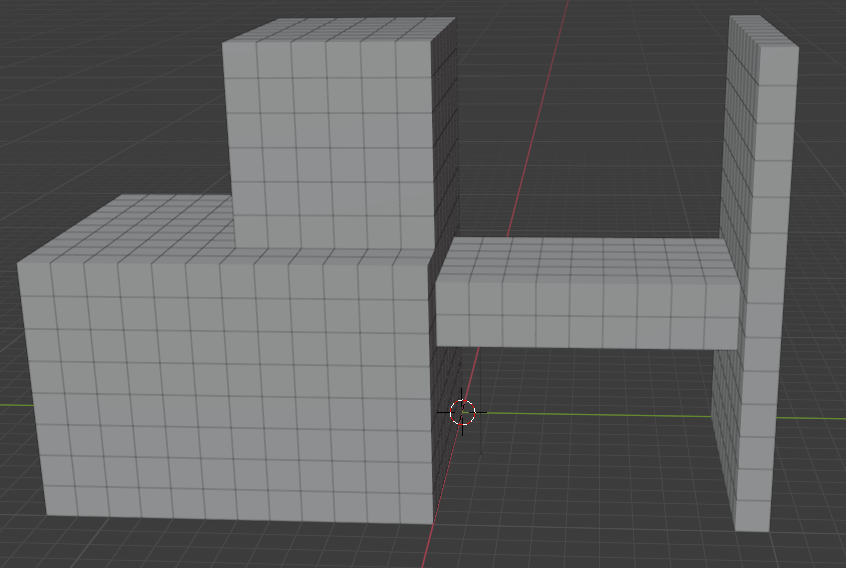}
        \caption{\label{fig:3D-structure}}
    \end{subfigure}
    \begin{subfigure}{.45\linewidth}
        \includegraphics[width=0.95\textwidth]{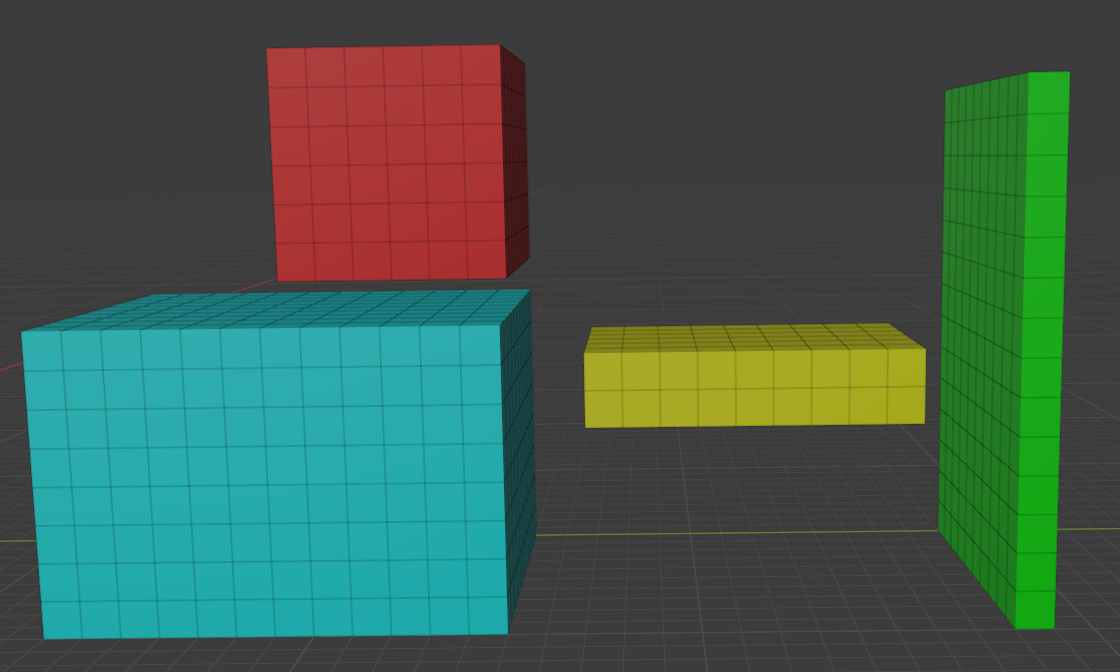}
        \caption{\label{fig:3D-structure-decomp}}
    \end{subfigure}
    \caption{(a) An example 3D shape $S$. (b) $S$ split into 4 {\tt block}s, each of which can be grown from its own \texttt{gene}. Note that the surfaces which will be adjacent when the {\tt block}s combine will also be assigned {\tt interface}s to ensure correct assembly of $S$.}
    \label{fig:2HAM-block-decomp}
    \vspace{-10pt}
\end{figure}

\vspace{-10pt}
\subsection{Decomposition into {\tt block}s}
\label{sec:2ham-decomp}
\vspace{-5pt}

Since $S$ is a shape in $\mathbb{Z}^3$, it is possible to split it into a set of rectangular prisms whose union is $S$. We do so using a simple greedy algorithm which seeks to maximize the size of each rectangular prism, which we call a {\tt block}, and we call the full set of {\tt block}s $B$.

After the application of a greedy algorithm to compute an initial set $B$, we refine it by splitting some of the {\tt block}s as needed to form a binding graph in the form of a tree $T$ such that every {\tt block} is connected to at least one adjacent {\tt block}, but also so that each {\tt block} has no more than one connected neighbor in each direction in $T$. This results in the final set of {\tt block}s that combine to define $S$, can join along the edges defined by $T$, and each {\tt block} has at most $6$ neighbors to which it combines. (Figure \ref{fig:2HAM-block-decomp} shows a simple example.)
\update{
    
    Note that for our shape-replicating construction to work for $S$, it also requires that $S$, once divided into rectangular prisms, is block-diffusable. Our algorithm does not ensure block-diffusability, and in fact, we conjecture that there exist shapes for which this is not possible without arbitrarily scaling the shapes.
    Below, we provide the algorithm which splits $S$ into a set of {\tt block}s.
    
    \begin{enumerate}
        \item Define $S' = S$.
        
        \item Initialize the set of {\tt blocks} $B = \emptyset$.
    
        \item Define the function $P$ so that on input $v \in S'$ (i.e. $v$ is a voxel in $S'$), $P(v)$ returns the largest (by volume) rectangular prism (as the set of coordinates contained within it) containing $v$ within $S'$.\label{step:loop1}
        
        \item Let $p_{max}$ be the largest rectangular prism (by volume) returned by $P$ for any $v \in S'$.
        
        \item Add $p_{max}$ as a {\tt block} to the set of {\tt block}s $B$, and remove the voxels of $p_{max}$ from $S'$. (Note that this may make $S'$ into a disconnected set of points, but that is okay.)
        
        \item If $S' \neq \emptyset$, return to step \ref{step:loop1}.
    \end{enumerate}
    
    We now have $B$ as a preliminary set of {\tt block}s, which we will modify as necessary to ensure that each {\tt block} has only one adjacent neighbor to which it will need to bind in each direction.
    
    \begin{enumerate}
        \item Define the graph $G$ such that for each $b \in B$, $G$ has a corresponding node, and there is an edge between each pair of nodes of $G$ that correspond to {\tt block}s that are adjacent to each other in $S$. \label{step:loop2}
        
        \item Generate a tree $T$ from graph $G$ by removing edges from each cycle until no cycles remain.
        
        \item For each $b \in B$, if there exist $b',b'' \in B$ where $b \ne b' \ne b'' \ne b$ such that $b$ is adjacent to both $b'$ and $b''$ along the same plane in $S$, and there are edges in $T$ (1) between the nodes representing $b$ and $b'$ and (2) the nodes representing $b$ and $b''$, then split $b$ into two new rectangular prisms, $b_1$ and $b_2$, such that each is adjacent to exactly one of $b'$ and $b''$ (this is always possible since all of $b,b',$ and $b''$ are rectangular prisms). \label{step:split}
        
        \item Remove $b$ from $B$ and add $b_1$ and $b_2$ to $B$.
        
        \item If any {\tt block} was split in step \ref{step:split}, loop back to step \ref{step:loop2}.
        
    \end{enumerate}
    
    The tree $T$ is a graph whose edges connect the nodes representing {\tt block}s which must bind to each other in the final assembly. At this point, each $b \in B$ will have at most 1 adjacent $b' \in B$ on each side to which it must bind, and each $b \in B$ will have at least one other $b' \in B$ to which it must bind. We will refer to any pair of {\tt block}s which must bind to each other as \emph{connected}.
}

\begin{figure}
    \centering
    \begin{subfigure}{\linewidth}
        \includegraphics[width=.9\textwidth]{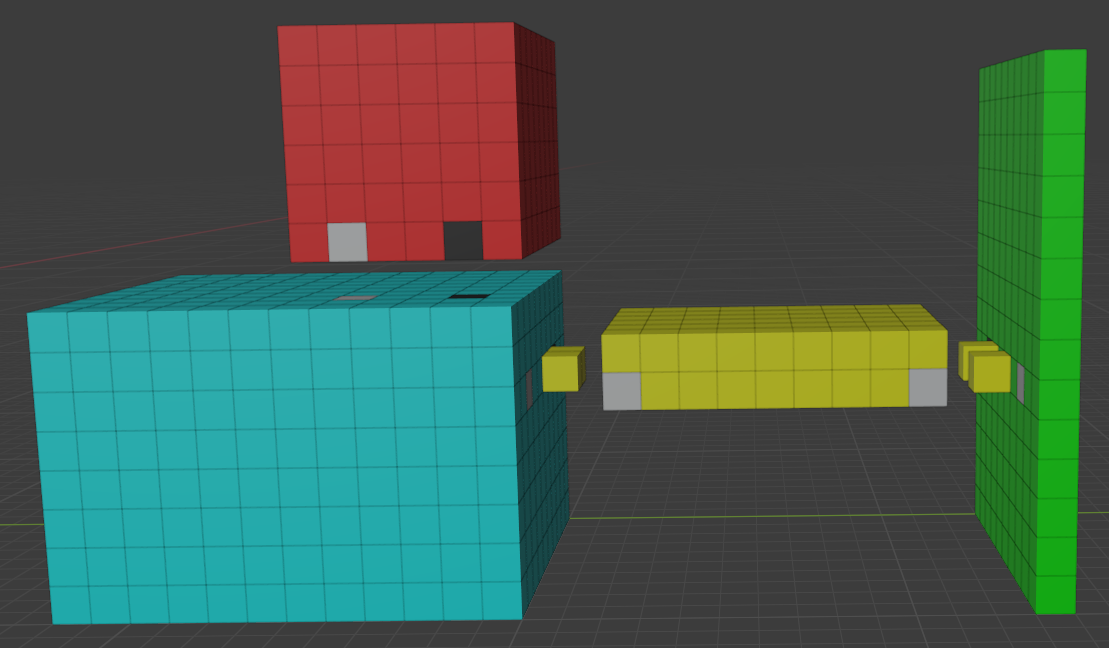}
        \caption{\label{fig:3D-structure-decomp-interfaces}}
    \end{subfigure}
    \hspace{0.06\textwidth}
    \begin{subfigure}{\linewidth}
        \includegraphics[width=0.9\textwidth]{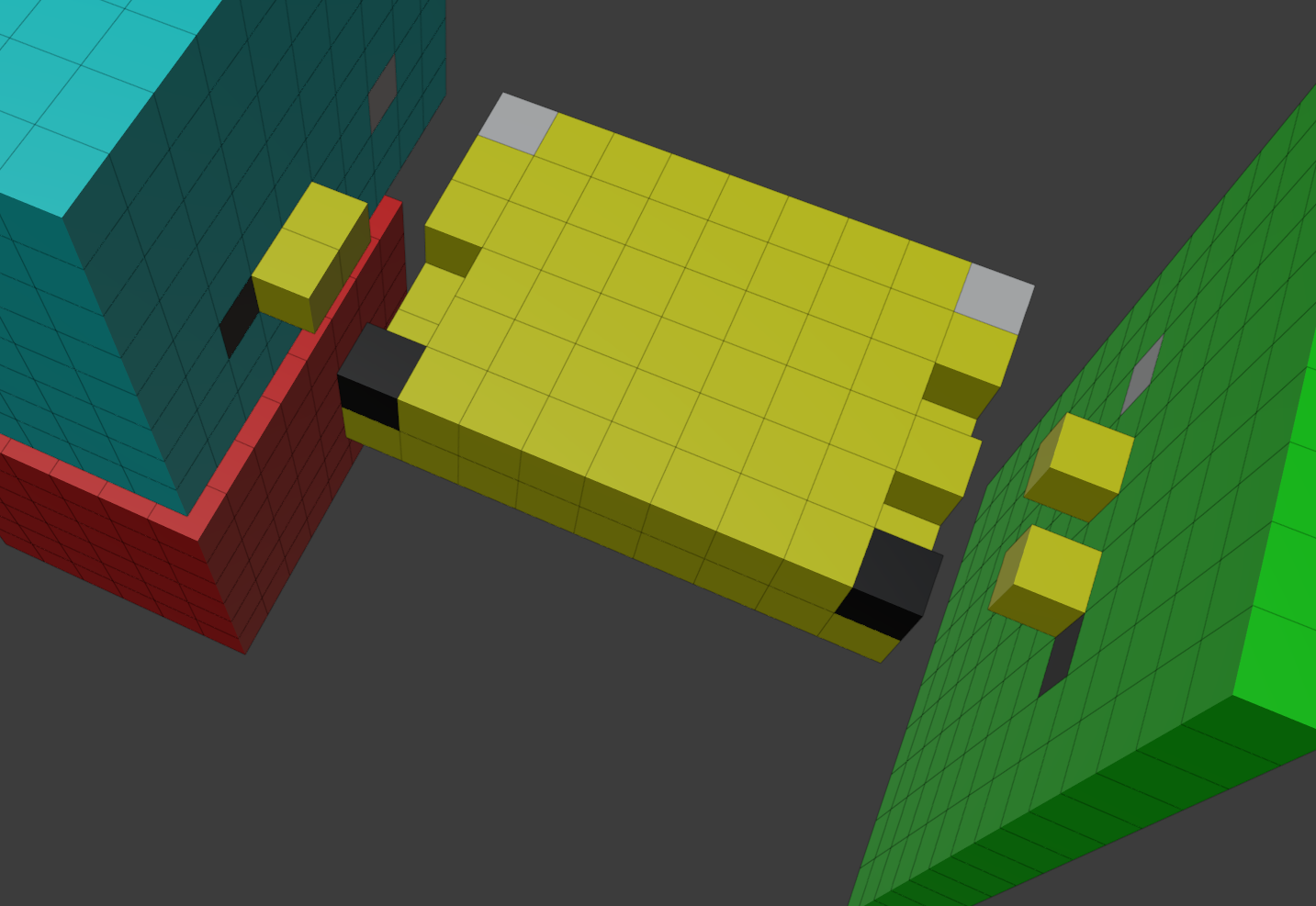}
        \caption{\label{fig:3D-structure-decomp-interfaces2}}
    \end{subfigure}
    \caption{(a) The {\tt block}s for the example shape $S$ from Figure \ref{fig:2HAM-block-decomp} with example {\tt interface}s included. (b) View from underneath showing more of the {\tt interface}s between {\tt block}s. Note that the actual {\tt interface}s created by the algorithm would be shorter, but to make the example more interesting their sizes have been increased.}
    \label{fig:2HAM-block-interfaces2}
    \vspace{-20pt}
\end{figure}

\vspace{-10pt}
\subsection{Scale factor and interface design}\label{sec:scale-and-interface}
\vspace{-5pt}

The {\tt block}s self-assemble individually, then separate from the {\tt genome} to freely diffuse until they combine together via {\tt interface}s along the surfaces between which there were edges in the binding tree $T$. Each {\tt interface} is assigned a unique length and number. The two {\tt block}s that join along a given {\tt interface} are assigned complementary patterns of ``bumps'' and ``dents'' and a pair of complementary glues on either side of those patterns (to provide the necessary binding strength between the blocks).

\update{
    We now describe the size and composition of the \\{\tt interface} between connected {\tt block}s.
    Each \texttt{interface} will include two specially designated glues, one on each end of the {\tt interface}, and assuming the length of the \texttt{interface} is $n$, an $n-2$ tile wide portion in between those glues which will eventually be mapped to a particular ``geometry'' of bumps and dents (i.e. tiles protruding from a surface, and openings for tiles in a surface). No \texttt{interface} can be shorter than $2$. Also, since each \texttt{interface} must be unique, there is only one valid \texttt{interface} of length $2$, and for each $n > 2$ there will be $2^{(n-2)/2}$ valid interfaces because each bit of the assigned number is represented by two bits in the geometry. For a 0-bit, the pattern $01$ is used, and for a 1-bit the pattern $10$ is used. This ensures that each geometry is compatible only with its complementary geometry (see \cite{GeoTiles} for further examples.) Figure \ref{fig:2HAM-block-interfaces2} shows an example of {\tt interface}s which could be added to the {\tt block}s of the example shape from Figure \ref{fig:2HAM-block-decomp}. Note, however, that for the sake of a more interesting example larger {\tt interface}s are shown than would be assigned by the algorithm presented, which would have created one {\tt interface} of size 2, with only White and Black glues, and two of size 4, one with a ``dent'' then ``bump'' to represent $01$ which maps to 0, and one with a ``bump'' then ``dent'' to represent $10$ which maps to 1.
    
    \begin{enumerate}
        \item Define the function $\texttt{RECT}$ such that, for each connected pair $b,b' \in B$,  $\texttt{RECT}(b,b')$ returns the rectangle along which $b$ and $b'$ are adjacent in $S$, and the function $\texttt{RECTMAX}(b,b') = \texttt{max}(m,n)$ where $m$ and $n$ are the lengths of the sides of the rectangle returned by $\texttt{RECT}(b,b')$ (i.e. it returns the length of the maximum dimension of the rectangle).
        
        \item Initialize the mapping $\texttt{INTERFACE-LENGTH}$ which maps a connected pair $b$ and $b'$ to an integer such that $\texttt{INTERFACE-LENGTH}(b,b')$ $ = 2$. (\texttt{INTERFACE-LENGTH} will eventually specify the length of the \texttt{interface} between {\tt block}s.)
        
        \item Define the function \texttt{COUNT} such that, for each $k > 1$, $\texttt{COUNT}(k)$ is equal to the number of connected pairs $b,b' \in B$ such that $\texttt{INTERFACE-LENGTH}(b,b')$ $ = k$. (That is, \texttt{COUNT} returns the number of pairs of {\tt block}s that are currently assigned interfaces of length $k$.)
        
        \item While there exists $k > 1$ such that $\texttt{COUNT}(k) > 2^{(k-2)/2}$:
        
        \begin{enumerate} 
            \item Select a connected pair $b,b'$ where \\$\texttt{INTERFACE-LENGTH}(b,b') = k$ and update the mapping $\texttt{INTERFACE-LENGTH}$ so that \\$\texttt{INTERFACE-LENGTH}(b,b') = k+1$.
        \end{enumerate}
        
        \item  If there exists a connected pair $b,b' \in B$ such that $\texttt{INTERFACE-LENGTH}(b,b') > \texttt{RECTMAX}(b,b')$, this (simplified) construction requires the shape $S$ to be scaled because there are too many {\tt interface}s of one or more lengths for them all to be unique\footnote{The number of unique {\tt interface}s for any length can easily be increased using methods discussed later.}. Therefore, replace $S$ with $S^2$ (the scaling of $S$ by $2$) and restart the construction from shape decomposition, at the beginning of Section~\ref{sec:2ham-decomp}.
        
    \end{enumerate}
    
    At this point, the mapping $\texttt{INTERFACE-LENGTH}$ defines a valid mapping of lengths to each \texttt{interface}. We now assign a valid geometric pattern (i.e. a series of ``bumps'' and ``dents'') to each.
    
    \begin{enumerate}
        \item Let $s$ equal the value of the maximum of the width, height, and depth of $S$ (i.e. the length of its greatest dimension).
        
        \item For each integer $1 < i \le s$, let $I_i = \{ (b,b') \mid $ where $b,b' \in B$ are connected and $\texttt{INTERFACE-LENGTH}(b,b') = i\}$. Thus, $I_i$ is the set of connected pairs of {\tt block}s which have {\tt interface}s of length $i$.
        
        \item For each $I_i$ where $|I_i| > 0$, assign an arbitrary, fixed ordering to $I_i$ and for $0 < |I_i| < j$, let $I_{i_j}$ be the $j$th connected pair in $I_i$.
        
        \item For each $I_{i_j}$:
        
        \begin{enumerate}
            \item Recall that $i$ is the assigned {\tt interface} length.
            
            \item Assign $j$ as the number assigned to the {\tt interface} (after the number of bits is doubled so that each 0-bit is represented by $01$ and each 1-bit by $10$).
            
            \item Let $(b,b') = I_{i_j}$ and $r = \texttt{RECT}(b,b')$
            
            \item As $r$ is a rectangle, it is 2-dimensional and has only two of width ($x$ dimension), height ($y$ dimension), and depth ($z$ dimension). If its width is $\ge i$, we call $r$ an East-West (EW) rectangle. Else, if its height is $\ge i$, we call $r$ a North-South (NS) rectangle. Otherwise, its depth must be $\ge i$ (by design of the algorithm determining the assigned value of $i$, it will fit in at least one dimension of $r$) and we call $r$ an Up-Down (UD) rectangle.
        
        
            \item Define $\texttt{RECT-ROW}$ as a function such that on input $b,b' \in B$, $\texttt{RECT-ROW}(b,b')$ returns a single row of coordinates as follows. Rectangle $r$ is either EW, NS, or UD and has one other non-zero dimension ($x$, $y$, or $z$) other than the dimension its type is named for. If that other non-zero dimension is $x$ (resp. $y$, resp. $z$), set direction $d = E$ (resp. $N$, resp. $U$).  If $\texttt{RECT}(b,b')$ returns EW (resp. NS, resp. UD) rectangle $r$, $\texttt{RECT-ROW}(b,b')$ returns the row furthest in direction $d$ which runs EW (resp. NS, resp UD) in $r$.
        
            \item Let $r' = \texttt{RECT-ROW}(b,b')$. If $r'$ is an EW (resp. NS, resp. UD) rectangle, we define the \texttt{interface} for $I_{i_j}$ such that the easternmost (resp. northernmost, resp. uppermost) location in $r'$ is assigned the Black glue, the adjacent $i-2$ locations are assigned the $i-2$ bits of the binary representation of the number $j$, in order, with the least significant bit in the easternmost (resp. northernmost, resp. uppermost) location, and the next location is assigned the White glue, making it the westernmost (resp. southernmost, resp. downwardmost) location containing a non-zero amount the {\tt interface} information. The other locations of the row of $r'$ are assigned ``empty'' values. Define the function $\texttt{INTERFACE}(b,b')$ such that it returns this \texttt{interface} definition for the entire row of $r'$ for the {\tt interface} between $b$ and $b'$. (Recall that by our construction, any connected pair can have at most one \texttt{interface}.)
        \end{enumerate}
        
    \end{enumerate}
}



\update{
    \vspace{-10pt}
    \subsection{Growth of a {\tt block}}\label{sec:block-growth}
    \vspace{-5pt}
    
    Each {\tt block} $b \in B$ making up shape $S$ has at most $6$ {\tt interface}s. Because of this constant bound, and the fact that each {\tt block} is a rectangular prism, it is possible to encode all of the information needed to grow an entire {\tt block} $b$ within a sequence of glues, taken from a set of glues that is constant over any shape $S$, that is no longer than the longest dimension of $b$.\footnote{Later we will also briefly mention ways in which the length can actually be as small as the $\log$ of the longest dimension.} We call each such sequence a {\tt gene}. In this section we show how a {\tt gene} can be encoded and how a {\tt block} can then grow from it.
    
    \begin{figure}
        \centering
        \includegraphics[width=\linewidth]{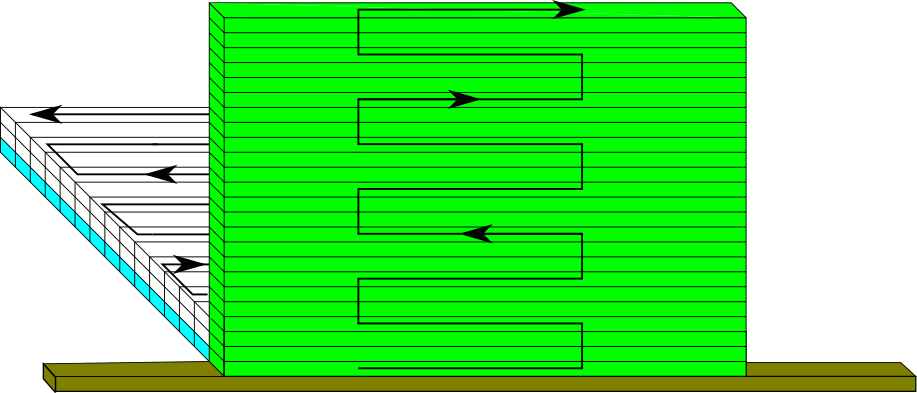}
        \caption{Schematic representation of the order of block growth (without directions shown for every row). Starting from a \texttt{gene} section, the green surface grows upward in a zig-zag pattern. As each row of the green face completes, one plane can grow perpendicularly to it (the first is shown in blue, with the next two in white). Each of these also grows in a zig-zag pattern away from the green face.}
        \label{fig:side-growth}
        \vspace{-10pt}
    \end{figure}
    
    Each {\tt block} grows so that one of its $6$ faces grows directly upward off of the {\tt block}'s {\tt gene}. The growth of this plane happens in a zig-zag manner, meaning that the first row grows completely from left to right (zigging), then the second from right to left (zagging), and the pattern continues until the growth terminates. (Shown schematically in green in Figure~\ref{fig:side-growth}.) The zig-zag pattern of growth allows for each row to transmit (and update) information it reads from the row below it (to be discussed shortly).
    
    As each row of the first face completes, a plane growing perpendicular to the first face can begin its growth. (The first such plane is shown in light blue in Figure~\ref{fig:side-growth}, and the next two in white.) Every row of each such plane also grows in a zig-zag manner, which allows information to be transmitted from the green initiating rows throughout each plane.
    
    To control the size of each plane, a pair of binary counters are used. The upward facing glues of the {\tt gene} encode a series of bits (which we will call the \emph{green} bits). As the face grows upward, every other row increments the value of the binary number represented by the bits, and every other row checks to see if all bits are equal to 1. If they are all equal to 1, upward growth terminates. An example can be seen in Figure~\ref{fig:counting}.
    
    \begin{figure}
        \centering
        \begin{subfigure}{\linewidth}
            \includegraphics[width=.9\textwidth]{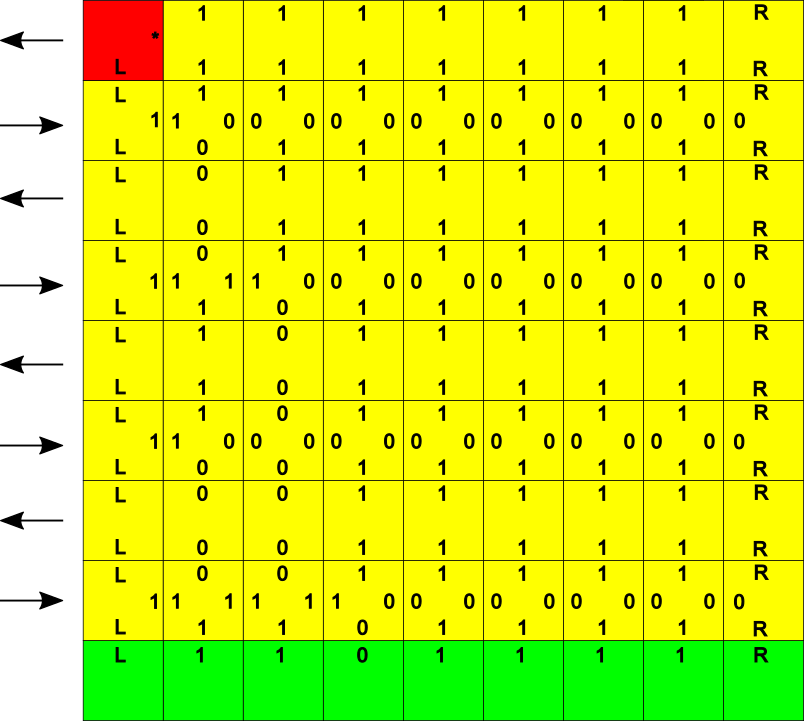}
            \caption{\label{fig:counting}}
        \end{subfigure}
        \hspace{0.06\textwidth}
        \begin{subfigure}{\linewidth}
            \includegraphics[width=.9\textwidth]{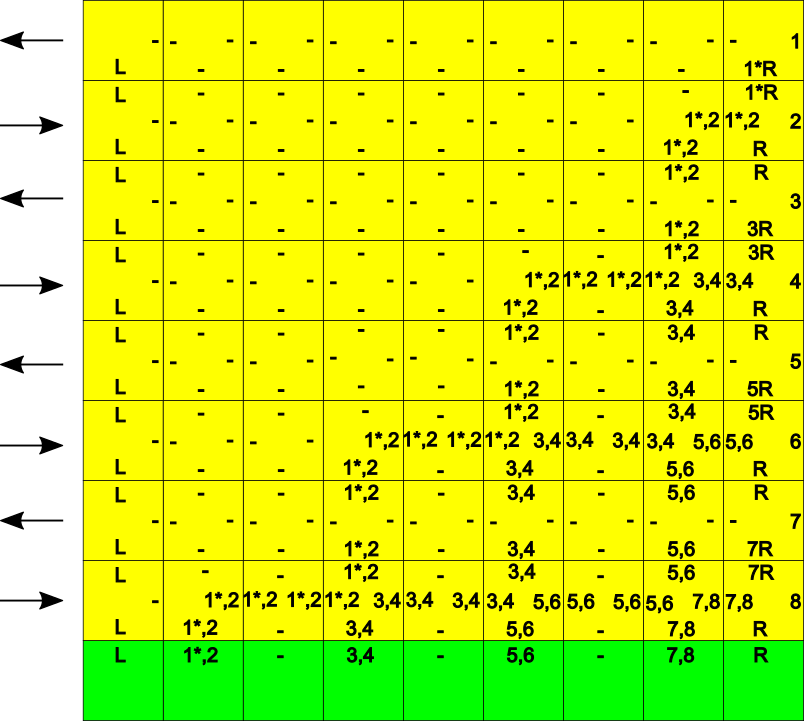}
            \caption{\label{fig:pattern-rotation}}
        \end{subfigure}
        \caption{Examples of basic growth patterns in {\tt blocks} in the hierarchical construction. (a) Example of a binary counter which increments every other row, and checks for all 1s on the others. Note that the least significant bit is on the left. When the ``checking'' row detects all 1s, a red tile marks the end of further upward growth. (b) A basic example of the rotation of a pattern.}
        \label{fig:2HAM-block-growth-details}
        \vspace{-10pt}
    \end{figure}
    
    \begin{figure}
        \centering
        \includegraphics[width=0.2\textwidth]{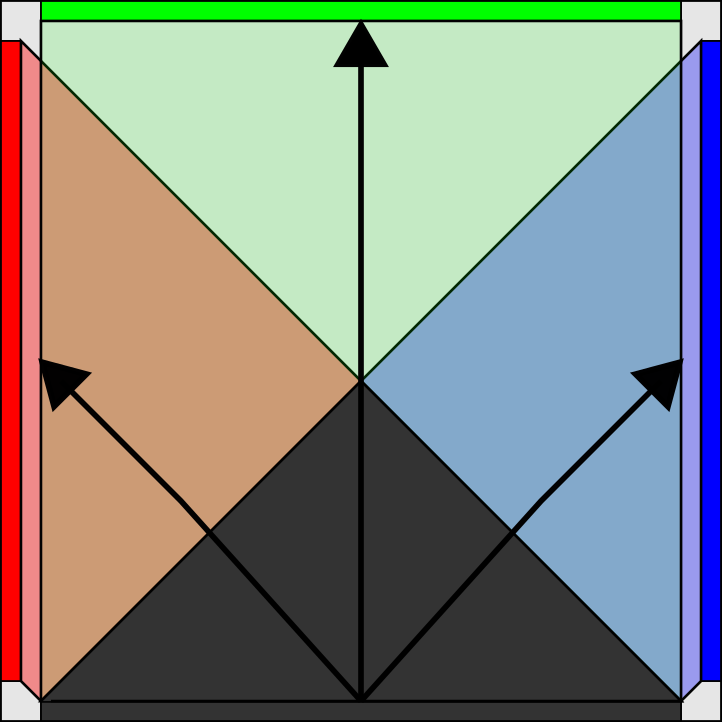}
        \caption{Schematic of how information can be transmitted to three sides of one plane of a {\tt block}, with rotations of alternating rows (as shown in Figure~\ref{fig:pattern-rotation}) rotating information to the left and right sides. \label{fig:3-sides}}
        \vspace{-10pt}
    \end{figure}
    
    We will call the bits of the counter which control the length of the perpendicular planes (shown as blue and white in Figure~\ref{fig:side-growth}) the \emph{blue} bits. These bits are also encoded in the upward facing glues of the {\tt gene} (i.e. each glue can encode both a green and a blue bit by making $4$ glues, one for each pair of bit values $00$, $01$, $10$, and $11$). However, as each row of the green face assembles, rather than using the blue bits to count, each row presents the blue bits on both its upward and backward facing glues. This allows them to be propagated up throughout the green face, unchanged, and to control the distance grown by each perpendicular plane, which uses them as the bits for its counter.
    
    With the {\tt gene}'s length implicitly encoding the size of one dimension of the growing {\tt block}, and the green and blue counter bits controlling the sizes of the other two dimensions, the {\tt block} grows into a rectangular prism of the correct dimensions. (Note that growing counters, zig-zag growth, rotating bits, etc. are very standard techniques in tile assembly literature - see \cite{IUSA,j2HAMIU,Versus,SolWin07,RotWin00} for just some examples - and issues like growing sides of odd length, despite the zig-zag pattern, are easily handled with a few extra glues that signal for one additional row to grow.)
    
    Each {\tt block} has a fixed orientation relative to the others when they are attached together to form the shape $S$, and since we (arbitrarily) assign each shape a canonical translation and rotation, each {\tt block} has a canonical orientation which allows us to refer to its sides by the directions they face in that orientation. Throughout, we talk about {\tt block}s in term of this orientation, irrespective of the orientation in which they grow.
    
    This (simplified version of the) construction has each {\tt gene} equal to the length of the longest dimension of the {\tt block} it initiates. This could lead to the first surface to grow being any of at least $4$ sides, so without lack of generality we fix a preferred ordering as: North, East, South, West, Up, Down. Therefore, of the multiple faces which share the longest dimension, that appearing first in the ordering grows ``first'' (i.e. as the green face, as shown in Figure~\ref{fig:side-growth}), and with the side attached to the {\tt gene} being that whose coordinates are the smallest along the direction of upward growth of the first face.
    
    \subsection{{\tt interface} growth}
    
    With the dimensions of each {\tt block} correctly controlled, the next thing to ensure is correct growth of the {\tt block}'s {\tt interface}s. As previously mentioned, there are at most $6$ of these (no more than one per side), and each \\{\tt interface} consists of two outward facing glues (Black and White) with a possible series of ``bumps'' and ``dents'' between them, geometrically encoding the bits of the number which is uniquely assigned to that {\tt interface}. If the {\tt interface} is on the North, East, or Up side, in the location of each bit $b = 1$ there is a tile which extends from the side as a ``bump'', and in the location of each bit $b = 0$, there is no such bump. If the {\tt interface} is on the South, West, or Down side, in the location of each bit $b = 1$ there is an empty tile location (i.e. a ``dent''), and in the location of each bit $b = 0$, there is no such dent. (See Figure \ref{fig:2HAM-block-interfaces2} for examples of {\tt interfaces} with ``bumps'' and ``dents''.)
    
    The information defining each {\tt interface} can be encoded as a series of glues representing the locations of the Black and White {\tt interface} glues plus each of the bits of the assigned {\tt interface} number, as well as the information about whether the $1$-bits are encoded as ``bumps'' or ``dents'' for the particular surface. Using the same technique as mentioned previously for adding information about an extra bit to the glues extending from the {\tt gene}, we can similarly add the information which defines each of the (up to $6$) {\tt interface}s of a {\tt block}. Therefore, we individually discuss the patterns by which the information specifying each {\tt interface} is propagated into the correct locations, and note that all of that information can be encoded in the outward facing glues of the {\tt gene} and then distributed to the proper locations in the {\tt block} during the growth process previously described. Once we've explained how the information about each {\tt interface} arrives at the correct location, we will discuss the tiles which encode it.
    
    There are $6$ sides, and for each side $2$ orientations which must be considered for the possible {\tt interface} on that side (note that on {\tt block} sides which don't have {\tt interface}s, nothing needs to be done beyond the growth of the side to the correct dimensions as previously described). One orientation we will refer to as ``parallel'' to the {\tt gene}, and the other as ``perpendicular'' (although these terms aren't technically accurate for all cases). The parallel cases are depicted in Figure~\ref{fig:parallel-interfaces}, and the perpendicular cases are depicted in Figure~\ref{fig:perpendicular-interfaces}.
    
    \begin{figure*}[ht]
        \centering
        \includegraphics[width=0.95\textwidth]{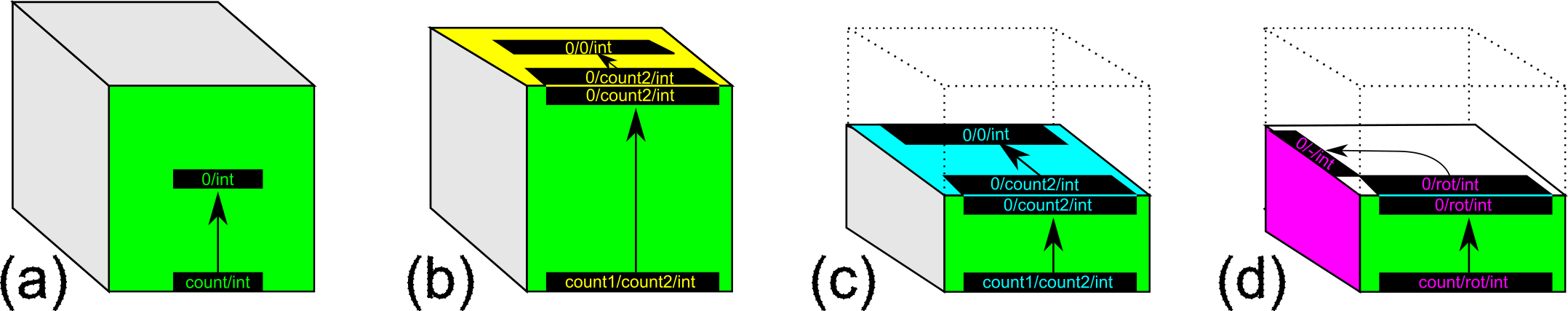}
        \caption{Schematic representation of the patterns by which {\tt interface} information is propagated into the correct positions for ``parallel'' {\tt interface}s. (a) A counter is used to determine the correct height for the {\tt interface} on the green side, (b) Two counters are used to position the {\tt interface} on the yellow side. The first counts to the top of the green side, then the bits of the second and the {\tt interface} are rotated onto the yellow colored plane and the second counts to the proper location for the {\tt interface} on that side. (c) Two counters are used to position the {\tt interface} on the back side of the {\tt block}. The first counts to the correct height, then the bits of the second and the {\tt interface} are rotated onto the blue colored plane and the second counter counts the distance to the back surface. (d) To position the {\tt interface} on the pink side, a counter first counts to the correct height, then the bits are rotated to the pink face during the outward growth of the white plane. Note that the side opposite the pink {\tt interface} is positioned analogously but with an opposite rotation, and the bottom {\tt interface} is positioned similar to the top (yellow) but without the necessity of the first counter.}
        \label{fig:parallel-interfaces}
        \vspace{-5pt}
    \end{figure*}
    
    \begin{figure*}[ht]
        \centering
        \includegraphics[width=0.95\textwidth]{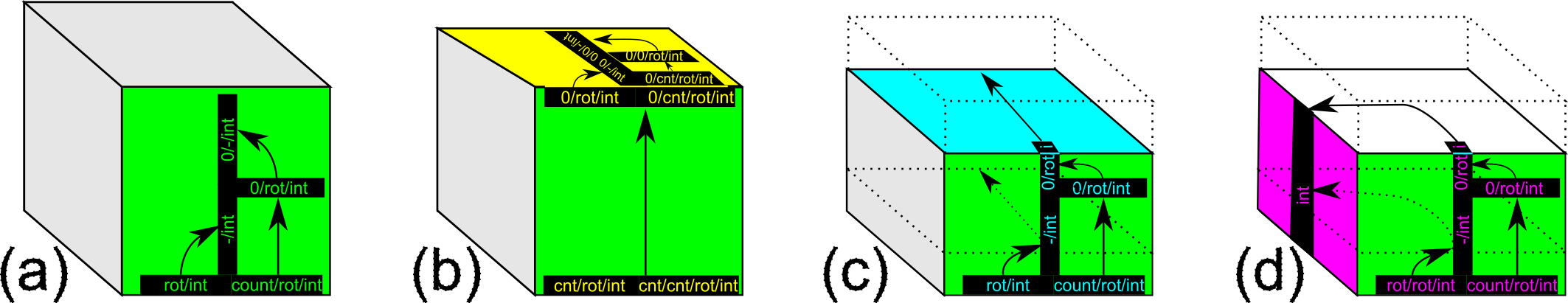}
        \caption{Schematic representation of the patterns by which {\tt interface} information is propagated into the correct positions for ``perpendicular'' {\tt interface}s. (a) For a perpendicular {\tt interface} on the green side, the information is split into two halves (depending on the actual length and position of the {\tt interface}). The left half is rotated upward immediately, and the second half has a counter which first moves it upward to the halfway point, and then it is rotated. Note that any offset from the center can be accommodated by shifting the location of the split and the height of the counter. If the {\tt interface} needs to be completely to the left or right, only one rotation is needed, and no splitting of the information or counting is needed. (b) The positioning of the {\tt interface} on the top is the same as for the green side, but a counter first propagates all information to the top, where it is rotated to the yellow surface. (And the same holds for the bottom surface but without the initial counter.) (c) To position the {\tt interface} on the back surface, the same rotations and counting are used as for the green surface. However, then the information from each row is carried all the way to the back surface following the counter which dictates that distance. (d) To position the {\tt interface} for the pink surface, again the same rotations and counting are used to align the information on the green surface, but then the information of each row is rotated to the pink surface as its plane grows away from the green surface. The surface opposite the pink is handled similarly, but with an opposite rotation.}
        \label{fig:perpendicular-interfaces}
        \vspace{-10pt}
    \end{figure*}
    
    It is important to note that the patterns shown in Figures~\ref{fig:parallel-interfaces} and \ref{fig:perpendicular-interfaces} suffice when each {\tt interface} is anywhere from the minimum allowed size (i.e. $2$) up to the maximum size, which is the full length of the side on which it is located. This is because the construction is designed so that the length of the {\tt gene}, and thus the green side, is the length of the longest dimension of the {\tt block}. Thus, there is room for the information in a longest-possible {\tt interface} to be correctly positioned, and shorter {\tt interfaces} can also be correctly positioned by correctly shifting the locations of information in the {\tt gene} so that the counters and rotations will propagate it correctly. Additionally, Figures~\ref{fig:parallel-interfaces} and \ref{fig:perpendicular-interfaces} depict the cases where each {\tt interface} is in the center of its surface, but any position along each surface can be accommodated by simply adjusting initial information alignment along the {\tt gene}, counter values, and/or the location of splits between rotations and counting.
    
    Recall that the {\tt block}s on either side of an {\tt interface} have complementary geometries, i.e. one has ``bumps'' in the $1$-bit locations and the other has ``dents''. Once the information encoding an {\tt interface} reaches the correct location on the correct surface, the locations assigned the Black and White glues of the {\tt interface} receive tiles which have strength-1 glues of those types exposed on the exterior of the {\tt block} for the {\tt block} with a bump {\tt interface}, and the {\tt block} with the dent {\tt interface} receives tiles which expose the complements of those glues (i.e. Black$^*$ and White$^*$, respectively). Additionally, in $1$-bit positions for a {\tt block} with a bump {\tt interface}, tiles attach which have strength-2 glues exposed, allowing the ``bump'' tiles to attach, and signals ensure that all ``bump'' tiles have attached before the Black tile can attach and enable the {\tt interface} to bind to its counterpart. See Figure \ref{fig:bump-tiles} for details of the signals.
    
    \begin{figure}[ht]
        \centering
        \includegraphics[width=0.4\textwidth]{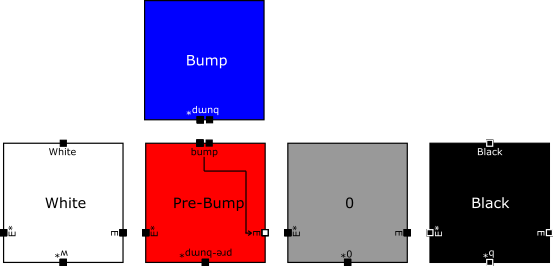}
        \caption{Templates for tile types which make up a ``bump'' {\tt interface} that extends in the same direction as the growth direction of the plane. These tile types are used for an {\tt interface} that grows from left to right, so each ``Pre-bump'' tile stalls the growth of the {\tt interface} until all of its associated ``Bump'' tiles have attached, ensuring that all ``bumps'' are in place before the Black tile is in place, since that allows the {\tt block} to attach to a {\tt block} with the complementary half of the {\tt interface}.}
        \label{fig:bump-tiles}
        \vspace{-10pt}
    \end{figure}
    
    \subsubsection{Formation of ``bumps'' and ``dents'', and detachment of {\tt block}s}
    
    The first portion of this construction which requires signals are the formations of the {\tt interface}s and the detachment of {\tt block}s from the {\tt genome}. (To make the exposition easier to understand, we will wait until Section \ref{sec:block-combo} to discuss a layer of signals which overlay each {\tt block}.) These situations are relatively straightforward to handle with signals, and in this section we provide an overview of the various cases and how they are handled. We also provide depictions of templates for the required signal tile types, which abstract away the fact that a variety of additional information (useful for the further growth of a {\tt block}) may be propagated through the signal tiles using standard glues. This is simply handled by making a set of signal tile types for each template provided, with a unique tile type for each glue type which needs to pass additional information through it. For each of the provided sets of templates for signal tile types, there are tile sets generated for each of the various permutations of such glues, as well as the orientations and locations of {\tt interface}s. However, the number of permutations and thus signal tile types is a constant, irrespective of the shape $S$. Which set is to be used for each {\tt interface} is encoded along with the definition of the {\tt interface} in the corresponding {\tt gene}.
    
    The general scenarios we will address are: (1) the growth of an {\tt interface} at the terminal edge of the plane of a {\tt block}, (2) the growth of an {\tt interface} in the middle of a plane of a {\tt block},  and (3) the detachment of a {\tt block} from the {\tt genome}, both when the attached row does and does not need to also encode an {\tt interface}.
    
    \paragraph{The growth of an {\tt interface} at the edge of the plane of a {\tt block}}
    \label{sec:end-of-plane}
    
    If the ``dents'' of an {\tt interface} appear in the row (or column) immediately adjacent to the edge of the plane in which the {\tt interface} information is being propagated, the White, Black, and $0$-bit tiles (i.e. those of the {\tt interface} not corresponding to ``dent'' locations) are attached with strength-2 glues to the row (or column) preceding the final row (or column).
    No tiles are ever placed in the locations of the ``dents'', and therefore no signals are required to allow tiles to detach, this basic scenario is handled without any signals.
    
    The growth of the ``bumps'' of an {\tt interface} requires that all ``bump'' tiles are in place before the Black and White glues are in place and active. 
    Otherwise, the {\tt interface} could be missing ``bumps'' and allow incorrect binding of the {\tt block} with another {\tt block}. The templates for the necessary tiles are shown in Figure~\ref{fig:bump-tiles}.
    Note that any ``Pre-Bump'' tiles will not activate the glue needed for the next tile of the {\tt interface} to attach until its ``Bump'' tile has attached. Therefore, only if all ``Bump'' tiles have attached throughout the {\tt interface}, will the {\tt interface} be able to grow to the point of the ``Black'' tile attaching.
    (Note that this assumes growth from the ``White'' side of the {\tt interface} to the ``Black'', and appropriately modified tiles exist for {\tt interface}s growing in the opposite direction.) The same tile types are used whether the {\tt interface} is at the edge of a plane (in the growth direction) or in the middle, except that in the middle of a plane the ``bumps'' must be going either into, or out of, the plane and thus the ``bump'' glue is either on the side of the tile facing into the page, or that facing outward. The directions for the White and Black glues are the same.
    
    \paragraph{The growth of an {\tt interface} in the middle of a plane of a {\tt block}}
    
    If the ``dents'' of an {\tt interface} appear in any row (or column) before the end of the plane, in the direction in which the {\tt interface} information is being propagated (i.e. in the middle of the plane), then tiles will initially be placed in the ``dent'' locations in order to allow information to be propagated through those locations, but signals will eventually cause them to deactivate glues and dissociate.
    The templates for the tile types of an {\tt interface} if it is positioned in such a location are shown in Figure~\ref{fig:dent-tiles}, and a an example sequence showing a portion of a growing plane and {\tt interface} which places, then eventually loses, tiles in the ``dent'' locations can be seen in Figure~\ref{fig:dent-tiles-sequence}. 
    
    The case in which the ``bumps'' of an {\tt interface} appear in the middle of a plane is (nearly) identical to the case in Section \ref{sec:end-of-plane} and was discussed there.
    
    \begin{figure}[ht]
        \centering
        \includegraphics[width=0.45\textwidth]{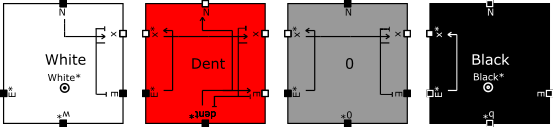}
        \caption{Templates for tile types which make up a ``dent'' {\tt interface} that is not at the edge of the plane in which it is growing. These tile types are used for an {\tt interface} that grows from left to right. Signal propagation which initiates tile detachment begins once a tile has attached to the north of a White tile, since the zig-zag growth pattern means that a row must have completed growth to the north of the {\tt interface}, growing right to left. This allows the ``dent'' tiles to propagate any needed information via their northern glues before they dissociate. An example growth sequence can be seen in Figure~\ref{fig:dent-tiles-sequence}.}
        \label{fig:dent-tiles}
        \vspace{-10pt}
    \end{figure}
    
    \begin{figure}[ht!]
        \centering
        \includegraphics[width=0.45\textwidth]{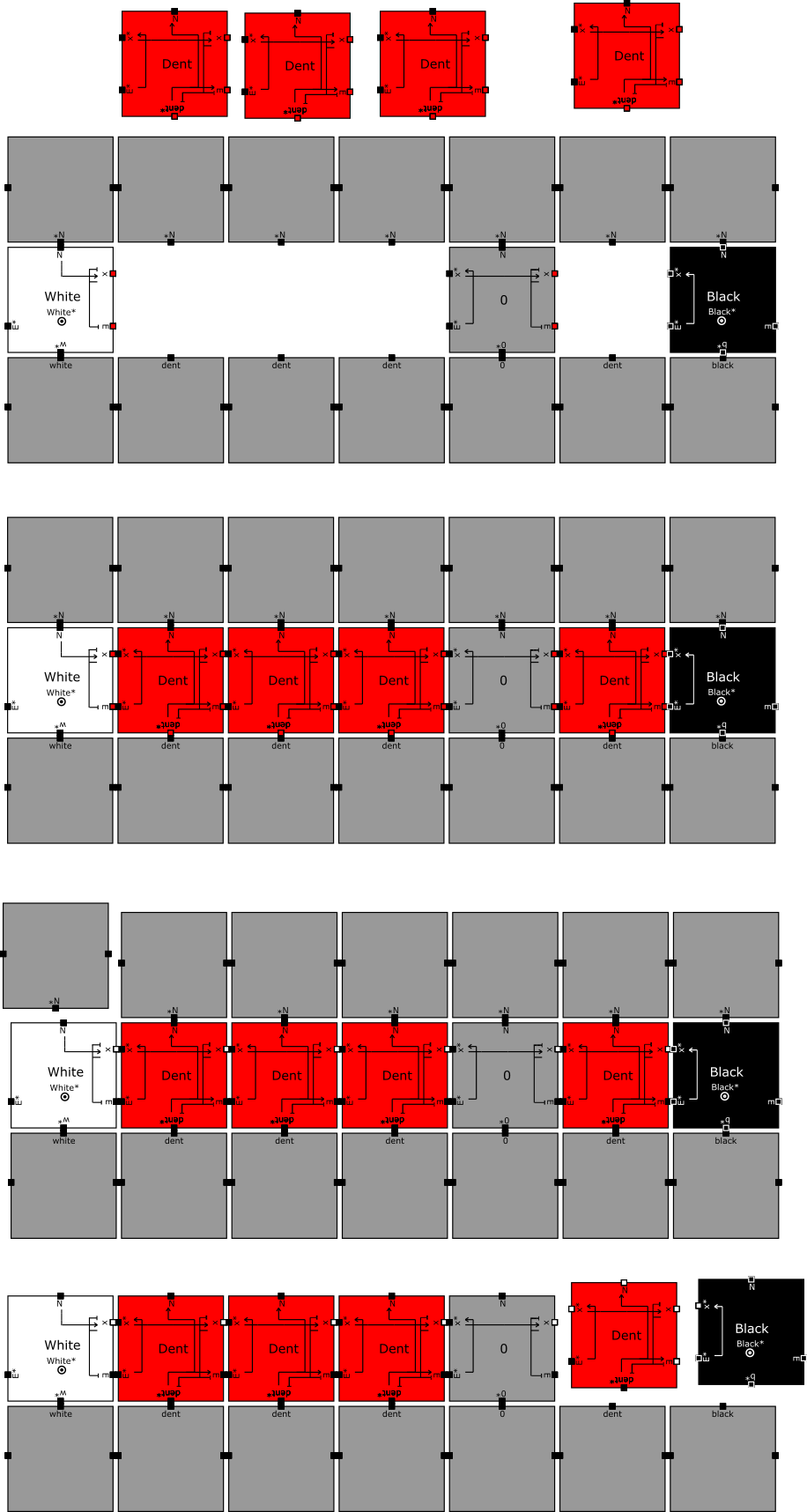}
        \caption{From bottom to top, example sequence of growth around ``dent'' tiles and their dissociation, using tiles of the templates shown in Figure~\ref{fig:dent-tiles}.}
        \label{fig:dent-tiles-sequence}
        \vspace{-10pt}
    \end{figure}
    
    \paragraph{The detachment of a {\tt block} from the {\tt genome}}
    
    We now discuss the detachment of a {\tt block} from the {\tt genome}, which we break into 3 sub-cases: (1) the attached row does not include an {\tt interface}, (2) the attached row encodes an {\tt interface} with ``dents'', and (3) the attached row encodes an {\tt interface} with ``bumps''.
    
    The templates for the tile types used in Case 1 are shown in Figure \ref{fig:block-detach-tiles}. The completion of the first row, consisting of one Left tile, one Right tile, and perhaps many Mid tiles, causes the ``x''glues to bind the tiles of the row together with strength 2 and to deactivate their glues attached to the {\tt genome} to the south. The stable assembly containing the complete first row can detach at any point during the continued growth of the {\tt block} and it will correctly complete. Also, once complete it will be free to bind with {\tt blocks} that have complementary {\tt interface}s.
    
    \begin{figure}[ht]
        \centering
        \includegraphics[width=0.35\textwidth]{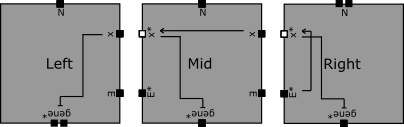}
        \caption{Templates for the tile types which make up the first row of a {\tt block} and allow the {\tt block} to detach once the first row has completed. It is assumed that the first row of a {\tt block} always grows from left to right. The ``x'' glues ensure that all tiles are bound with strength 2, and thus is it safe for the row to detach from the {\tt genome} as soon as it completes. The rest of the {\tt block} correctly grows from this row.}
        \label{fig:block-detach-tiles}
        \vspace{-10pt}
    \end{figure}
    
    The templates for the tile types used in Case 2 are shown in Figure \ref{fig:detach-dent-tiles}. In this case, a special tile type is also needed to bind to the north of the White tile to ensure that they are bound with strength 2, since it's possible for the White tile to be the leftmost and to have a ``dent'' location next to it, meaning it will end up attached only to the tile to its north. Essentially, the tile types ensure that the row to the north completes before any detachments. Then, signal propagation from left to right causes Dent tiles to dissociate, all tiles to deactivate the glues binding them to the {\tt genome}, and 0 and Black tiles to bind to their northern neighbors with strength 2 to ensure the stability of the assembly.
    
    \begin{figure}[ht]
        \centering
        \includegraphics[width=0.45\textwidth]{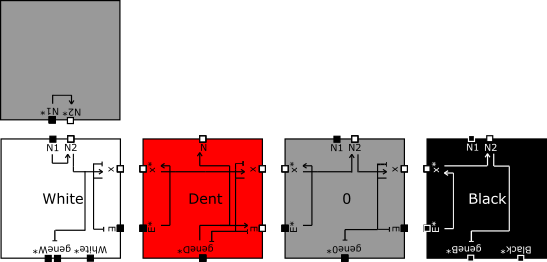}
        \caption{Templates for the tile types which make up the first row of a {\tt block} that allow the {\tt block} to detach once the first row has completed, when that first row also has to create an {\tt interface} with ``dents''. Note that the tiles shown here are for the scenario in which the {\tt interface} extends for the entire length of the row, since this is a slightly more complicated case due to the necessity of $\tau$-stability of the White and Black tiles. If the {\tt interface} only occupies a portion of the row, a few trivial edits are made to the template for the necessary tile types.}
        \label{fig:detach-dent-tiles}
        \vspace{-10pt}
    \end{figure}
    
    The templates for the tile types used in Case 3 are shown in Figure \ref{fig:detach-dent-tiles}. In this case, the glues attaching the row to the {\tt genome} are detached via the propagation of an ``x'' signal which also ensures that tiles are bound with strength 2, and a ``w'' signal ensures that all Bump tiles are attached before the White glue can be activated, making sure the {\tt interface} is correctly completed before it is able to bind to another {\tt block}.
    
    \begin{figure}[ht]
        \centering
        \includegraphics[width=0.45\textwidth]{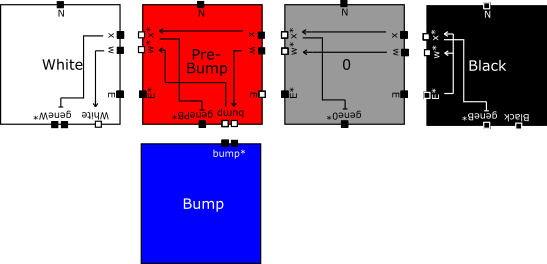}
        \caption{Tiles which make up the first row of a {\tt block} which allow the {\tt block} to detach once the first row has completed, when that first row also has to create an {\tt interface} with ``bumps''. (Note that the tiles shown here are for the scenario in which the {\tt interface} extends for the entire length of the row. If the {\tt interface} only occupies a portion of the row, a few trivial edits are made to the template for the necessary tile types.) Once the row completes with the attachment of a Black tile, the propagation of the ``x'' signal initiates the deactivation of glues attached to the {\tt genome} and the ``w'' signal ensures that all Bump tiles have attached before the White glue is activated. The signals also ensure the $\tau$-stability of the row so it can safely detach.}
        \label{fig:detach-bump-tiles}
        \vspace{-10pt}
    \end{figure}
    
    
    \subsection{Combination of {\tt block}s to form the target shape}\label{sec:block-combo}
    Once a {\tt block} has detached from its {\tt gene}, it is a freely floating supertile which may or may not require additional tile attachments to complete its own growth. However, only {\tt interface}s that have completed are able to bind with strength $2$ to the complementary {\tt interface}s of other {\tt block}s. Additionally, we now discuss a set of signals that allow for a {\tt block} to determine when all tiles have attached. The growth of each plane in a {\tt block} follows the same zig-zag pattern so that the final tile placed in each plane (other than possibly ``bump'' tiles of {\tt interface}s) falls into a single vertical column. These tiles are augmented with signals such that when the final tile of the bottommost plane attaches, it activates a glue that allows it to bind to the tile above it (whose complementary glue will be activated when it attaches). The tile above it in turn passes this signal upward, with each in the column doing the same until the final tile of the top plane is reached. Once that tile (which is of a special type) is placed, it is guaranteed that all tiles of all planes (other than possibly ``bump'' tiles of {\tt interface}s) have attached since each plane signals its completion in order from bottom to top.
    
    Upon receiving the ``completion'' signal, the final tile of the top plane then sends that signal outward, spreading across all tiles on all $6$ surfaces of the {\tt block}. These ``surface'' tiles are all equipped with signals that allow them to receive and pass on this completion signal (and during the growth of the {\tt block} it is always known which tiles will be on a surface since they are at an edge of their plane of growth). The previous description of the signals which activate the Black and White glues (and their complements) on {\tt interface}s was slightly simplified to omit this final detail: the previously described signals which activate those glues actually activate glues facing neighboring tiles so that only at that point they are able to receive the completion signal. It is the reception of this signal which actually activates the Black and White (and Black* and White*) glues on the {\tt interface}s.
    
    The addition of the extra layer of ``completion'' signals ensures that only a {\tt block} that has received all of the tiles of its body can have active {\tt interface}s. Once an {\tt interface} is active and able to bind to the complementary {\tt interface} of another {\tt block}, the {\tt block} combines to a growing supertile consisting of the {\tt block}s forming an assembly of shape $S$. Furthermore, by the definition of a block-diffusable shape and the fact that $S$ is such a shape, it is always possible for a free {\tt block} to attach as needed in any such growing supertile. Thus, the {\tt block}s will eventually form completed, and terminal, assemblies of shape $S$.
    }

\vspace{-10pt}
\subsection{Overview of the hierarchical construction}
\vspace{-5pt}

We have described how we can begin with an arbitrary block-diffusable 3D shape $S$, decompose it into rectangular prisms called {\tt block}s with complementary {\tt interface}s between them, encode the information needed to make each block into a {\tt gene} subassembly of a {\tt genome} seed assembly, and how the {\tt block}s can independently grow, detach from the {\tt genome}, and attach to each other to form an assembly of the target shape $S$ (or a scaled version if needed). By the design of the {\tt interface}s, the {\tt block}s can only combine in the correct manner. 
Once a {\tt block} is freely diffusing and complete, it can combine along its {\tt interface}s with the {\tt block}s that have complementary {\tt interface}s since, due to the fact that $S$ is a block-diffusable shape, free {\tt block}s can always diffuse into the proper locations to form the complete shape. We've described a tile set $U$ that can be used to (1) form the linear seed assembly $\sigma_S$ , and (2) to self-assemble the {\tt block}s which correctly combine to form the target assembly.  The STAM* system $\mathcal{T_S} = (U, \sigma_S, 2)$ will produce an infinite number of copies of terminal assemblies of shape $S$ (properly scaled if necessary). The only fuel (a.k.a. consumed, junk assemblies) will be singleton Dent tiles that attached during {\tt block} growth then detached. Note that this construction can be combined with the previous constructions as well, to create a version of a shape self-replicator.

\update{
\vspace{-10pt}
    \subsection{Possible enhancements to the hierarchical construction}
    \vspace{-5pt}
    
    There are many ways in which this construction could be easily modified to further optimize tile complexity and other parameters. For example, to shrink the length of the {\tt genome}, {\tt gene}s could be compressed so that they are no longer required to be as long as the largest dimension of a {\tt block}. Instead, in cases where {\tt interface}s are shorter than {\tt block} side lengths and appropriately positioned, it is possible to shrink the {\tt gene} encoding a block to as small as $\log$-width. This can be done by incorporating counters that also grow out the width of a {\tt block}.
    Additional, even asymptotically optimal, compression could be achieved by instead encoding the shortest program which can output the {\tt gene} necessary to grow a {\tt block} and then a ``fuel efficient'' Turing machine \cite{jSignals} can be simulated with signal tiles which grow from the {\tt genome} until that encoding is output, allowing {\tt block} growth to proceed from there. Note that this option could greatly increase the the fuel consumed. 
    
    As another example, the necessity to scale certain shapes could be removed by only slightly increasing tile complexity, i.e. the size of $U$. For example, by adding a constant number $m$ of tile types to also be candidates for the ends of {\tt interfaces} (along with the White and Black tiles), the number of {\tt interface}s of each length (which is the limiting number potentially requiring scaling of a shape) can be increased by a factor on the order of $m^2$. There are many other such variations that can be used to balance several factors of the construction to optimize trade-offs for desired goals. Also, for many variations on the specific algorithm which is used to determine the encoding of $S$ into the {\tt genome}, no changes are even required to $U$, so the algorithm can be modified to favor particular tradeoffs over others (e.g. scale factor over {\tt genome} length) without any other modifications to the system.
    
    Finally, it is easy to combine this construction with the previous constructions. For instance, tile types could be added to $U$ from the construction in Section \ref{sec:simple-replicator} that also create duplicate copies of $\sigma_S$. Additionally, an actual self-replicating system could be built by including the shape-deconstruction capabilities of the construction in Section \ref{sec:deconstruct}. Let $M$ be a Turing machine that performs the following computation. Given an input string consisting of the turns of a path through $\mathbb{Z}^3$ (i.e. the path encoded in a seed assembly genome of the construction in Section \ref{sec:simple-replicator}), it first computes the points of the shape $S$ generated by that path. It then performs the computations for the hierarchical replicator of this section to compute a valid input {\tt genome} for it. Simulation of an arbitrary Turing machine is straightforward even with static aTAM tiles (e.g. \cite{jSADS,jCCSA,SolWin07}) and can additionally be made ``fuel efficient'' using signal tiles \cite{jSignals}. Therefore, there exists a system which can take as input an assembly as for the construction of Section \ref{sec:deconstruct} and use the components of that construction to deconstruct it into a linear genome. Tiles which simulate $M$ then perform the generation of the input {\tt genome} for the hierarchical replicator, which proceeds to make copies of assemblies of shape $S$. This is a more complicated self-replicator which consumes much more fuel (i.e. the TM computation tiles - but note that using techniques of \cite{jSignals} that amount is greatly reduced, and the junk assemblies can all be guaranteed to be of small, constant size) but after the {\tt genome} is computed once it is infinitely replicated along with copies of the shape.

\vspace{-10pt}
\section{The Requirement for Deconstruction}\label{sec:need-to-deconstruct}
\vspace{-5pt}

    \begin{definition}
    Given a tile set $T$, a \emph{porous assembly} $\alpha$, over tiles in $T$, is one in which it is possible for unbound tiles of one or more types in $T$ to pass freely through either (1) the body of one or more tiles in $\alpha$, or (2) the gaps between tiles in $\alpha$ (which means between bound glues if the tiles are bound to each other), or (3) a combination of both. Conversely, a \emph{non-porous assembly} is one in which no unbound tiles can pass through any of the tile bodies or gaps between tiles.
    \end{definition}
    
    For theoretical results, we tend to consider all tile bodies to be solid, or at least solid enough to prevent the diffusion of other tiles through them. Whether or not an assembly is porous then depends upon factors such as the spacing between tiles, lengths of glues, and spacing of glues. For instance, the seed assemblies for the construction in Section \ref{sec:deconstruct} are non-porous assuming glues are spread evenly along the edges of tiles.
    
    In this section we prove that in the STAM* there cannot be a universal shape self-replicator in systems with non-porous assemblies that does not use (an arbitrary amount of) deconstruction.
    
    \begin{theorem}\label{thm:need-to-deconstruct}
    Let $U$ be an STAM* tile set such that for an arbitrary 3D shape $S$, the STAM* system $\mathcal{T} = (U,\sigma_S,\tau)$ with $\dom \sigma = S$, $\mathcal{T}$ is a shape self-replicator for $S$ and $\sigma$ is non-porous. Then, for any $r \in \mathbb{N}$, there exists a shape $S$ such that $\mathcal{T}$ must remove at least $r$ tiles from the seed assembly $\sigma_S$.
    \end{theorem}
    
    \begin{proof}
    We prove Theorem \ref{thm:need-to-deconstruct} by contradiction. Therefore, assume that $U$ is a tile set in the STAM* capable of shape replicating any shape $S$ and that seed assembly $\sigma_S$ is non-porous. Let $t = |U|$, $g$ be the maximum number of glues on any tile type in $U$, and $s$ be the maximum number of signals on any tile type in $U$. Note that for any position in an assembly over tiles in $U$, there is a maximum number of $\lambda = t(3^g)(3^s)$ possible tile types and tile states (accounting for all possible states of glues and signals).
    
    We define a shape $c$ which is an $n \times n \times n$ cube, for some $n \in \mathbb{N}$ to be defined, with every point on the exterior of the cube included in the shape. For every $xy$ plane (i.e. horizontal plane) in the interior of the cube, the points contained within $c$ follow the pattern shown in Figure \ref{fig:imposs-center}, where the grey locations are all included and a subset of the green locations are included. Note that only one plane has a connection to the exterior, and no other tiles of any plane in the interior are adjacent to a location of the exterior. Define the set $C$ as the set of all such $c$ where there is one for each possible pattern of green locations included and excluded.
    
    \begin{figure}
        \centering
        \includegraphics[width=0.15\textwidth]{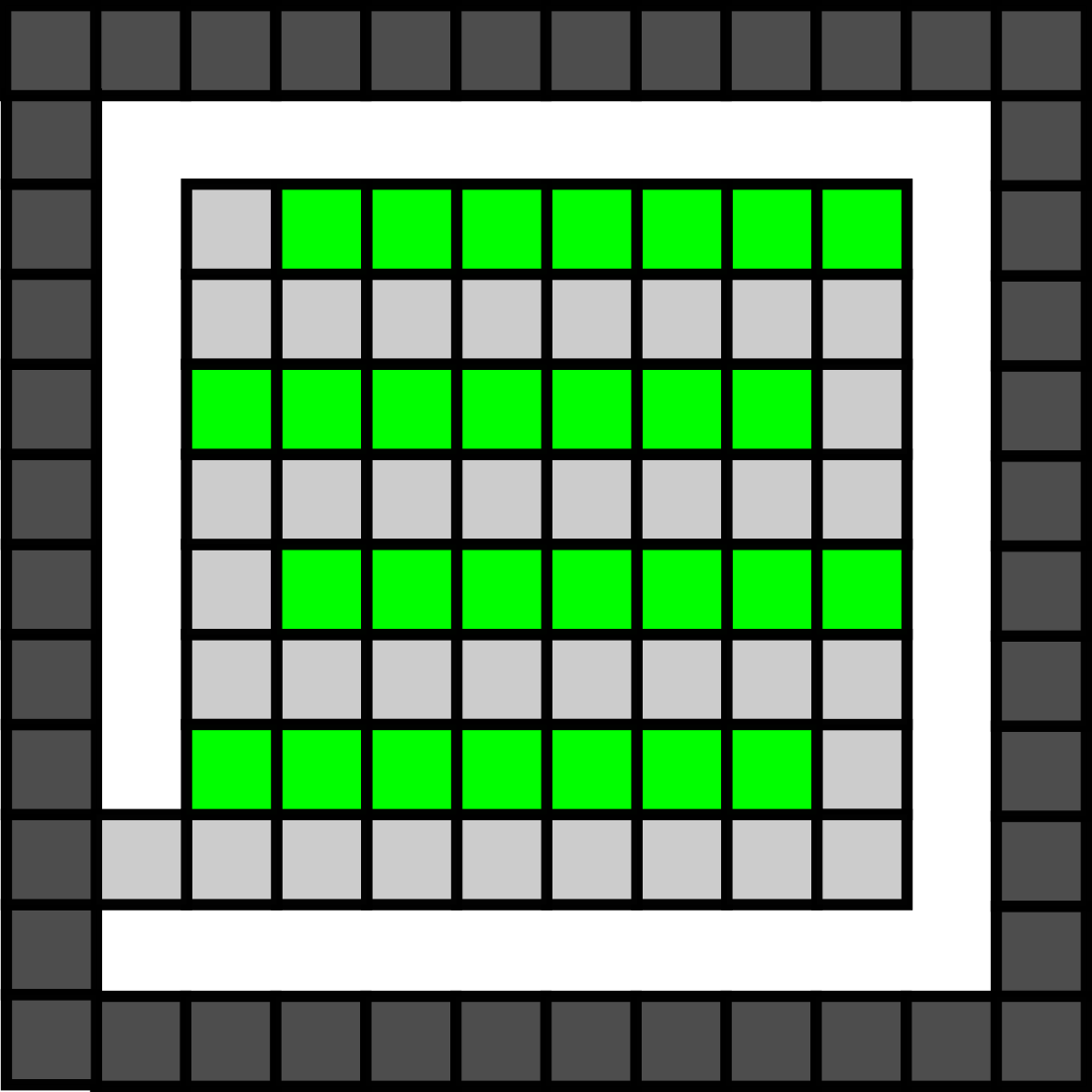}
        \caption{An example interior $xy$ plane within the cube $c$ of the proof of Theorem \ref{thm:need-to-deconstruct}. The plane in this example has the single connection to the exterior of the cube (dark grey), and all light grey locations are included, along with a subset of the green locations.}
        \label{fig:imposs-center}
        \vspace{-10pt}
    \end{figure}
    In order to ensure that only a single location of a single $xy$ plane in the interior of the cube is adjacent to the exterior (i.e. to leave a gap all around) the number of $xy$ planes with occupied locations is $n-4$. The width of each green row is $n-5$. The number of green rows in each $xy$ plane is $(n-4)/2$. Therefore, the number of green interior positions is $(n-4)(n-5)(n-4)/2$. The number of shapes which include every possible subset of those green positions is $2^{(n-4)(n-5)(n-4)/2}$, and this is the size of the set $C$.
    
    Conversely, the number of unit cube locations on the exterior of each $n \times n \times n$ cube is $6(n-1)^2$.
    
    By our assumption, for every $c \in C$, there exists an STAM* system $\mathcal{T}_c = (U,\sigma_c,\tau)$ such that $\mathcal{T}_c$ shape self-replicates $c$. However, for each such $\sigma_c$, the total number of options for a tile in each exterior location (including states) is $\lambda$, and therefore the total number of unique subassemblies composing the exterior surfaces of the cube is $\lambda^{6(n-1)^2}$. Also, since $s$ is the maximum number of signals on any tile type in $T$, $s!$ represents every possible ordering of completion of signals on the tile with the most signals. We can choose a value of $n$ (for the side lengths of the cubes) such that $(s!)\lambda^{6(n-1)^2+1} < 2^{(n-4)(n-5)(n-4)/2}$, since the exponents of the left and right sides grow on the order of $n^2$ and $n^3$, respectively, and all other terms are constants with respect to $n$. Let $n$ be such a sufficiently large value and then note that by the pigeonhole principle, for two $c_1,c_2 \in C$, the systems $\mathcal{T}_{c_1}$ and $\mathcal{T}_{c_2}$ must have identical subassemblies composing the exteriors of their seed assemblies as well as the single tile attaching each exterior to the interior planes. Additionally, there must be an assembly sequence such that the single tile of each exterior subassembly that is connected to the interior planes must experience the same ordering of completion of signals (since anything that could happen on their exteriors must be the same for both, and there were enough assemblies with the same subassemblies to guarantee the same order of completion of their signals for at least two of them). Since $\sigma_{c_1}$ and $\sigma_{c_2}$ are non-porous, there can be no other factors in $\mathcal{T}_{c_1}$ and $\mathcal{T}_{c_2}$ which influence the growth of assemblies, and so both systems must be able to yield the same terminal assemblies. This contradicts that they shape self-replicate $c_1$ and $c_2$ since these are different shapes. Finally, in order to achieve the arbitrary bound $r$ for required tile removals, we can simply adapt our target shape to be a ``chain'' of $r$ cubes (all of which can be made to be unique) connected by a single-tile-wide path of tiles and otherwise completely separated. The previous argument holds for each of the $r$ cubes, and since none can be replicated without the removal of at least one tile, a lower bound of the removal of at least $r$ tiles is established.
    
    \end{proof}
}


\bibliographystyle{spmpsci}
\bibliography{tam,experimental_refs}




\end{document}